\documentclass{article}
\usepackage[a4paper, portrait, margin=1.1811in]{geometry}
\usepackage[english]{babel}
\usepackage[utf8]{inputenc}
\usepackage[T1]{fontenc}
\usepackage{helvet}
\usepackage{etoolbox}
\usepackage{graphicx}
\usepackage{titlesec}
\usepackage{caption}
\usepackage{booktabs}
\usepackage{xcolor} 
\usepackage{multirow}
\usepackage{makecell}
\usepackage{longtable}
\usepackage[colorlinks, citecolor=cyan]{hyperref}
\usepackage{caption}
\graphicspath{ {./images/} }
\usepackage{scrextend}
\usepackage{fancyhdr}
\usepackage{graphicx}
\usepackage{amsmath}
\newcounter{lemma}

\newcounter{theorem}

\usepackage{amsthm}

\usepackage{subfig}
\fancypagestyle{plain}{
	\fancyhf{}

	}
\makeatletter
\patchcmd{\@maketitle}{\LARGE \@title}{\fontsize{16}{19.2}\selectfont\@title}{}{}
\makeatother
\usepackage{authblk}

\newcommand{\keywords}[1]{\textbf{Keywords}:#1}
\newsavebox\affbox
\author[1]{\textbf{Mingzhe Yang}}
\author[1]{\textbf{Zhipeng Wang}}
\author[1]{\textbf{Kaiwei Liu}}
\author[3]{\textbf{Yingqi Rong}}
\author[2]{\textbf{Bing Yuan}}
\author[1,2*]{\textbf{Jiang Zhang}}

\affil[1]{ School of Systems Science, Beijing Normal University, 100875, Beijing, China}
\affil[2]{Swarma Research, 102300, Beijing, China
}
\affil[3]{ Johns Hopkins University, 21218, Baltimore, MD, USA
}

\titlespacing\section{0pt}{12pt plus 4pt minus 2pt}{0pt plus 2pt minus 2pt}
\titlespacing\subsection{12pt}{12pt plus 4pt minus 2pt}{0pt plus 2pt minus 2pt}
\titlespacing\subsubsection{12pt}{12pt plus 4pt minus 2pt}{0pt plus 2pt minus 2pt}

\titleformat{\section}{\normalfont\fontsize{10}{15}\bfseries}{\thesection.}{1em}{}
\titleformat{\subsection}{\normalfont\fontsize{10}{15}\bfseries}{\thesubsection.}{1em}{}
\titleformat{\subsubsection}{\normalfont\fontsize{10}{15}\bfseries}{\thesubsubsection.}{1em}{}

\titleformat{\author}{\normalfont\fontsize{10}{15}\bfseries}{\thesection}{1em}{}

\title{\textbf{\huge Finding emergence in data by maximizing effective information}\\
 }
\date{}    

\begin{document}

\pagestyle{headings}	
\newpage
\setcounter{page}{1}
\renewcommand{\thepage}{\arabic{page}}

\captionsetup[figure]{labelfont={sf},labelformat={default},labelsep=period,name={Figure }}	\captionsetup[table]{labelfont={bf},labelformat={default},labelsep=period,name={Table }}
\setlength{\parskip}{0.5em}
	
\maketitle

\abstract{Quantifying emergence and modeling emergent dynamics in a data-driven manner for complex dynamical systems is challenging due to the lack of direct observations at the micro-level. Thus, it's crucial to develop a framework to identify emergent phenomena and capture emergent dynamics at the macro-level using available data. Inspired by the theory of causal emergence (CE), this paper introduces a machine learning framework to learn macro-dynamics in an emergent latent space and quantify the degree of CE. The framework maximizes effective information, resulting in a macro-dynamics model with enhanced causal effects. Experimental results on simulated and real data demonstrate the effectiveness of the proposed framework. It quantifies degrees of CE effectively under various conditions and reveals distinct influences of different noise types. It can learn a one-dimensional coarse-grained macro-state from fMRI data, to represent complex neural activities during movie clip viewing. Furthermore, improved generalization to different test environments is observed across all simulation data.}


\keywords{causal emergence; dynamics learning; effective information; coarse-graining; invertible neural network}

\section{Introduction}
The climate system, ecosystems, bird flocks, ant colonies, cells, brains, and many other complex systems are composed of numerous interacting elements and exhibit a wide range of nonlinear dynamical behaviors~\cite{sayama2015introduction,odell2002agents}. In the past few decades, the research topic of data-driven modeling complex systems has gained significant attention, driven by the increasing availability and accumulation of data from real dynamical systems~\cite{Wang_Lai_Grebogi_2016,kipf2018neural,Casadiego_Nitzan_Hallerberg_Timme_2017,zhang2022universal}. However, complex systems always exhibit emergent behaviors~\cite{holland2000emergence,sayama2015introduction}. That means some interesting emergent patterns or dynamical behaviors such as waves~\cite{ben2023turing}, periodic oscillations~\cite{Matthews_Strogatz_2002} and solitons~\cite{du2023emergent} can hardly be directly observed and identified from the micro-level behavioral data. Therefore, the identification and measure of emergence and the capture of emergent dynamical patterns solely from observational raw data have become crucial challenges in complex systems research~\cite{Rosas_Mediano_Jensen_Seth_Barrett_Carhart-Harris_Bor_2020,zhang2022neural,O’toole_Nallur_Clarke_2017,kemeth2022learning}. But in order to address these problems, it is necessary to first develop a quantitative understanding of emergence.

Emergence, as a distinctive feature of complex systems~\cite{holland2000emergence}, has historically been challenging to quantify and describe in quantitative terms~\cite{Keijzer_2008,Shalizi_Moore_2003,seth2008measuring,haugen2023detecting}. Most conventional measures or methods either rely on predefined macro-variables (e.g., \cite{fisch2010quantitative,mnif2011quantitative,fisch2011divergence}) or are tailored to specific scenarios in engineered systems (e.g. \cite{seth2008measuring,
osmundson2008kr14,raman2022framework,teo2013formalization}). However, there is a need for a unified method to quantify emergence across different contexts. The theory of causal emergence (CE), introduced by Erik Hoel in 2013, offers a framework to tackle this challenge~\cite{hoel2013quantifying,hoel2017map}. Hoel's theory aims to understand emergence through the lens of causality. The connection between emergence and causality is implied in the descriptive definition of emergence, as stated in the work by J. Fromm\cite{Fromm_2005}. According to this definition, a macro-level property, such as patterns or dynamical behaviors, is considered emergent if it cannot be explained or directly attributed to the individuals in the system. The theory of causal emergence formalizes this concept within the framework of discrete Markov dynamical systems. As shown in Figure \ref{fig:causalemergence}(a), Hoel et al's theory states that if a system exhibits stronger causal effects after a specific coarse-graining transformation compared to the original system, then CE has occurred. In Hoel's framework, the degree of causal effect in a dynamical system is quantified using effective information (EI) \cite{tononi2003measuring}. EI can be understood as an intervention-based version of mutual information between two successive states in a dynamical system over time. It is a measure that solely depends on the system's dynamics. If a dynamical system is more deterministic and non-degenerate, meaning that the temporal adjacent states can be inferred from each other in both directions of the time arrow, then it will have a larger EI \cite{hoel2013quantifying}. This measure has been shown to be compatible with other well-known measures of causal effect~\cite{comolatti2022causal}.  Figure \ref{fig:causalemergence}(b) gives an example of CE for simple Markov chain.

\begin{figure}
    \centering
    \includegraphics[width=\textwidth]{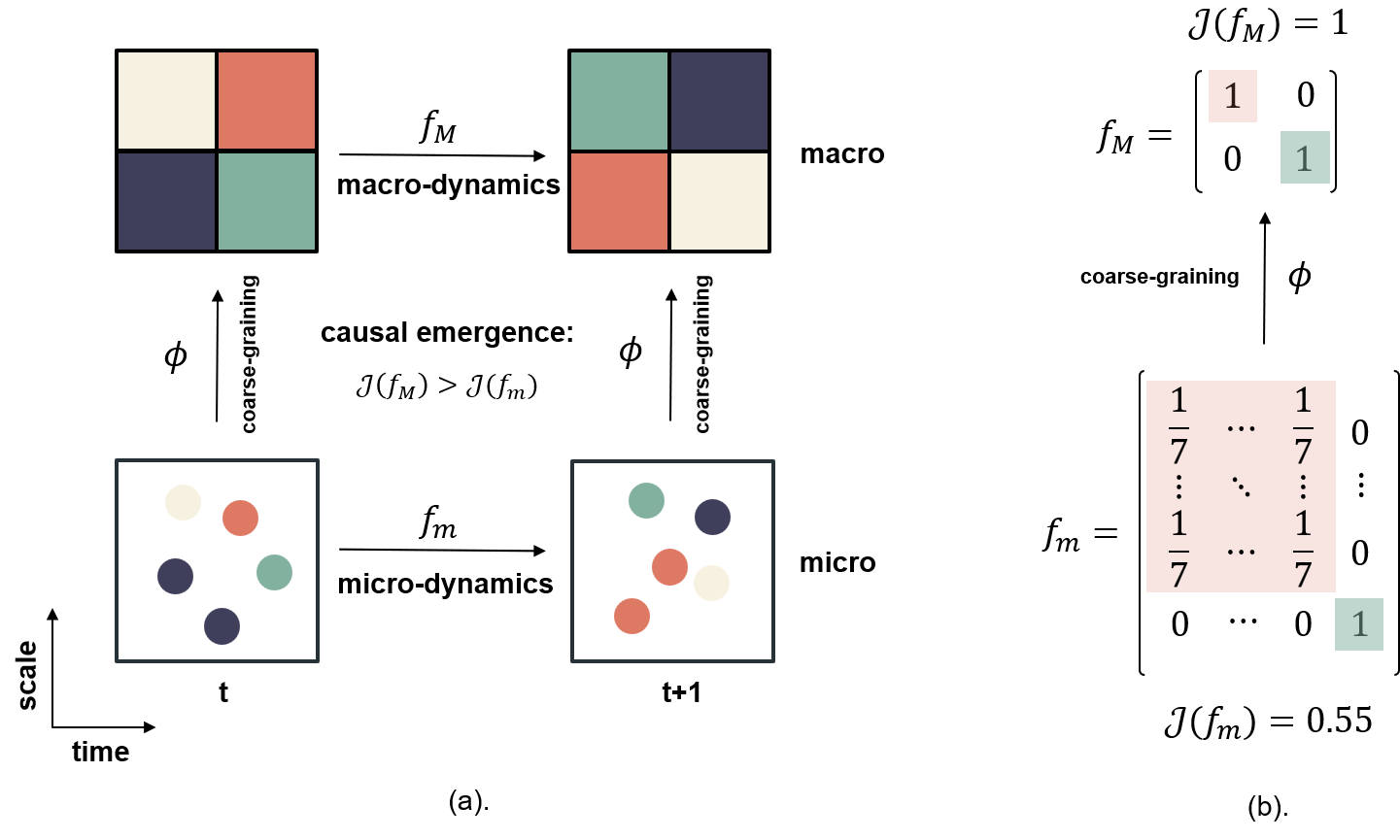}
    \caption{(a) An illustration of the fundamental concept of Erik Hoel's theory of causal emergence(CE). The effective information (EI) is denoted as $\mathcal{J}$ in this paper. (b) A case demonstrating CE in a discrete Markov chain. The micro-dynamics consist of eight micro-states. During the coarse-graining process, the first seven states are grouped together as one macro-state, while the eighth micro-state corresponds to the second macro-state. As a result, a transition probability matrix is formed at the macro-scale, where the effective information $\mathcal{J}(f_M)=1$ (calculated using Equation \ref{eq:definition_EI}), which is greater than $\mathcal{J}(f_m)=0.55$. This difference, $\Delta\mathcal{J}=0.45$, indicates the occurrence of CE, as $\Delta\mathcal{J}>0$. }
    \label{fig:causalemergence}
\end{figure}

While the causal emergence theory has successfully quantified emergence using EI and has found applications in various fields~\cite{Griebenow_Klein_Hoel_2019,Hoel_Levin_2020,Marrow_Michaud_Hoel_2020}, there are some drawbacks. Firstly, the markov transition matrix of the micro-dynamics should be given rather than constructed from data. Secondly, a predefined coarse-graining strategy must be provided or optimized through maximizing EI, but this optimization process is computationally complex\cite{klein2020emergence, Griebenow_Klein_Hoel_2019}. Although Rosas et al.\cite{Rosas_Mediano_Jensen_Seth_Barrett_Carhart-Harris_Bor_2020} proposed a new framework for CE based on partial information decomposition theory~\cite{williams2010nonnegative}, which does not require a predefined coarse-graining strategy, it still involves iterating through all variable combinations on the information lattice to compare synergistic information, resulting in significant computational complexity. Rosas et al. also proposed an approximate method to mitigate the complexity~\cite{Rosas_Mediano_Jensen_Seth_Barrett_Carhart-Harris_Bor_2020}, but it requires a pre-defined macro-variable.

Therefore, the challenge of finding emergence in data, which involves determining whether CE has occurred within a system and to what extent based solely on observational data of its behavior, remains unresolved. The most daunting task is that all the elements, including the markov dynamics at both the micro- and macro-levels, as well as the coarse-graining strategy to obtain macro-variables, need to be learned from raw data and cannot be predefined in advance~\cite{zhang2022neural}. Once the learning process is completed, we can compare the strength of causal effects (measured by EI) in dynamics at different scales to identify CE from the data. Therefore, the problem of finding CE within Hoel's theoretical framework is essentially equivalent to the challenge of data-driven modeling within a coarse-grained space for complex systems~\cite{zhang2022neural}. Building models for complex systems at multiple coarse-grained levels within learned emergent spaces is of utmost importance for both identifying CE and conducting data-driven modeling in complex systems.


Recently, several machine learning frameworks have emerged for learning and simulating the dynamics of complex systems within coarse-grained latent or hidden spaces~\cite{vlachas2022multiscale,kemeth2022learning,floryan2022data,Cai_Ji_2020,Chen_Li_Yang_Li_Liu_2022}. By learning the dynamics in a hidden space, these frameworks enable the removal of noise and irrelevant factors from raw data, enhancing prediction capabilities and adaptability to diverse environments. While these learning systems can capture emergent dynamics, they may not directly address the fundamental nature of CE, which entails stronger causality. According to Judea Pearl's hierarchy of causality, prediction-based learning is situated at the level of association and cannot address the challenges related to intervention and counterfactuals~\cite{Goldberg_2019}. Empirically, dynamics learned solely based on predictions may be influenced by the distributions of the input data, which can be limited by data diversity and the problem of over fitting models~\cite{Shen_Liu_He_Zhang_Xu_Yu_Cui_2021}. However, what we truly desire is an invariant causal mechanism or dynamics that are independent of the input data. This allows the learned mechanism or dynamics to be adaptable to broader domains, generalizable to areas beyond the distribution of training data, and capable of accommodating diverse interventions~\cite{scholkopf2022statistical,Janzing_Peters_Sgouritsa_Zhang_Mooij_lkopf_2012,peters2016causal}. Unfortunately, few studies have explored the integration of causality and latent space dynamics to address the challenges of data-driven modeling in complex systems~\cite{zhang2020invariant}.

Inspired by the theory of causal emergence, this paper aims to address the challenge of learning causal mechanisms within a learned coarse-grained macro-level(latent) space. The approach involves maximizing the EI of the emergent macro-level dynamics, which is equivalent to maximizing the degree of causal effect in the learned coarse-grained dynamics~\cite{comolatti2022causal}. By artificial intervening the distributions of input data as uniform distribution, this maximization process helps separating the invariant causal mechanisms from the variant data distribution as much as possible. To achieve this, a novel machine learning framework called Neural Information Squeezer Plus (NIS+) is proposed. NIS+ extends the previous framework (NIS) to solve the problem of maximizing EI under coarse-grained representations. As shown in Figure \ref{fig:nis+}, NIS+ not only learns emergent macro-dynamics and coarse-grained strategies but also quantifies the degree of CE from time series data and captures emergent patterns. Mathematical theorems ensure the flexibility of our framework in different application scenarios, such as cellular automata and multi-agent systems. Numerical experiments demonstrate the effectiveness of NIS+ in capturing emergent behaviors and finding CE in various simulations, including SIR dynamics\cite{kermack1927contribution}, bird flocking (Boids model)\cite{attanasi2014information}, and the Game of Life\cite{gardner1970fantastic}. Additionally, NIS+ is applied to find emergent features in real brain data from 830 subjects while watching the same movie. Experiments also demonstrate that the learned dynamical model exhibits improved performance in generalization compared to other models.

\begin{figure}[!ht]
	\centering
	\includegraphics[width=\textwidth]{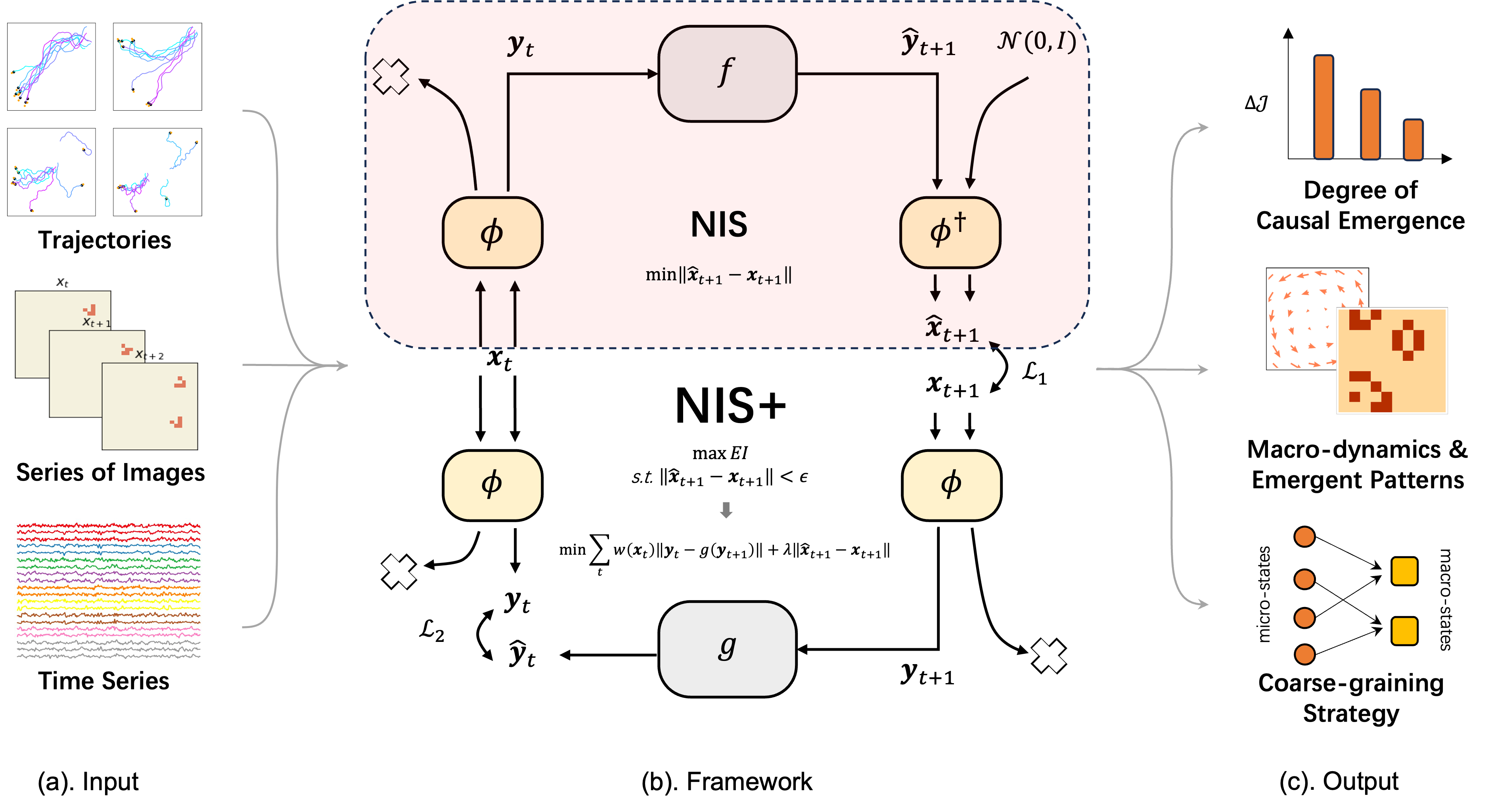}
	\caption{The workflow and architecture of our proposed framework, Neural Information Squeezer Plus (NIS+), for discovering causal emergence within data. (a) Various forms of input data from our studied simulation systems such as the Boid flocking model (multi-agent system), Conway's Game of Life (two-dimensional cellular automata), and real brain fMRI time series data. (b) The framework of our proposed model, NIS+, which incorporates our previous model, NIS. The boxes represent functions or neural networks, and the arrow pointing to a cross represents the operation of information discarding. $\boldsymbol{x}_t$ and $\boldsymbol{x}_{t+1}$ represent the observational data of micro-states, while $\hat{\boldsymbol{x}}_{t+1}$ represents the predicted micro-state. $\boldsymbol{y}_t=\phi(\boldsymbol{x}_{t})$ and $\boldsymbol{y}_{t+1}=\phi(\boldsymbol{x}_{t+1})$ represent the macro-states obtained by encoding the micro-states using the encoder. $\hat{\boldsymbol{y}}_t=\phi(\hat{\boldsymbol{x}}_{t})$ and $\hat{\boldsymbol{y}}_{t+1}=\phi(\hat{\boldsymbol{x}}_{t+1})$ represent the predicted macro-states obtained by encoding the predictions of micro-states. The mathematical problems that each framework aims to solve are also illustrated in the figure. (c) The various output forms of NIS+, which include the degree of CE, the learned macro-dynamics, captured emergent patterns, and the strategy of coarse-graining. }
	\label{fig:nis+}
\end{figure}

\section{Finding causal emergence in data}

Finding CE in time series data involves two sub-problems: \textbf{emergent dynamics learning} and \textbf{causal emergence quantification}. In the first problem, both the dynamics in micro- and macro-level and the coarse-graining strategy should be learned from data. After that, we can measure the degree of CE and judge if emergence occurs in data according to the learned micro- and macro-dynamics.

\subsection{Problem definition}

Suppose the behavioral data of a complex dynamical system is a time series $\{\boldsymbol{x}_t\}$ with time steps $t=1,2,...,T$ and dimension $p$, they form observable micro-states. And we assume that there are no unobserved variables. The problem of \textbf{emergent dynamics learning} is to find three functions according to the data: a coarse-graining strategy $\phi: \mathcal{R}^p\rightarrow \mathcal{R}^q$, where $q\leq p$ is the dimension of macro-states which is given as a hyper-parameter, a corresponding anti-coarsening strategy $\phi^{\dag}:\mathcal{R}^p\rightarrow \mathcal{R}^q$, and a macro-level markov dynamics $f_{q}$, such that the EI of macro-dynamics $f_q$ is maximized under the constraint that the predicted $\hat{\boldsymbol{x}}_{t+1}$ by $\phi$, $f_q$, and $\phi^{\dag}$ is closed to the real data of $\boldsymbol{x}_{t+1}$:

\begin{equation}
\begin{aligned}
\label{old optimization}
&\max_{\phi,f_q,\phi^+} \mathcal{J}(f_q),\\
&s.t. \begin{cases}
|| \hat{\boldsymbol{x}}_{t+1}-\boldsymbol{x}_{t+1} || < \epsilon,\\
\hat{\boldsymbol{x}}_{t+1}=\phi^{\dag}(f_q(\phi(\boldsymbol{x}_t))).
\end{cases}
\end{aligned}
\end{equation}
Where, $\mathcal{J}$ is the appropriate version of EI to be maximized. For continuous dynamical systems that mainly considered in this paper, $\mathcal{J}$ is the dimension averaged effective information (dEI)\cite{zhang2022neural}(see method section \ref{sec:ei for nn}) , $\epsilon$ is a given small constant. 

The objective of maximizing EI in Equation \ref{old optimization} is to enhance the causal effect of the macro-dynamics $f_q$. However, the direct maximization of EI, as depicted in Figure \ref{fig:causalemergence}, can result in trivial solutions, as highlighted by Zhang et al. \cite{zhang2022neural}. Additionally, it can lead to issues of ambiguity and violations of commutativity, as pointed out by Eberhardt, F and Lee, L.L. \cite{eberhardt2022causal}. To address these issues, constraints are introduced to prevent such problems. The constraints in Equation \ref{old optimization} ensure that the macro-dynamics $f_q$ can simulate the micro-dynamics in the data as accurately as possible, with the prediction error being less than a given threshold $\epsilon$. 

This is achieved by mapping the micro-state $x_t$ to the macro-state $\boldsymbol{y}_t=\phi(\boldsymbol{x_t})$ at each time step. The macro-dynamics $f_q$ then makes a one-step prediction to obtain $\hat{\boldsymbol{y}}_{t+1}=f_q(\phi(\boldsymbol{x}_t))$. Subsequently, an anti-coarsening strategy $\phi^{\dag}$ is employed to decode $\hat{\boldsymbol{y}}_{t+1}$ back into the micro-state, resulting in the reconstructed micro-state $\hat{\boldsymbol{x}}_{t+1}=\phi^{\dag}(\hat{\boldsymbol{y}}_{t+1})$. The entire computational process can be visualized by referring to the upper part(NIS) in Figure \ref{fig:nis+}b.

In practice, we obtain the value of $\epsilon$ by setting the normalized MAE (mean absolute error divided by the standard deviation of $\boldsymbol{x}$) on the training data. The choice of normalized MAE ensures a consistent standard across different experiments, accounting for the varying numerical ranges. 


By changing $q$ we can obtain macro-dynamics in various dimensions. If $q=p$, then $f_p$ becomes the learnt micro-dynamics. Then we can compare $\mathcal{J}_q$ and $\mathcal{J}_p$ for any $q$. The problem of \textbf{causal emergence quantification} can be defined as the calculation of the following difference,

\begin{equation}
\label{eq:causal_emergence_identification}
\Delta\mathcal{J}\equiv \mathcal{J}(f_q)-\mathcal{J}(f_p),
\end{equation}
where $\Delta\mathcal{J}$ is defined as the degree of causal emergence. The micro-dynamics $f_p$ is also obtained by solving Equation \ref{old optimization} with dimension $q=p$. If $\Delta\mathcal{J}>0$, then we say CE occurs within the data. 

\subsection{Solution}
While, solving the optimization problem defined in Equation \ref{old optimization} directly is difficult because the objective function $\mathcal{J}$ is the mutual information after intervention which deserved special process.

To tackle this problem, we convert the problem defined in Equation \ref{old optimization} into a new optimization problem without constraints, that is:
\begin{gather}
\label{new_optimization}
\min_{f,g,\phi,\phi^{\dag}} \sum_{t=1}^{T-1} w(\boldsymbol{x}_t)||\boldsymbol{y}_t-g(\boldsymbol{y}_{t+1})||+\lambda||\hat{\boldsymbol{x}}_{t+1}-\boldsymbol{x}_{t+1} ||,
\end{gather}
where $\hat{\boldsymbol{x}}_{t+1}=\phi^{\dag}(f(\phi(\boldsymbol{x}_t)))$. $\boldsymbol{y}_t=\phi(\boldsymbol{x}_t)$ and $\boldsymbol{y}_{t+1}=\phi(\boldsymbol{x}_{t+1})$ are the macro-states. $g: \mathcal{R}^q\rightarrow\mathcal{R}^q$ is a new function that we introduce to simulate the inverse macro-dynamics on the macro-state space, that is to map each macro-state at $t+1$ time step back to the macro-state at $t$ time step. $\lambda$ is a Lagrangian multiplier which will be taken as a hyper-parameter in experiments. $w(\boldsymbol{x}_t)$ is the inverse probability weights which is defined as:

\begin{equation}
    \label{eq:probreweight}
    w(\boldsymbol{x}_t)=\frac{\Tilde{p}(\boldsymbol{y}_t)}{p(\boldsymbol{y}_t)}=\frac{\Tilde{p}(\phi(\boldsymbol{x}_t))}{p(\phi(\boldsymbol{x}_t))},
\end{equation}
where $\Tilde{p}$ is the new distribution of macro-states $\boldsymbol{y}_t$ after intervention for $do(\boldsymbol{y}_t\sim U_q)$, and $p$ is the natural distribution of the data. In practice, $p(\boldsymbol{y}_t)$ is estimated through kernel density estimation (KDE)\cite{rosenblatt1956remarks}(see method section \ref{sec:KDE}). The approximated distribution, $\Tilde{p}(\boldsymbol{y}_t)$, is assumed to be a uniform distribution, which is characterized by a constant value. Consequently, the weight $w$ is computed as the ratio of these two distributions. Mathematical theorems mentioned in method section~\ref{sec:nis+} and proved in support information section~\ref{sec:optimize} guarantee this new optimization problem (Equation \ref{new_optimization}) is equivalent to the original one (Equation \ref{old optimization}).

\section{Results}
We will validate the effectiveness of the NIS+ framework through numerical experiments where data is generated by different artificial models (dynamical systems, multi-agent systems, and cellular automata). Additionally, we will apply NIS+ to real fMRI data from human subjects to uncover interesting macro-level variables and dynamics. In these experiments, we will evaluate the models' prediction and generalization abilities. We will also assess their capability to identify CE and compare it with Rosas' $\Psi$ indicator proposed in \cite{Rosas_Mediano_Jensen_Seth_Barrett_Carhart-Harris_Bor_2020}, an alternative measure for quantifying CE approximately. For the experiments, a threshold of normalized MAE of 0.3 was selected based on the observation that exceeding this value leads to a significant change in the relationship between normalized MAE and noise (see below).

\subsection{SIR}
The first experiment revolves around the SIR (Susceptible, Infected, and Recovered or Died) model, a simple dynamical system. In this experiment, the SIR dynamics serve as the ground truth for the macro-level dynamics, while the micro-level variables are generated by introducing noise to the macro-variables. The primary objective is to evaluate our model's ability to effectively remove noise, uncover meaningful macroscopic dynamics, identify CE, and demonstrate generalization beyond the distribution of the training dataset.

Formally, the macro-dynamics can be described as:
\begin{eqnarray}
\begin{aligned}
\begin{cases}  
\frac{\mathrm{d}S}{\mathrm{d}t}=-\beta SI,  \\
\frac{\mathrm{d}I}{\mathrm{d}t}=\beta SI - \gamma I, \\
\frac{\mathrm{d}R}{\mathrm{d}t}= \gamma I,
\end{cases}
\end{aligned}
\end{eqnarray}
where $S,I,R\in[0,1]$ represents the proportions of healthy, infected and recovered or died individuals in a population, $\beta=1$ and $\gamma=0.5$ are parameters for infection and recovery rates, respectively. Figure \ref{fig:sir}(a) shows the phase space $(S,I,R)$ of the SIR dynamics. Because the model has only two degrees of freedom, as $S$, $I$, and $R$ satisfy $S+I+R=1$, all macro-states are distributed on a triangular plane in three dimensions, and only $S$ and $I$ are used to form the macro-state variable $\boldsymbol{y}=(S,I)$, and $R$ is removed.  

We then expand $\boldsymbol{y}$ into a four-dimensional vector and introduce Gaussian noises to form a microscopic state:
\begin{eqnarray}
\begin{aligned}
\label{eq:sir noise}
\begin{cases}
    \boldsymbol{S}'=(S,S)+\boldsymbol{\xi}_1,  \\
    \boldsymbol{I}'=(I,I)+\boldsymbol{\xi}_2.
\end{cases}
\end{aligned}
\end{eqnarray}
Among which, $\boldsymbol{\xi}_1,\boldsymbol{\xi}_2 \sim \scriptsize{N}(0,\Sigma)$ are two-dimensional Gaussian noises and independent each other, and $\Sigma$ is the correlation matrix. In this way, we obtain a micro-states sequence $\boldsymbol{x}_t = (\boldsymbol{S}'_t,\boldsymbol{I}'_t)$ as the training samples in the experiment. We randomly select initial conditions by sampling points within the triangular region depicted in Figure \ref{fig:sir}(a) and generate time series data using the aforementioned process. These generated data are then utilized to train the models. 

\begin{figure}[htbp]    
  \centering
    \includegraphics[width=1\textwidth]{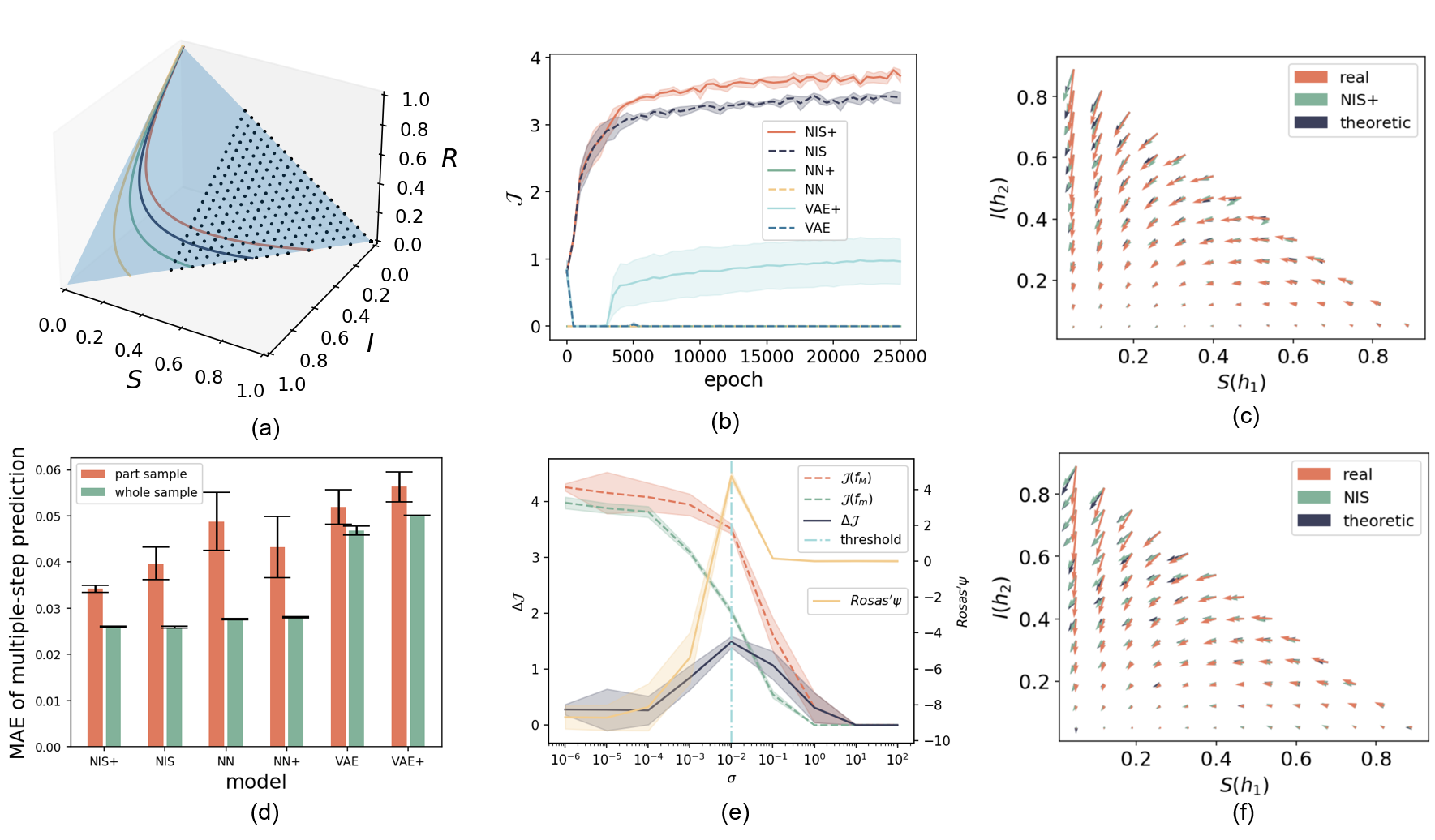}

      
  \caption{The experimental results of NIS+ and compared models on the SIR model with observational noise. (a) The phase space of the SIR model, along with four example trajectories with the same infection and recovery or death rates. The full dataset (blue area) and the partial dataset (dotted area) used for training are also displayed, consisting of 63,000 and 42,000 uniformly distributed data points, respectively. (b) The curves depict the change in dimension-averaged effective information ($\mathcal{J}$) with training epochs for different models.  The lines represent the means, while the band widths represent the standard deviations of five repeated experiments. (c) A comparison is made among the vector fields of the SIR dynamics, the learned macro-dynamics of NIS+, and the macro-dynamics transformed by the Jacobian of the learned encoder. Each arrow represents a direction, and the magnitude of the derivative of the dynamics at that coordinate point. 
  For detailed procedures, please refer to the support information section~\ref{sec:sir vector field}. (d) A comparison is conducted to evaluate the errors in multi-step predictions for different models trained on either partial datasets (with 42,000 missing data points) or complete datasets(see details in the support information section~\ref{sec:sir data detail}). These models include NIS+, NIS, a feed-forward neural network (NN), a feed-forward neural network with inverse probability weighting and inverse dynamics learning techniques (NN+), a Variational Autoencoder(VAE), and its reweighted and inverse dynamics version (VAE+). Please refer to the details of the parameters in method section~\ref{sec:nn vae}. (e). The variations in the measure of CE ($\Delta\mathcal{J}$) and EIs for micro-dynamics ($\mathcal{J}(f_m)$) and macro-dynamics  ($\mathcal{J}(f_M)$) are plotted as the standard deviation $\sigma$ of observation noise changes. All these indicators are averaged across dimensions. 
  Following Rosas' definition and calculation method for CE(see method section \ref{sec:Rosas psi}), the yellow line demonstrates the changes in Rosas' $\Psi$ \cite{Rosas_Mediano_Jensen_Seth_Barrett_Carhart-Harris_Bor_2020}. The vertical line represents the threshold for the normalized MAE equaling 0.3. When $\sigma$ is larger than the threshold, the constraint of error in Equation \ref{old optimization} is violated, and the results are not reliable. (f) A comparison is made among the vector fields of the SIR dynamics, the learned macro-dynamics of NIS, and the macro-dynamics transformed by the encoder Jacobian matrix of NIS, in comparison with (c).}
  \label{fig:sir}            
\end{figure}
%

We conduct comparative analysis between NIS+ and other alternative models, including NIS (without EI maximization compared to NIS+), feed-forward neural network (NN), and variational autoencoder (VAE). To make a fair comparison, we ensure that all benchmark models have a roughly equal number of parameters. Moreover, we employ the same techniques of probability reweighting and inverse dynamics learning on the feed-forward neural network (NN+) and variational autoencoder (VAE+) as utilized in NIS+. We evaluate the performances of all candidate models by requiring them to predict future states for multiple time steps (10 steps) on a separate test dataset. The results show that NIS+(green bars) outperforms other competitors on multi-step prediction as shown in Figure \ref{fig:sir}(d) no matter if they use the techniques like probability reweighting to maximize EI, this manifests that an invertible neural network as an encoder and decoder is necessary (the details can be referred to method section~\ref{sec:nn vae}). 


Further, to assess the model's generalization ability beyond the region of the training dataset, in addition to the regular training and testing, we also conduct experiments where the model was trained on a subset of the data and tested on the complete dataset. The training samples in this experiment are shown in the dotted area in Figure \ref{fig:sir}(a)(the area with $S \le \frac{1}{3}$ is missing), and the test samples are also in the blue triangle. As shown by the red bars in Figure \ref{fig:sir}(d), the performances of out-of-distribution generalization of NIS+ are better than other benchmarks although the test region is beyond the trained region. And the differences among different models are larger on partial dataset. 

To further test whether the models successfully learn the ground-truth macro-dynamics, we conduct a comparison between the vector fields of the real SIR dynamics, represented by $d\boldsymbol{y}/dt$, and the learned emergent dynamics $d(h_1,h_2)/dt$. This comparison is illustrated in Figure \ref{fig:sir}(c) for NIS+ and (f) for NIS. In both sub-figures, the learned vectors align with the ground-truth (real) dynamics and match the theoretical predictions based on the Jacobian of the encoder (for more details, please refer to support information section~\ref{sec:sir vector field}). However, it is evident that NIS+ outperforms NIS in accurately capturing the underlying dynamics, especially in peripheral areas with limited training samples.

Next, we test NIS+ and other comparison models on EI maximization and CE quantification, the results are shown in Figure \ref{fig:sir}(b) and (e). First, to ensure that EI is maximized by NIS+, Figure \ref{fig:sir}(b) illustrates the evolution of EI (dimension averaged) $\mathcal{J}$ over training epochs. It is evident that the curves of NIS+ (red solid), NIS (black dashed), and VAE+(green solid) exhibit upward trends, but NIS+ demonstrates a faster rate of increase. This indicates that NIS+ can efficiently maximize $\mathcal{J}$ to a greater extent than other models. Notably, NIS also exhibits a natural increase in EI as it strives to minimize prediction errors.

Second, to examine NIS+'s ability to detect and quantify CE, we compute $\Delta\mathcal{J}$s and compare them with Rosas' $\Psi$ indicators as the noise level $\sigma$ in micro-states increased (see support information section \ref{sec:sir train} for details). We utilize the learned macro-states from NIS+ as the prerequisite variable $V$ to implement Rosas' method \cite{Rosas_Mediano_Jensen_Seth_Barrett_Carhart-Harris_Bor_2020}. The results are depicted by the black and yellow solid lines in Figure \ref{fig:sir}(e).

Both indicators exhibit a slight increase with $\sigma$ and $\Delta\mathcal{J}>0$ always holds when it is less than 0.01, but Rosas' $\Psi>0$ after $\sigma=10^{-3}$. Therefore, NIS+ indicates that CE consistently occurs at low noise levels, whereas Rosas' method does not. NIS+'s result is more reasonable since it can extract macro-dynamics similar to the ground-truth from noisy data, and this deterministic dynamics should have a larger EI than the noisy micro-dynamics. We also plot the curves $J(f_M)$ (red dashed line) and $J(f_m)$ (green dashed line) for macro- and micro-dynamics, respectively. These curves decrease as $\sigma$ increase, but $J(f_m)$ decreases at a faster rate, leading to the observed occurrence of CE. However, when Rosas' $\Psi<0$, we cannot make a definitive judgment as $\Psi$ can only provide a sufficient condition for CE. Both indicators reach their peak at $\sigma=10^{-2}$, which corresponds to the magnitude of the time step ($dt=0.01$) used in our simulations and reflects the level of change in micro-states.

On the other hand, if the noise becomes too large, the limited observational data makes it challenging for NIS+ to accurately identify the correct macro-dynamics from the data. Consequently, the degree of CE $\Delta\mathcal{J}$ decreases to zero. Although NIS+ determines that there is no CE when $\sigma>10$, this result is not reliable since the normalized prediction errors have exceeded the selected threshold 0.3 after $\sigma=10^{-2}$ (the vertical dashed and dotted line).

Therefore, these experiments manifest that by maximizing EI and learning an independent causal mechanism, NIS+ can effectively disregard noise within the data and accurately learn the ground truth macro-dynamics, and generalization to unseen data. Additionally, NIS+ demonstrates superior performance in quantifying CE. More details about experimental settings are shown in support information section \ref{sec:sir}. 

\subsection{Boids}
\label{sce:Boids}

The second experiment is on Boids model which is a famous multi-agent model to simulate the collective behaviors of birds \cite{Reynolds1987,Reynolds1999}. In this experiment, we test the ability of NIS+ on capturing emergent collective behaviors and CE quantification on different environments with intrinsic and extrinsic noises. To increase the explainability of the trained coarse-graining strategy, we also try to give an explicit correspondence between the learned macro-states and the micro-states.



We conducted simulations following the methodology of \cite{Reynolds1987} with $N=16$ boids on a $300\times300$ canvas to generate training data. The detailed dynamical rules of the Boids model can be found in support information section \ref{sec:Boidsmodel}. In order to assess the ability of NIS+ to discover meaningful macro-states, we separate all boids into two groups, and artificially modified the Boids model by introducing distinct constant turning forces for each group. This modification ensured that the two groups exhibited separate trajectories with different turning angles, as depicted in Figure \ref{fig:boids1}. The micro-states of each boid at each time step consisted of their horizontal and vertical positions, as well as their two-dimensional velocities. The micro-states of all boids formed a $4N$ dimensional vector of real numbers, which was used as input for training NIS+.

As depicted by the triangles in Figure \ref{fig:boids1}, the predicted emergent collective flying behaviors for 50 steps closely follow the ground-truth trajectories of the two groups, particularly at the initial stages. These predicted trajectories are generated by decoding the predicted macro-states into the corresponding micro-states, and the two solid lines represent their averages. The hyperparameter $q=8$ is chosen for this experiment based on the observation that the CE consistently reaches its highest value when $q=8$, as indicated in Figure \ref{fig:boids3}.

To enhance the interpretability of the learned macro-states and coarse-graining function in NIS+, we utilize the Integrated Gradient (IG) method\cite{Sundararajan2017}(see method section \ref{sec:IG}) to identify the most significant micro-states for each learned emergent macro-state dimension. We normalized the calculated IG and enhanced the maximum gradient of the micro-state in each macro-state and disregard the velocity dimensions of each boid due to their lower correlations with macro-states. The normalized IG is drawn into a matrix diagram(Figure \ref{fig:boids7}). As depicted by Figure \ref{fig:boids7}, the 1st, 2nd, 5th, and 6th dimensions in macro-states correspond to the boids in the first group (with ID<8), while the 3rd, 4th, 7th, and 8th dimensions correspond to the second group (with ID>=8). Thus, the learned coarse-graining strategy uses two positional coordinates to represent all other information to form one dimension of macro-state. 


To compare the learning and prediction effects of NIS+ and NIS, we assess their generalization abilities by testing their performances on initial conditions that differed from the training dataset. During the simulation for generating training data, the positions of all boids are constrained within a circle with a radius of $r$, as depicted in Figure \ref{fig:boids1}. However, we assess the prediction abilities of both models when the initial positions are located on the larger circles. Figure \ref{fig:boids2} shows the MAEs of NIS+ and NIS, which increase with the radius $r$, where smaller prediction errors indicate better generalization. The results clearly demonstrate NIS+'s superior generalization across all tested radii $r$ compared to NIS.

Furthermore, to examine the impact of intrinsic and observational perturbations on CE, two types of noises are introduced. Intrinsic noise is incorporated into the rule by adding random turning angles to each boid at each time step. These angles are uniformly distributed within the interval $\alpha\cdot [-\pi,\pi]$, where $\alpha\in[0,1]$ is a parameter controlling the magnitude of the intrinsic noise. On the other hand, extrinsic noise is assumed to affect the observational micro-states. In this case, we assume that the micro-states of each boid cannot be directly observed, but instead, noisy data is obtained. The extrinsic or observational noise $\delta\sim \mathcal{N}(0,\delta_{max})$ is added to the micro-states, and $\delta_{max}$ is the parameter determining the level of this noise.  

The results are shown in Figure \ref{fig:boids5} and \ref{fig:boids6}, where the normalized MAE increases in both cases, indicating more challenging prediction tasks with increasing intrinsic and extrinsic noises. However, the differences between these two types of noises can be observed by examining the degrees of CE ($\Delta\mathcal{J}$). Figure \ref{fig:boids5} demonstrates that $\Delta\mathcal{J}$ increases with the level of extrinsic noise ($\delta_{max}$), suggesting that coarse-graining can mitigate noise within a certain range and enhance causal effects. When $\delta_{max}<0.1$, the normalized MAE is smaller than $0.3$ (black dashed horizontal line), satisfying the constraint in Equation \ref{old optimization}. In this case, the degree of CE increases with $\delta_{max}$. However, when the threshold of 0.3 is exceeded, and even though $\Delta\mathcal{J}$ decreases, we cannot draw any meaningful conclusion because the violation of the constraint in Equation \ref{old optimization} undermines the reliability of the results.

On the other hand, Figure \ref{fig:boids6} demonstrates that $\Delta\mathcal{J}$ decreases as the level of intrinsic noise ($\alpha$) increases. This can be attributed to the fact that the macro-level dynamics learner attempts to capture the flocking behaviors of each group during this stage. However, as the intrinsic noise increases, the flocking behaviors gradually diminish, leading to a decrease in CE. We have not included cases where $\alpha>0.6$ because the normalized MAE exceeds the threshold of 0.3, the constraints in Equation \ref{old optimization} is violated. Figure \ref{fig:boids4} illustrates real trajectories and predictions for random deflection angle noise with $\alpha=0.4$. It can be observed that in the early stage, the straight-line trend can be predicted, but as the noise-induced deviation gradually increases, the error also grows, which intuitively reflects the reduction in CE. To compare, we also test the same curves for Rosas' $\Psi$, the results are shown in support information section \ref{sec:Boidsmodel} because all the values are negative with large magnitudes.

Therefore, we conclude that NIS+ learns the optimal macro-dynamics and coarse-graining strategy by maximizing EI. This maximization enhances its generalization ability for cases beyond the training data range. The learned macro-states effectively capture average group behavior and can be attributed to individual positions using the IG method. Additionally, the degree of CE increases with extrinsic noise but decreases with intrinsic noise. This observation suggests that extrinsic noise can be eliminated through coarse-graining, while intrinsic noise cannot. More details about the experiments of Boids can be referred to support information section \ref{sec:Boidsmodel}.

\begin{figure}[htbp]    
  \centering            
 \subfloat[]  
  {    
      \label{fig:boids1}\includegraphics[trim={10 0 2 0},clip,width=0.310\textwidth]{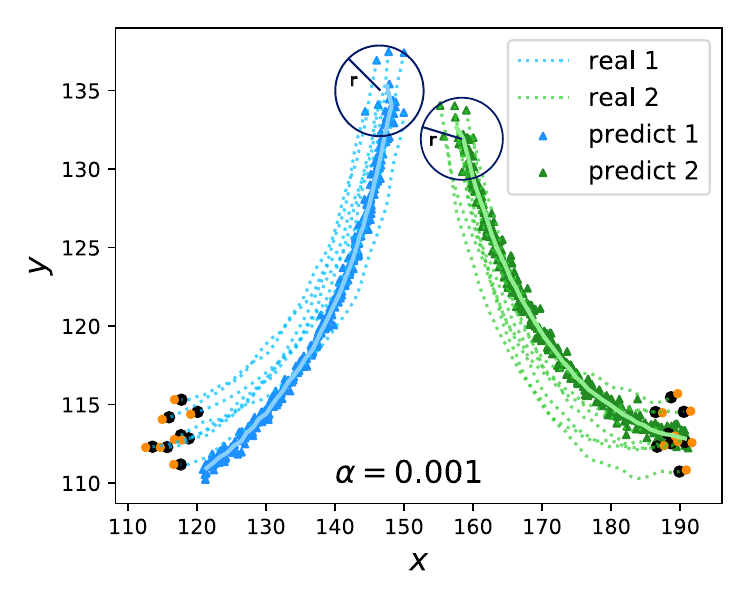}
      
  }
   \subfloat[]  
  {
      \label{fig:boids2}\includegraphics[trim={10 0 2 0},clip,width=0.310\textwidth]{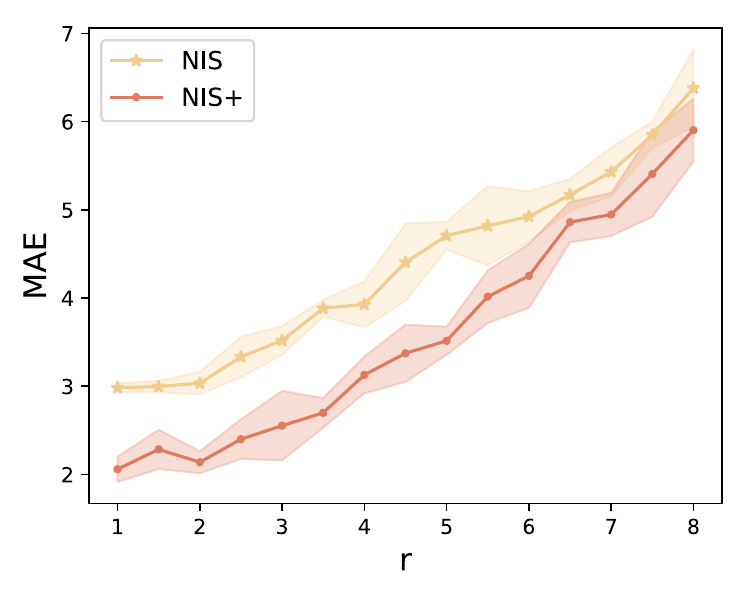}
}
  \subfloat[]
  {
      \label{fig:boids3}\includegraphics[trim={10 0 2 0},clip,width=0.310\textwidth]{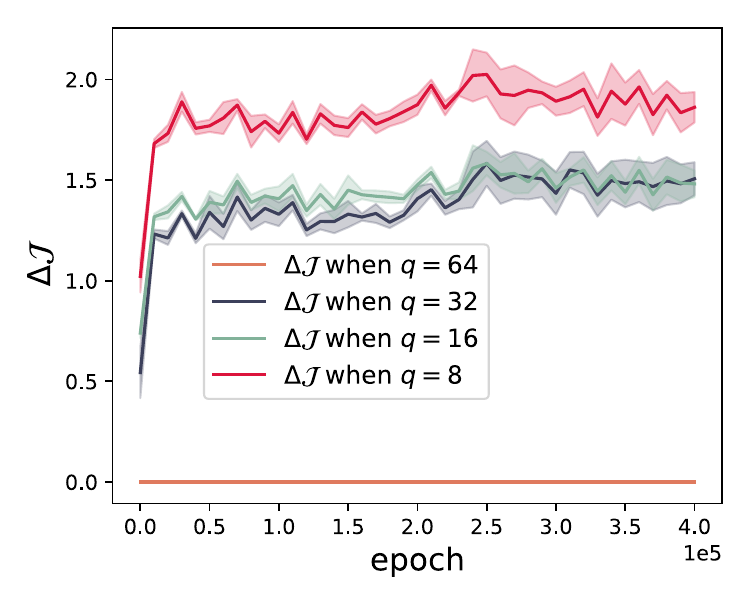}
  }\\
    \subfloat[]
  {
      \label{fig:boids7}
      \includegraphics[trim={0 90 0 110},clip,width=0.6\textwidth]{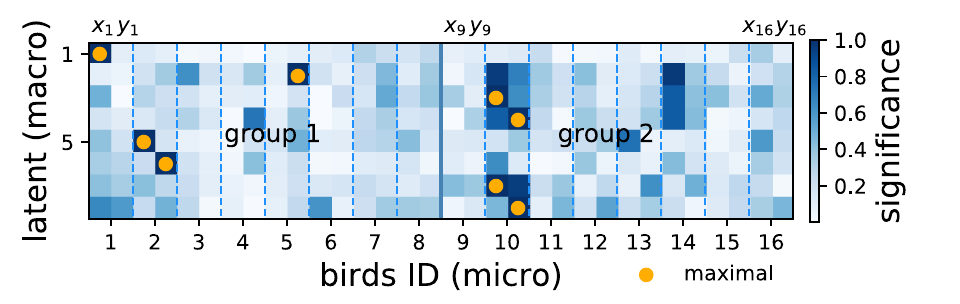}
  }\\
  \subfloat[]
  {
      \label{fig:boids4}\includegraphics[trim={12 0 0 0},clip,width=0.310\textwidth]{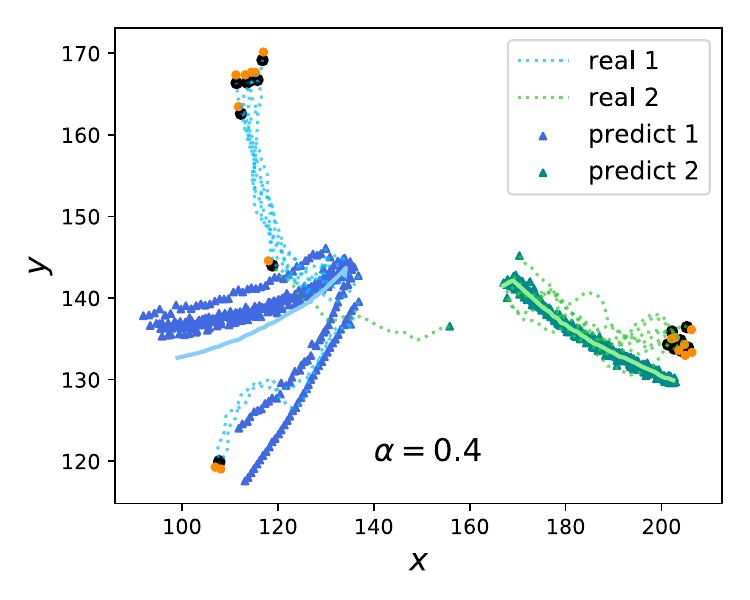}
  }
  \subfloat[]
  {
      \label{fig:boids5}\includegraphics[trim={13 -8.1 13 0},clip,width=0.312\textwidth]{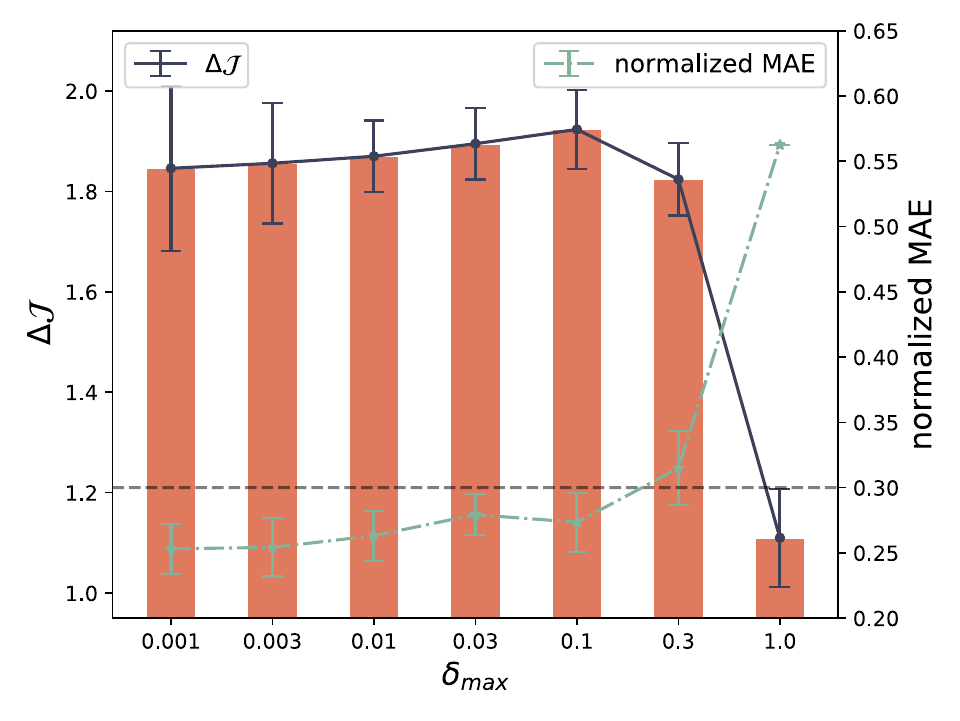}
  }
\subfloat[]
  {
      \label{fig:boids6}\includegraphics[trim={13 -8.1 13 0},clip,width=0.312\textwidth]{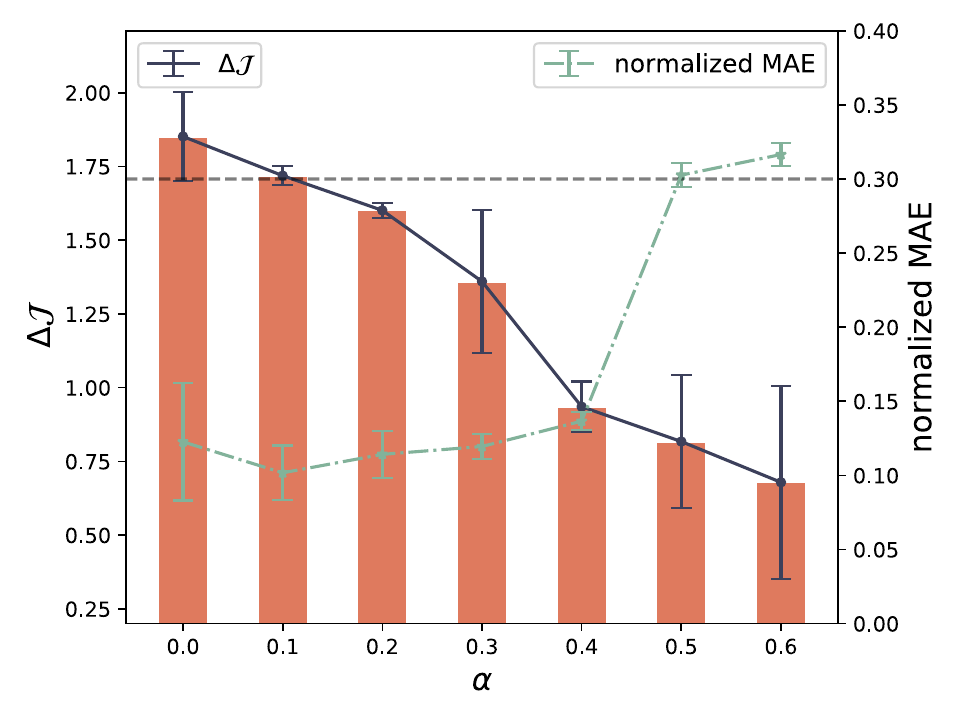}
  }

  \caption{The experimental results of NIS+ on learning the collective flocking behaviors of the Boids model.  (a) and (e) present real and predicted data on boid trajectories under various conditions. Concretely, they present the comparison results for multi-step (50 steps) predictions under the condition of two separating groups, and random deflection angles. Their intrinsic noise levels of $\alpha$ are 0.001 and 0.4, respectively. (b) showcases the escalation of mean absolute error (MAE) for multi-step predictions as the radius $r$, which represents the range of initial positions of boids in (a), extends beyond the limits of the training data. (c) depicts the trend of dimension-averaged causal emergence ($\Delta\mathcal{J}$) changes with training epochs of NIS+ using different hyperparameters of $q$, which represents the scales of different macro-states. (d) presents the saliency map, which visually depicts the association between each macroscopic dimension and the spatial coordinates of each boid. We highlight the most significant corresponding micro-states for each macro-state dimension with orange dots, determined using the integrated gradients (IG) method applied to our model. The horizontal axis represents the $\boldsymbol{x}$ and $\boldsymbol{y}$ coordinates of 16 boids in the microscopic state, while the vertical axis represents the 8 macroscopic dimensions. The dodger-blue dashed line distinguishes the coordinates of different individual boids, while the steel-blue solid line separates the boid groups. (f) and (g) shows the changes in $\Delta\mathcal{J}$ and normalized MAE under different noise levels, (f) for the extrinsic (observational, added on micro-states) noise ($\delta_{max}$) and (g) for intrinsic noise ($\alpha$, added by modifying the dynamical rule of Boids model). In both (f) and (g), the horizontal lines represent the threshold for the violation of the constraint of error in Equation \ref{old optimization}. When the normalized MAE is larger than the threshold 0.3, the constraint is violated, and the results are not reliable.}   
  \label{fig:boids}            
\end{figure}

\subsection{Real fMRI time series data for brains}\label{Schaefer Results}

We test our models on real fMRI time series data of brains for 830 subjects, called AOMIC ID1000\cite{Snoek2021}. The fMRI scanning data is collected when the subjects watch the same clip of movie. Thus, similar experiences of subjects under similar natural stimulus is expected, which corresponds to time series of the same dynamics with different initial conditions. We pre-process the raw data through Schaefer atlas method~\cite{Schaefer2017} to reduce the dimensionality of time series for each subject from roughly 140,000 (it varies among subjects.) to 100 such that NIS+ can operate and obtain more clear results. Then, the first 800 time series data are selected for training and the remaining 30 time series are for test. We also compare our results with another fMRI dataset AOMIC PIOP2~\cite{Snoek2021} for 50 subjects in resting state. Further description about the dataset can be referred to the method section.

To demonstrate the predictive capability of NIS+ for micro-states, Figure \ref{fig:brain}(a) illustrates the changes in normlized MAE with the prediction steps of the micro-dynamics on test data for different hyperparameters $q$. It is evident that NIS+ performs better in predictions when $q=27$ and $q=1$. Specifically, the curve for $q=27$ exhibits a slower rate of increase compared to the curve for $q=1$ as the prediction steps increase. This suggests that selecting the hyperparameter $q$ as 27 may be more suitable than 1.

However, Figure \ref{fig:brain}(b) suggests a different outcome. When comparing the degree of CE ($\Delta\mathcal{J}$) for different hyperparameters $q$ (the green bars), the highest $\Delta\mathcal{J}$ is observed when $q=1$. Conversely, a negative $\Delta\mathcal{J}$ value is obtained when $q=27$. This indicates that the improved prediction results may be attributed to overfitting when $q=27$. Thus, $q=1$ outperforms other values of $q$ in terms of $\Delta\mathcal{J}$. This finding is supported by the NIS framework (the red bars), despite observing a larger standard deviation in $\Delta\mathcal{J}$ when $q=1$. Furthermore, we also compare the results of CE with resting data and observe that peaks are reached at $q=7$, which is just the number of subsystems in Schaefer atalas, for both NIS (deep blue bars) and NIS+ (yellow bars). Therefore, we can conclude that when subjects watch movies, the activities in different brain areas can be represented by a single real number at each time step. More analysis for the resting data can be referred to the support information section~\ref{sec:resting state}.



To investigate how NIS+ coarse-grains the input data into a single-dimensional macro-state, we also utilize the IG method to identify the most significant dimensions of the micro-state \cite{Sundararajan2017}. The results are depicted in Figure \ref{fig:brain}(c) and (d). We observe that the visual (VIS) sub-networks exhibit the highest attribution (Figure \ref{fig:brain}(c)). These visual sub-networks represent the functional system that subjects utilize while watching movie clips. Furthermore, we can visualize the active areas in finer detail on the brain map (Figure \ref{fig:brain}(d)), where darker colors indicate greater attribution to the single macro-state. Therefore, the regions exhibiting similar darkest colors identified by NIS+, which correspond to the deep visual processing brain region, could potentially represent the ``synergistic core''~\cite{synergestic_core} when the brain is actively engaged in watching movies. The numeric neurons in these areas may collaborate and function collectively. However, this conclusion should be further confirmed and quantified by decomposing the mutual information between micro-states and macro-states into synergistic, redundant, and unique information~\cite{williams2010nonnegative,Rosas_Mediano_Jensen_Seth_Barrett_Carhart-Harris_Bor_2020}.

In conclusion, NIS+ demonstrates its capability to learn and coarse-grain the intricate fMRI signals from the brain, allowing for the simulation of complex dynamics using a single macro-state. This discovery reveals that the complex neural activities during movie watching can be encoded by a one-dimensional macro-state, with a primary focus on the regions within the visual (VIS) sub-network. The robustness of our findings is further supported by comparable results obtained through alternative methods for data pre-processing, as demonstrated in the support information section \ref{sec:brain}. 

\begin{figure}[!ht]
	\centering
	\includegraphics[width=\textwidth]{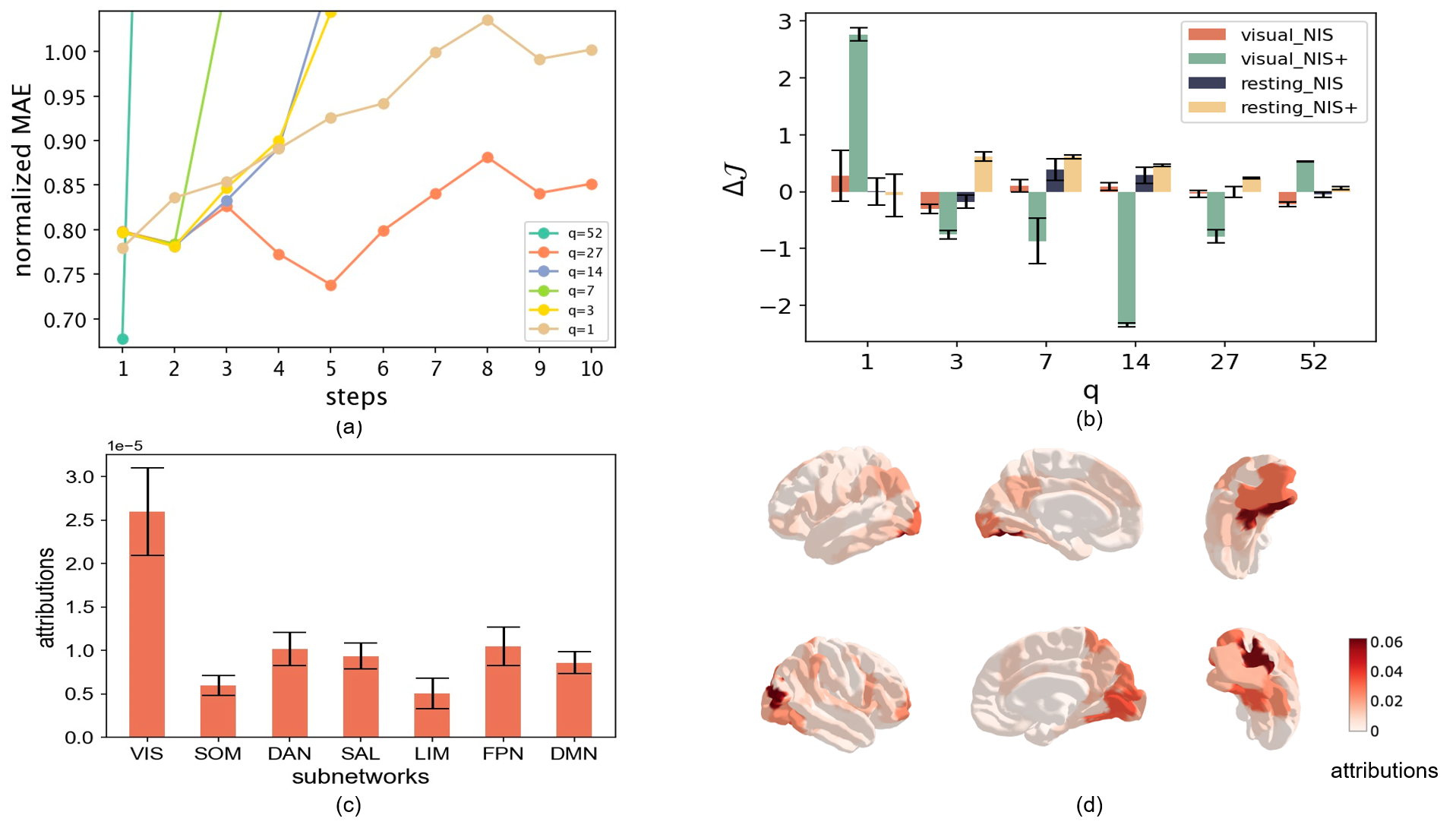}
	\caption{The learning results, the degree of causal emergence, and the attribution analysis of NIS+ and NIS on the fMRI data for the brains. (a) The mean errors of the multi-step predictions increase with the prediction steps under different scales ($q$) on the test dataset. (b) Measures of CE (dimension averaged, $\Delta\mathcal{J}$) are compared among different models and different datasets including movie-watching fMRI (visual fMRI) and resting fMRI. The bars show the averaged results for 10 repeating experiments, and the error bar represents the standard deviation. (c) The average attributions of the sub-networks under the Schaefer Atlas, calculated using the integrated gradient (IG) analysis method on the encoder with a scale of $q=1$, are presented. The error bars represent the standard errors. (d) Attribution maps for movie-watching(visual) fMRI data are displayed. The maps show the left hemisphere from the view of the left, left hemisphere from the view of the right, right hemisphere from the view of the right, and right hemisphere from the view of the left. Also, the right column reflects a detailed map for visual areas. The upper is left visual areas and the bottom is right visual areas. The colors represent the normalized absolute values of the integrated gradient.  }
	\label{fig:brain}
\end{figure}

\section{Concluding Remarks}
Inspired by the theory of causal emergence, this paper introduces a novel machine learning framework called Neural Information Squeezer Plus (NIS+) to learn emergent macro-dynamics, and suitable coarse-graining methods directly from data. Additionally, it aims to quantify the degree of CE under various conditions. 

The distinguishing feature of our framework, compared to other machine learning frameworks, is its focus on maximizing the effective information (EI) of the learned macro-dynamics while maintaining effectiveness constraints. This enables the learned emergent macro-dynamics to capture the invariant causal mechanism that is as independent as possible from the distribution of input data. This feature not only enables NIS+ to identify CE in data across different environments but also enhances its ability for generalization on the environments which are distinct from training data. By introducing the error constraint in Equation \ref{old optimization}, we can address the issues raised in \cite{eberhardt2022causal} and enhance the robustness of the maximization of EI framework. As a result, NIS+ extends Hoel's theory of CE to be applicable to both discrete and continuous dynamical systems, as well as real data.


Three experiments were conducted to evaluate the capabilities of NIS+ in learning and generalization, and quantifying CE directly from data. Furthermore, we applied this framework to the domain of the Game of Life (refer to the support information section \ref{sec:gamelife}). These experiments encompassed three simulation scenarios and one real fMRI dataset for 830 human subjects while watching same movie clips.

The experiments indicate that by maximizing EI, NIS+ outperforms other machine learning models in tasks such as multi-step predictions and pattern capturing, even in environments that were not encountered during the training process. Consequently, NIS+ enables the acquisition of a more robust macro-dynamics in the latent space.

Furthermore, the experiments show that NIS+ can quantify CE in a more reasonable way than Rosas' $\Psi$ indicator. With this framework, we can distinguish different scenarios in the data and identify which settings contain more regular patterns, as demonstrated in the experiment conducted on the Game of Life. The experiment on the Boid model also provides insights into how two types of noise can impact the degrees of CE. The conclusion is that extrinsic noise may increase CE, while intrinsic noise may decrease it. This indicates that extrinsic noise, arising from observational uncertainty, can be mitigated by the learned coarse-graining strategy. On the other hand, intrinsic noise, stemming from inherent uncertainty in the dynamical rules, cannot be eliminated.


The process of coarse-graining a complex system can be effectively learned by NIS+ and illustrated using the Integrated Gradients (IG) method, as demonstrated in the experiments involving the Boid model and the brain dataset. Through this method, the relationship between macro-states and micro-states can be visualized, allowing for the identification of the most significant variables within the micro-states. In the brain experiment, it was observed that the most crucial information stored in the unique macro-state exhibited a strong correlation with the micro-variables of the visual subnetwork. This finding highlights the ability of NIS+ to identify an appropriate coarse-graining strategy that aligns with meaningful biological interpretations.


NIS+ holds potential for various applications in data-driven modeling of real complex systems, such as climate systems, collective behaviors, fluid dynamics, brain activities, and traffic flows. By learning more robust macro-dynamics, the predictive capabilities of these systems can be enhanced. For instance, El Niño, which arises from the intricate interplay of oceanic and atmospheric conditions, exemplifies the emergence of a major climatic pattern from underlying factors. Understanding these emergent macro-dynamics can be instrumental in modeling and predicting El Niño events. By leveraging NIS+ to capture and quantify the CE in such complex systems, we can gain valuable insights and improve our ability to forecast their behavior.

Further extensions of the EI maximization method to other problem domains, such as image classification and natural language understanding, in addition to dynamics learning, warrant attention. Specifically, applying this principle to reinforcement learning agents equipped with world models has the potential to enhance performance across various environments. By incorporating the EI maximization approach into these domains, we can potentially improve the ability of models to understand and interpret complex visual and textual information. This can lead to advancements in tasks such as image recognition, object detection, language understanding, and machine translation.

The relationship between causal emergence and causal representation learning (CRL)\cite{scholkopf2022statistical} deserves further exploration. The NIS+ framework can be seen as a form of CRL, where the macro-dynamics serve as the causal mechanism and the coarse-graining strategy acts as the representation. As such, techniques employed in CRL can be applied to find CE in data. Conversely, the concept of CE and coarse-graining can be leveraged in CRL to enhance the interpretability of the model. By integrating the principles and methodologies from both CE and CRL, we can potentially develop more powerful and interpretable models. This synergy can lead to a deeper understanding of the causal relationships within complex systems and enable the extraction of meaningful representations that capture the emergent patterns and causal mechanisms present in the data. Further research in this direction can contribute to advancements in both CE and causal representation learning.

Another interesting merit of NIS+ is its potential contribution to emergence theory by reconciling the debate on whether emergence is an objective concept or an epistemic notion dependent on the observer. By designing a machine to maximize EI, we can extract objective emergent features and dynamics. The machine serves as an observer, but an objective one. Therefore, if the machine observer detects interesting patterns in the data, emergence occurs.

However, there are several limitations in this paper that should be addressed in future studies. Firstly, the requirement of a large amount of training data for NIS+ to learn the macro-dynamics and coarse-graining strategy may not be feasible in many real-world cases. If the training is insufficient, it may lead to incorrect identification of CE. Therefore, it is necessary to incorporate other numeric method, such as Rosas' $\Phi$ID~\cite{Rosas_Mediano_Jensen_Seth_Barrett_Carhart-Harris_Bor_2020}, to make accurate judgments. One advantage of NIS+ is its ability to identify coarse-grained macro-states, which can then be used as input for Rosas' method. Secondly, the interpretability of neural networks, particularly for the macro-dynamics learner, remains a challenge. Enhancing the interpretability of the learned models can provide valuable insights into the underlying mechanisms and improve the trustworthiness of the results. Additionally, the learned macro-dynamics in this work are homogeneous, resulting in the same EI value across different locations in the macro-state space. However, in real-world applications, this may not always hold true. Therefore, it is necessary to develop local measures for CE, especially for systems with heterogeneous dynamics. This would facilitate the identification of local emergent patterns, making the analysis more nuanced and insightful. Lastly, the current framework is primarily designed for Markov dynamics, while many real complex systems exhibit long-term memory or involve unobservable variables. Extending the framework to accommodate non-Markov dynamics is an important area for future research. Addressing these limitations and exploring these avenues for improvement will contribute to the advancement of the field and enable the application of NIS+ to a wider range of complex systems.

\section{Methods and Data}

In order to provide a comprehensive understanding of the article, we will first dedicate a certain amount of space to introduce Erik Hoel's theory of causal emergence, and we introduce the details of the NIS and NIS+ frameworks and other used techniques. After that, the details of the fMRI time series data is introduced.

\subsection{Machine Learning Frameworks}
We will introduce the details of the two frameworks NIS and NIS+ in this section.
\subsubsection{Neural Information Squeezer (NIS)}
\label{sec:nis}
NIS use neural networks to parameterize all the functions to be optimized in Equation \ref{old optimization}, in which the coarse-graining and anti-coarsening function $\phi$ and $\phi^{\dag}$ are called encoder and decoder, respectively, and the macro-dynamics $f_q$ is called dynamics-learner. Second, considering the symmetric position between $\phi$ and $\phi^{\dag}$, invertible neural network(with RealNVP framework\cite{Dinh_Sohl-Dickstein_Bengio_2016}, details can be referred to the following sub-section) is used to reduce the model complexity and to make mathematical analysis being possible\cite{zhang2022neural}. 

Concretely, 
\begin{equation}
    \label{eq:phi}
    \phi\equiv Proj_q(\psi_{\omega}),
\end{equation}
where $\psi_{\omega}:\mathcal{R}^p\rightarrow\mathcal{R}^p$ is an invertible neural network with parameters $\omega$, and $Proj_q$ represents the projection operation with the first $q$ dimensions reserved to form the macro-state $\boldsymbol{y}$, and the last $p-q$ dimensional variable $\boldsymbol{y}'$ are dropped. Empirically, $\boldsymbol{y}'$ can be approximately treated as a Gaussian noise and independent with $\boldsymbol{y}$ or we can force $\boldsymbol{y}'$ to be independent Gaussian by training the neural network. Similarly, in a symmetric way, $\phi^{\dag}$ can be approximated in the following way: for any input $\boldsymbol{y}\in \mathcal{R}^q$,
\begin{equation}
\label{eq:phidag}
    \phi^{\dag}(\boldsymbol{y})= \psi_{\omega}^{-1}(\boldsymbol{y}\bigoplus \xi),
\end{equation} 
where $\xi$ is a standard Gaussian random vector with $p-q$ dimensions, and $\bigoplus$ represents the vector concatenation operation. 

Last, the macro-dynamics $f_q$ can be parameterized by a common feed-forward neural network $f_{\theta}$ with weight parameters $\theta$. Its number of input and output layer neurons equal to the dimensionality of the macro state $q$. It has two hidden layers, each with 64 neurons, and the output is transformed using LeakyReLU. The detailed architectures and parameter settings of $\psi_{\omega}$ and $f_{\theta}$ can be referred to the supporting information (support information section \ref{sec:parameter table}). To compute EI, this feed-forward neural network is regarded as a Gaussian distribution, which models the conditional probability $p(\boldsymbol{\hat{y}}_{t+1}|\boldsymbol{y}_t)$.

Previous work~\cite{zhang2022neural} has been proved that this framework merits many mathematical properties such as the macro-dynamics $f_{\theta}$ forms an information bottleneck, and the mutual information between $\boldsymbol{y}_t$ and $\hat{\boldsymbol{y}}_{t+1}$ approximates the mutual information between $\boldsymbol{x}_t$ and $\boldsymbol{x}_{t+1}$ when the neural networks are convergent such that the constraint in Equation \ref{old optimization} is required.

\subsubsection{Neural Information Squeezer Plus (NIS+)}
\label{sec:nis+}
To maximize EI defined in Equation \ref{old optimization}, we extend the framework of NIS to form NIS+. In NIS+, we first use the formula of mutual information and variational inequality to convert the maximization problem of mutual information into a machine learning problem, and secondly, we introduce a neural network $g_{\theta'}$ to learn the inverse macro-dynamics that is to use $\boldsymbol{y}_{t+1}=\phi(\boldsymbol{x}_{t+1})$ to predict $\boldsymbol{y}_t$ such that mutual information maximization can be guaranteed. Third, the probability re-weighting technique is employed to address the challenge of computing intervention for a uniform distribution, thereby optimizing the EI. All of these techniques compose neural information squeezer plus (NIS+).

Formally, the maximization problem under the constraint of inequality defined by Equation \ref{old optimization} can be converted as a minimization of loss function problem without constraints, that is:
\begin{gather}
\label{new optimization}
\min_{\omega,\theta,\theta'} \sum_{t=1}^{T-1} w(\boldsymbol{x}_t)||\boldsymbol{y}_t-g_{\theta'}(\boldsymbol{y}_{t+1})||+\lambda||\hat{\boldsymbol{x}}_{t+1}-\boldsymbol{x}_{t+1} ||,
\end{gather}
where $\omega,\theta,\theta'$ are the parameters of neural networks of $\psi_{\omega}$, $f_{\theta}$, and $g_{\theta'}$, respectively. $\boldsymbol{y}_t=\phi(\boldsymbol{x}_t)=Proj_q(\psi_{\omega}(\boldsymbol{x}_t))$ and $\boldsymbol{y}_{t+1}=\phi(\boldsymbol{x}_{t+1})=Proj_q(\psi_{\omega}(\boldsymbol{x}_{t+1}))$ are the macro-states.  $\lambda$ is a Lagrangian multiplier which will be taken as a hyper-parameter in experiments. $w(\boldsymbol{x}_t)$ is the inverse probability weights which is defined as:

\begin{equation}
    w(\boldsymbol{x}_t)=\frac{\Tilde{p}(\boldsymbol{y}_t)}{p(\boldsymbol{y}_t)}=\frac{\Tilde{p}(\phi(\boldsymbol{x}_t))}{p(\phi(\boldsymbol{x}_t))},
\end{equation}
where $\Tilde{p}$ is the new distribution of macro-states $\boldsymbol{y}_t$ after intervention for $do(\boldsymbol{y}_t\sim U_q)$, and $p$ is the natural distribution of the data. In practice, $p(\boldsymbol{y}_t)$ is estimated through kernel density estimation (KDE)\cite{rosenblatt1956remarks}(Please refer method section \ref{sec:KDE}). The approximated distribution, $\Tilde{p}(\boldsymbol{y}_t)$, is assumed to be a uniform distribution, which is characterized by a constant value. Consequently, the weight $w$ is computed as the ratio of these two distributions.

We can prove a mathematical theorem to guarantee such problem transformation as mentioned below:

\newtheorem{mythm}{Theorem}[section]
\begin{mythm}[Problem Transformation Theorem]
\label{thm:optimize}
For a given value of $q$, assuming that $\omega^\ast, \theta^\ast$, and $\theta'^\ast$ are the optimal solutions to the unconstrained objective optimization problem defined by equation (\ref{new optimization}). Then $\phi^\ast\equiv Proj_q(\psi_{\omega^\ast}),f^\ast_q\equiv f_{\theta^\ast}$, and $\phi^{\dag,\ast}(\cdot)\equiv \psi^{-1}_{\omega^\ast}(\cdot\bigoplus \xi)$, where $\xi\sim \mathcal{N}(0,I_{p-q})$ are the optimal solution to the constrained objective optimization problem Eq.(\ref{old optimization}).

\end{mythm}

The details of proof can be seen in support information section \ref{sec:optimize}.

\begin{figure}[htbp]
	\centering
    \includegraphics[width=0.9\textwidth]{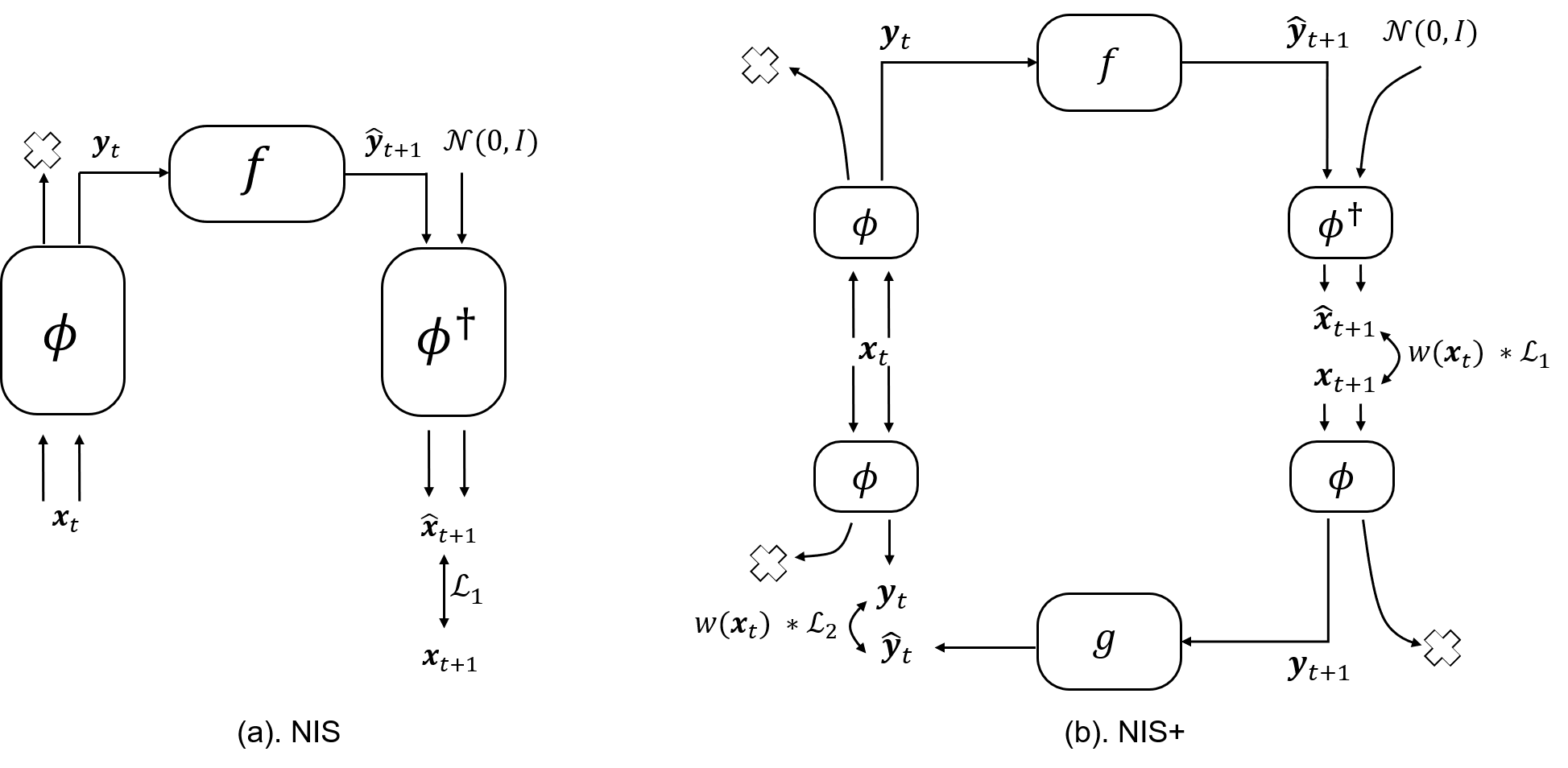}

\caption{The frameworks of NIS (a) and NIS+ (b). The boxes in the diagram represent functions or neural networks, while the arrow pointing to a cross signifies the operation of information dropping. $\boldsymbol{x}_t$ and $\boldsymbol{x}_{t+1}$ represent the observational data of micro-states, and $\hat{\boldsymbol{x}}_{t+1}$ represents the predicted micro-state. The macro-states, denoted as $\boldsymbol{y}_t=\phi(\boldsymbol{x}_{t})$ and $\boldsymbol{y}_{t+1}=\phi(\boldsymbol{x}_{t+1})$, are obtained by encoding the micro-states using the encoder. Similarly, the predicted macro-states, $\hat{\boldsymbol{y}}_t=\phi(\hat{\boldsymbol{x}}_{t})$ and $\hat{\boldsymbol{y}}_{t+1}=\phi(\hat{\boldsymbol{x}}_{t+1})$, are obtained by encoding the predictions of micro-states.}
\label{fig:frame}
\end{figure}

Discriminate from Figure \ref{fig:frame}(a), a new computational graph for NIS+ which is depicted in Figure \ref{fig:frame}(b) is implied by the optimization problem defined in Equation \ref{new optimization}. For given pair of data: $\boldsymbol{x}_t, \boldsymbol{x}_{t+1}$, two dynamics, namely the forward dynamics $f_{\theta}$ and the reverse dynamics $g_{\theta'}$ are trained by optimizing two objective functions $\mathcal{L}_1=\sum_{t=1}^{T-1} w(\boldsymbol{x}_t)||\boldsymbol{y}_t-g_{\theta'}(\boldsymbol{y}_{t+1})||$ and $\mathcal{L}_2=\sum_{t=1}^{T-1}||\hat{\boldsymbol{x}}_{t+1}-\boldsymbol{x}_{t+1} ||$, simultaneously. In the training process, the encoder, $\phi$ and the decoder, $\phi^{\dag}$ are shared.

This novel computational framework of NIS+ can realize the maximization of EI under the coarse-grained emergent space. Therefore, it can optimize an independent causal mechanism($f$) represented by $f_{\theta}$ on the emergent space. It can also be used to quantify CE in raw data once the macro-dynamics $f_{\theta}$ is obtained for different $q$.

\subsubsection{Extensions for Practical Computations}
\label{sec:extensions}
When we apply NIS+ to the data generated by various complex systems such as cellular automata and multi-agent systems, the architecture of encoder(decoder) should be extended for stacked and parallel structures. Fortunately, NIS+ is very flexible to different real scenarios which means the important properties such as Theorem \ref{thm:optimize} can be retained.

\begin{figure}[htbp]    
  \centering  
  \includegraphics[width=0.9\textwidth]{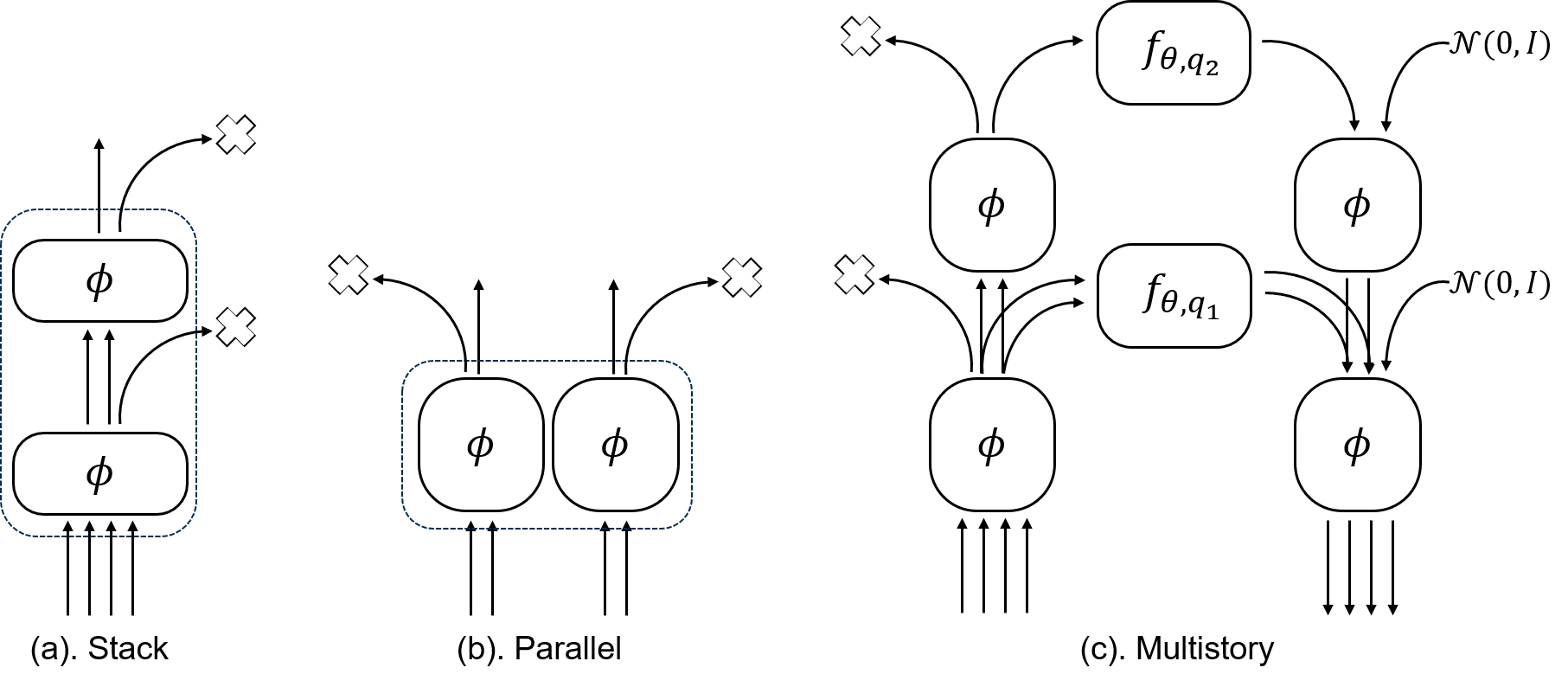}

  \caption{Three extended structures of the NIS+ encoder, denoted as (a) and (b), as well as the complete structure (c), are designed to facilitate the implementation of complex tasks. (a) represents a stacked encoder, where basic encoders are layered one on top of another. This arrangement allows for the filtering of input information in a hierarchical manner. (b) showcases a parallel encoder, where multiple basic encoders with shared parameters are combined to form a larger encoder. This parallel configuration enables the processing of input information in parallel, enhancing efficiency. (c) illustrates a multistory structure of the complete NIS+ framework. In this structure, the encoders and decoders are stacked, and the dynamics learners in different layers operate in parallel. This design facilitates simultaneous training for different dynamics learners with distinct dimensions and enables parallel searching for the optimal dimension $q$.}   
  \label{fig:structures}            
\end{figure}

Firstly, it is important to consider that when dealing with high-dimensional complex systems, discarding multiple dimensions at once can pose challenges for training neural networks. Therefore, we replace the encoder in NIS+ with a \textbf{Stacked Encoder}. As shown in Figure \ref{fig:structures}(a), this improvement involves stacking a series of basic encoders together and gradually discard dimensions, thereby reducing the training difficulty. 

In addition, complex systems like cellular automota are always comprise with multiple components with similar dynamical rules. Therefore, we introduce the \textbf{Parallel Encoder} to represent each component or groups of these components and combine the results together. Concretely, as shown in Figure \ref{fig:structures}(b), inputs may be grouped based on their physical relationships(e.g. adjacent neighborhood), and each group is encoded using a basic encoder. By sharing parameters among the encoders, neural networks can capture homogeneous coarse-graining rules efficiently and accurately. Finally, the macro variables obtained from all the encoders are concatenated into a vector to derive the overall macro variables. Furthermore, we can combine this parallel structures with other well-known architectures like convolutional neural networks.

To enhance the efficiency of searching for the optimal scale, we leverage the multiple scales of hidden space obtained through the stacked encoder by training multiple dynamics learners in different scales simultaneously. This forms the framework of \textbf{Multistory NIS+} as shown in Figure \ref{fig:structures}(c). This approach is equivalent to searching macro-dynamics for different $q$ and thereby avoiding to retrain the encoders.

Theorem \ref{thm:optimize_extension} guarantees that the important properties such as Theorem \ref{thm:optimize} can be retained for the extensions of stacked and parallel encoders. And a new universal approximation theorem (Theorem \ref{thm.universal}) can be proved such that multiple stacked encoders can approximate any coarse-graining function (any map defined on $\mathcal{R}^p\times\mathcal{R}^q$).

\begin{mythm}[Problem Transformation in Extensions of NIS+ Theorem]
\label{thm:optimize_extension}
When the encoder of NIS+ is replaced with an arbitrary  combination of stacked encoders and parallel encoders, the conclusion of Theorem 2.1 still holds true.
\end{mythm}

\begin{mythm}[Universal Approximating Theorem of Stacked Encoder]
\label{thm.universal}
For any continuous function $f$ which is defined on $K\times \mathcal{R}^p$, where $K\in \mathcal{R}^p$ is a compact set, and $p>q\in \mathcal{Z^+}$, there exists an integer $s$ and an extended stacked encoder $\phi_{p,s,q}: \mathcal{R}^p\rightarrow \mathcal{R}^q$ with $s$ hidden size and an expansion operator $\eta_{p,s}$ such that:

\begin{equation}
    \phi_{p,s,q}\simeq f,
\end{equation}
Thereafter, an extended stacked encoder with expansion operators is of universal approximation property which means that it can approximate and simulate any coarse-graining function defined on $\mathcal{R}^p\times \mathcal{R}^q$.
\end{mythm}

The extended stacked encoder $\phi_{p,s,q}$ refers to stacking two basic encoders together and gradually reducing dimensions, encoding the input from $p$ dimensions to $q$ dimensions, with an intermediate dimension of $s$. Besides basic operations such as invertible mapping and projection, a new operation, vector expansion($\eta_{p,s}$, e.g., $\eta_{2,5}(1,2,3)=(1,2,3,1,2)$), should be introduced, and locate in between the two encoders. All the proves of these theorems are referred to support information section \ref{sec:op extension} and \ref{sec:universal}.

\subsubsection{Training NIS+}
\label{sec:train nis+}
In practice, two stages are separated during training process for NIS+ and all extensions. The first stage only trains the forward neural networks (the upper information channel as shown in Figure \ref{fig:nis+}) such that the loss $\mathcal{L}_1$ is small enough. Then, the second stage which only trains the neural networks for the reverse dynamics (the lower information channel as shown in Figure \ref{fig:nis+}) is conducted. Because the trained inverse dynamics $g_{\theta'}$ is never used, the second stage is for training $\phi_{\omega}$ in essence.

\section{Acknowledgement}
The authors would like to acknowledge all the members who participated in the "Causal Emergence Reading Group" organized by the Swarma Research and the "Swarma - Kaifeng" Workshop. We are especially grateful to Professors Yizhuang You, Peng Cui, Jing He,  Dr.Chaochao Lu, and Dr. Yanbo Zhang for their invaluable contributions and insightful discussions throughout the research process. We also thank the large language model ChatGPT for its assistance in refining our paper writing.\\


\clearpage
\setcounter{section}{0}
\setcounter{figure}{0}
\setcounter{equation}{0}
\textbf{\huge Support Information}
\newtheorem{lem}{Lemma}[section]
\section{A Brief Introduction of Causal Emergence Theory}
The basic idea of Hoel's theory is that we have two different views (micro and macro) of a same Markov dynamical system, and CE occurs if the macro-dynamics has stronger causal effect than the micro-dynamics. As shown in Figure \ref{fig:causalemergence}, for example, the micro-scale system is composed by many colliding balls. The macro-scale system is a coarse-grained description of the colliding balls with colorful boxes where each box represents the number of balls within the box. In Figure \ref{fig:causalemergence}, the vertical axis represents the scale(micro- and macro-), and the horizontal axis represents time, which depicts the evolution of the system's dynamics in both scales. All the dynamical systems considered are Markovian, therefore two temporal successive snapshots are necessary. The dynamics in both micro- and macro-scales are represented by transitional probability matrix(TPM): $f_m$ for micro-dynamics and $f_M$ for macro-dynamics. 

We can measure the strength of its causal effects for each dynamics(TPM) by effective information (EI). EI can be understood as an intervention-based version of
mutual information between two successive states in a dynamical system over time. For any TPM $f$, EI can be calculated by:
\begin{equation}
    \label{eq:definition_EI}
    \mathcal{J}=I(S_{t};S_{t-1}|do(S_{t-1}\sim U(S)))=\frac{1}{N}\sum_{i,j=1}^{N}\left(f_{ij}\log f_{ij}-f_{ij}\log \left(\frac{\sum_{k=1}^{N}f_{kj}}{N}\right)\right),
\end{equation}
where, $S_t$ is the random variable to represent the state of the system at time $t$, $do(S_{t-1}\sim U(S))$ represents the do-operator~\cite{Goldberg_2019} that intervene the state of the system at time $t-1$ to force that $S_{t-1}$ follows a uniform distribution on the value space $S$ of $S_t$. $f_{ij}$ is the probability at the $i$th row and the $j$th column. $N=|S|$ is the total number of states and is also the number of rows or columns. $\log$ represents the logarithmic operation with a base of 2. Equation \ref{eq:definition_EI} can be further decomposed into two terms: determinism and non-degeneracy which characterize that how the future state can be determined by the past state and vice versa, respectively\cite{hoel2013quantifying}.

With EI, we can compare two TPMs. If $\mathcal{J}(f_M)$ is larger than $\mathcal{J}(f_m)$, we say CE occurs in this dynamical system. In practice, we calculate another value called causal emergence to judge if CE takes place, that is:

\begin{equation}
\label{eq:causal_emergence}
\Delta \mathcal{J}=\mathcal{J}(f_M)-\mathcal{J}(f_m).
\end{equation}

Therefore, if $\Delta\mathcal{J}>0$, then CE occurs. $\Delta\mathcal{J}$ can be treated as the quantitative measure of emergence for Markov dynamics. An example of CE on Markov chain is shown in Figure \ref{fig:causalemergence}(b).

\section{Measures and calculations of EI and causal emergence for neural networks}
\label{sec:ei for nn}
The measure of EI of discrete Markov chains can be calculated by Equation \ref{eq:definition_EI}, however, it is not suitable for neural networks because a neural network is a deterministic function. Therefore, we need to extend the definition of EI for general neural networks. We continue to use the method defined in \cite{zhang2022neural}. The basic idea is to understand a neural network as a Gaussian distribution conditional on the input vector with the mean as the deterministic function of the input values, and the standard deviations as the mean square errors of this neural network on the training or test datasets.

In general, if the input of a neural network is $X=(x_1,x_2,\cdot\cdot\cdot,x_n)\in [-L,L]^n$, which means $X$ is defined on a hyper-cube with size $L$, where $L$ is a very large integer. The output is $Y=(y_1,y_2,\cdot\cdot\cdot,y_m)$, and $Y=\mu(X)$. Here $\mu$ is the deterministic mapping implemented by the neural network: $\mu: \mathcal{R}^n\rightarrow \mathcal{R}^m$, and its Jacobian matrix at $X$ is $\partial_{X'} \mu(X)\equiv \left\{\frac{\partial \mu_i(X')}{\partial X'_j}\left|_{X'=X}\right.\right\}_{nm}$. If the neural network can be regarded as a Gaussian distribution conditional on given $X$, then the effective information (EI) of the neural network can be calculated in the following way:
\begin{equation}
\begin{aligned}
\label{eq:ei_L}
        EI_L(\mu)=I(do(X\sim U([-L,L]^{n};Y)\approx & -\frac{m+m \ln (2\pi)+ \sum_{i=1}^m\sigma_i^2}{2}\\
        & +n\ln (2L) + \mathbf{E}_{X\sim U([-L,L]^n} \left(\ln |\det(\partial_{X'} \mu(X))|\right).
\end{aligned}
\end{equation}
where, $\Sigma=diag(\sigma_1^2,\sigma_2^2,\cdot\cdot\cdot,\sigma_m^2)$ is the co-variance matrix, and $\sigma_i$ is the standard deviation of the output $y_i$ which can be estimated by the mean square error of $y_i$ under given $x_i$. $U([-L,L]^n)$ is the uniform distribution on $[-L,L]^n$, and $|\cdot|$ is absolute value, and $\det$ is determinant. If $\det(\partial_{X'} \mu(X))\equiv 0$ for all $X$, then we set $EI\approx 0$.

However, Equation \ref{eq:ei_L} can not be applied in real cases directly because it will increase as the dimension of input $n$ or output $m$ increases\cite{zhang2022neural}. Therefore, larger EI is always expected for the models with higher dimension. The way to solve this problem is by dividing the input dimension to define dimension averaged effective information(dEI) and is denoted as $\mathcal{J}$:
\begin{equation}
    \mathcal{J}_L = \frac{EI_L(\mu)}{n}
\end{equation}

When the numbers of input and output are identical ($m=n$, this condition is always hold for all the results reported in the main text), Equation \ref{eq:ei_L} becomes:
\begin{equation}
    \label{eq:dEI}
    \mathcal{J}_L(\mu)=-\frac{1+\ln(2\pi)+\sum_{i=1}^n\sigma_i^2/n}{2}+\ln(2L)+\frac{1}{n}\mathbf{E}_{X\sim U([-L,L]^n} \left(\ln |\det(\partial_{X'} f(X))|\right).
\end{equation}

We always use this measure to quantify EI for neural networks in the main text. However, this measure is also not perfect because it is dependent on a free parameter $L$, the range of the domain for the input data. Our solution is to calculate the dimension averaged CE to eliminate the influence of $L$. For macro-dynamics $f_M$ with dimension $q$ and micro-dynamics $f_m$ with dimension $p$, we define dimension averaged CE as:
\begin{equation}
    \label{eq.dCE}
    \Delta\mathcal{J}_L(f_M,f_m)=\mathcal{J}_L(f_M)-\mathcal{J}_L(f_m)=\frac{EI_L(f_M)}{q}-\frac{EI_L(f_m)}{p}.
\end{equation}

If $f_M$ and $f_m$ are parameterized by neural networks of $\mu_M$ with dimension $q$ and $\mu_m$ with dimension $p$, then

\begin{equation}
    \label{eq.dCE_gauss}
    \begin{aligned}
        \Delta\mathcal{J}=&\left(\frac{1}{q}\mathbf{E}_{X_M}\ln|\det \partial_{X_M}\mu_M|-\frac{1}{p}\mathbf{E}_{X_m}\ln|\det \partial_{X_m}\mu_m|\right)\\
        &-\left(\frac{1}{q}\sum_{i=1}^{q}\ln\sigma_{i,M}^2-\frac{1}{p}\sum_{i=1}^{p}\ln\sigma_{i,m}^2\right),
    \end{aligned}
\end{equation}
where $\sigma_{i,M}$ and $\sigma_{i,m}$ are the standard deviations for $\mu_M$ and $\mu_m$ on the $i$th dimension, respectively. At this point, the influences of the input or output dimensions and the parameter $L$ has been completely eliminated, making it a more reliable indicator. Therefore, the comparisons between the learned macro-dynamics reported in the main text are reliable because the measure of dimension averaged CE is used.

In practice, the first step is to select a sufficiently large value for $L$ to ensure that all macro-states are encompassed within the region $[-L, L]^q$. Subsequently, the determinant of Jacobian matrices can be estimated using Monte Carlo integration, which involves sampling data points within this region. The standard deviations can be estimated using a test dataset, and the probability reweighting technique is employed to account for interventions and ensure a uniform distribution.

\section{Invertible neural network in encoder(decoder)}
In NIS, both the encoder($\phi$) and the decoder($\phi^{\dag}$) use the invertible neural network module RealNVP\cite{Dinh_Sohl-Dickstein_Bengio_2016}. 

\begin{figure}
    \centering
    \includegraphics[width=0.8\textwidth]{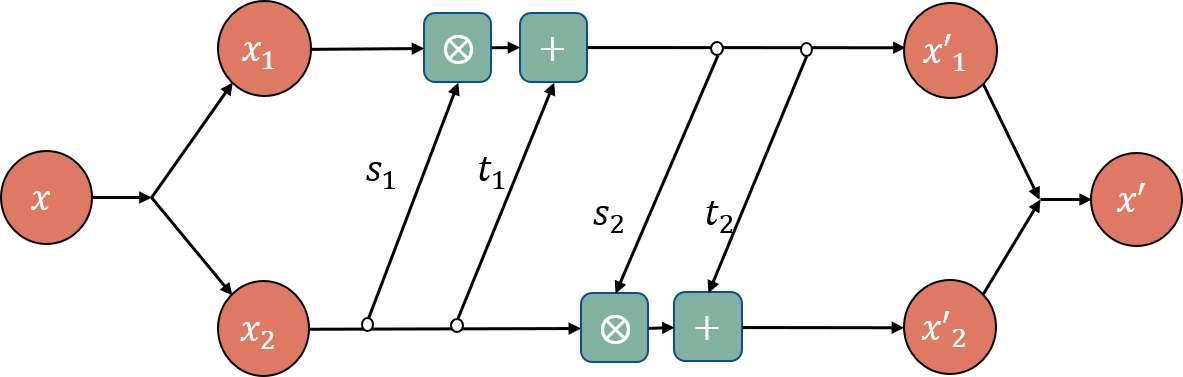}
    \caption{Structure diagram of RealNVP.}
    \label{fig:realnvp}
\end{figure}

Concretely, if the input of the module is $\mathbf{x}$ with dimension $p$ and the output is $\mathbf{x'}$ with the same dimension, then the RealNVP module can perform the following computation steps:

\begin{equation}
    \label{eq.x1x2}
    \left\{
    \begin{aligned}
        &\mathbf{x_1} = \mathbf{x}_{1:m}\\
        &\mathbf{x_2} = \mathbf{x}_{m:p}
    \end{aligned}
    \right.
\end{equation}
where, $m$ is an integer in between 1 and $p$.
\begin{equation}
    \label{eq:sandt}
    \left\{
    \begin{aligned}
        &\mathbf{x'_1} = \mathbf{x_1}\bigotimes s_1(\mathbf{x_2})+t_1(\mathbf{x_2})\\
        &\mathbf{x'_2} = \mathbf{x_2}\bigotimes s_2(\mathbf{x'_1})+t_2(\mathbf{x'_1})
    \end{aligned}
    \right.
\end{equation}
where, $\bigotimes$ is element-wised time, $s_1,s_2,t_1$ and $t_2$ are feed-forward neural networks with arbitrary architectures, while their input-output dimensions must match with the data. For the experiments in this article, they each have two intermediate hidden layers. The input and output layers have the same number of neurons as the dimensions of the micro-state samples. Each hidden layer has 64 neurons, and the output of each hidden layer is transformed by the non-linear activation function LeakyReLU. In practice, $s_1$ or $s_2$ always do an exponential operation on the output of the feed-forward neural network to facilitate the inverse computation. 

Finally,
\begin{equation}
    \mathbf{x'}=\mathbf{x'_1}\bigoplus \mathbf{x'_2}
\end{equation}
It is not difficult to verify that all three steps are invertible. Equation \ref{eq:sandt} is invertible because the same form but with negative signs can be obtained by solving the expressions of $\mathbf{x_1}$ and $\mathbf{x_2}$ with $\mathbf{x'_1}$ and $\mathbf{x'_2}$ from Equation \ref{eq:sandt}.

To simulate more complex invertible functions, we always duplex the basic RealNVP modules by stacking them together. In the main text, we use duplex the basic RealNVP module by three times.

Due to the reversibility of RealNVP, if its output is used as input in a subsequent iteration, the new output will be the same as the original input. Therefore, the neural networks used in the encoder-decoder can share parameters. The process of encoding to obtain the macro-state involves projecting the forward output and retaining only the first few dimensions. On the other hand, the process of decoding to obtain the micro-state involves concatenating the macro-state with a normal distribution noise, expanding it to match the micro-state dimensions, and then feeding it back into the entire invertible neural network.

\section{KDE for probability reweighting}
\label{sec:KDE}
In order to use inverse probability weighting technique, we need to estimate the probability distribution of the samples. KDE(Kernel Density Estimation) is a commonly used estimation method that can effectively eliminate the influence of outliers on the overall probability distribution estimation. In this experiment, we estimate the probability distribution of macro-states $\boldsymbol{y}_t$. Below is an introduction to the operation details of KDE. For a sample $(x_1,x_2,\cdot\cdot\cdot,x_n)$, we have the following kernel density estimation:
\begin{equation}
    \label{eq.kde}
    \begin{aligned}
        \hat{f}_h(x)=\frac{1}{nh}\sum_{i=1}^n K(\frac{x-x_i}{h}).
    \end{aligned}
\end{equation}
$n$ represents the number of samples, the hyperparameter $h$ denotes the bandwidth, which is determined based on the rough range of the data and typically set to 0.05 in this article. $K$ is the kernel, specifically the standard normal density function. After obtaining the estimation function, each sample point is evaluated individually, resulting in the corresponding probability value for each sample point. By dividing the probability of the target distribution (uniform distribution) by our estimated probability value, we obtained the inverse probability weights corresponding to each sample point. To cover all sample points, the range of the uniform distribution needs to be limited by a parameter $L$, which ensures that a square with side length $2L$ can encompass all the sample points in all dimensions.

Another thing to note is that the weights obtained at this point are related to the sample size, so they need to be normalized. However, this will result in very small weight values corresponding to the samples. Therefore, it is necessary to multiply them by the sample quantity to amplify the weight values back to their normal scale. Then, multiply this weight value by the training loss of each sample to enhance the training of sparse regions, achieving the purpose of inverse probability weighting. Since the encoder parameters change with each iteration, causing the distribution of $\boldsymbol{y}_t$ to change, we will re-estimate the probability distribution of the entire sample using KDE every 1000 epochs (a total of at least 30000 epochs of iteration, the number of iterations may vary in different experiments).

\section{Calculation method of Rosas' $\Psi$}
\label{sec:Rosas psi}
We present here the computing method for Rosas' $\Psi$, a measurable indicator of CE proposed by Rosas~\cite{Rosas_Mediano_Jensen_Seth_Barrett_Carhart-Harris_Bor_2020}. Its definition is as follows:
\begin{equation}
    \label{eq.psi}
    \begin{aligned}
        \Psi(Y)=I(Y_t,Y_{t'})-\sum_j I(X_t^j,Y_{t'}),
    \end{aligned}
\end{equation}
where $Y$ represents the given macrostate, $Y_t$ and $Y_{t'}$ represent the macrostates at two consecutive time points, and $X_t^j$ represents the $j$-th dimension of the microstate at time $t$. According to Rosas' theory, Rosas' $\Psi$ being greater than 0 is a sufficient but not necessary condition for emergence to occur. In this paper, we employ the method proposed by Lu et al.~\cite{lu2020enhancing} to compute the continuous variable mutual information in the experimental configuration.

\section{Comparison models in the SIR experiment}
\label{sec:nn vae}
To compare the advantages of the NIS+ framework itself, we constructed a regular fully connected neural network (NN) with 33,668 parameters and a Variational Autoencoder (VAE) with 34,956 parameters in the SIR experiment (the NIS+ model consists of a total of 37,404 parameters, which includes the inverse dynamics component). NN consists of an input layer, an output layer, and four hidden layers with 4, 64, 128, 128, 64, 4 neurons respectively. NN+ refers to the NN that combines inverse probability weighting technique, with the same parameter framework. For each data point $\boldsymbol{x}_t$, we estimate the KDE probability distribution and calculate the inverse probability weighting.  Finally, the weights are multiplied with the loss to give different training weights to samples with different densities. 

VAE is a generative model that maps micro-level data to a normal distribution in latent space. During each prediction, it samples from the normal distribution and then decodes it~\cite{kingma2013auto}. It consists of an encoder and a decoder, both of which have 4 hidden layers, each layer containing 64 neurons. The input and predicted output are both four-dimensional variables, and the latent space also has four dimensions, with two dimensions for mean and two dimensions for variance. The macro-dynamic learner fits the dynamical changes of the two-dimensional mean variables. During the training process, the loss function consists of two parts: reconstruction error, which is the MAE error between the predicted output and the true label, and the KL divergence between the normal distribution of the latent space samples and the standard normal distribution. These two parts are combined with a 2:1 ratio and used for gradient backpropagation. VAE+ combines inverse probability weighting and bidirectional dynamic learning in the training of VAE, and adopts the same training method as NIS+. We keep the variance variables obtained from encoding unchanged, and use inverse probability weighting to obtain weight factors for the mean variables. Additionally, a reverse dynamic learner is constructed to predict $\boldsymbol{y}_{t}$ given $\boldsymbol{y}_{t+1}$. 

Because NN+ does not have a multi-scale framework for encoding and decoding, it directly estimates and calculates inverse probability weights at the micro-scale data level, and cannot improve the forward dynamics EI through training reverse dynamics. In VAE+, we focus on the dynamics trained in the latent space, aiming to estimate and calculate inverse probability weights in the latent space, and can optimize the encoder by training the reverse dynamics in the latent space to improve EI.The comparison results are shown in Figure \ref{fig:sir}(b) and (d).

\section{Integrated Gradients for attribution}
\label{sec:IG}
Here we introduce a method called Integrated Gradients\cite{Sundararajan_Taly_Yan_2017} that is used to explain the connection between macro-states and micro-states. To formally describe this technique, let us consider a function $F: \mathcal{R}^p \to [0,1]$ that represents a deep network. Specifically, let $x \in \mathcal{R}^p$ denote the input under consideration, and $x' \in \mathcal{R}^p$ denote the baseline input. In both the fMRI and Boids experiments, the baseline is set to 0.

We compute the gradients at each point along a straight-line path (in $\mathcal{R}^p$) from the baseline $x'$ to the input $x$. By accumulating these gradients to obtain the final measure of integrated gradients(IG). More precisely, integrated gradients are defined as the path integral of the gradients along the straight-line path from the baseline $x'$ to the input $x$.

The integrated gradient along the i-th dimension for an input x and baseline $x'$ is defined as follows. Here, $\partial F(x)/\partial x_i$ represents the gradient of $F(x)$ along the $i$-th dimension. The Integrated Gradients(IG) that we have computed are shown in Equation \ref{eq.inte_gra}.

\begin{equation}
    \label{eq.inte_gra}
    \begin{aligned}
        IG_i(x)=(x_i-x_i')\times \int^1_{\alpha=0}\frac{\partial F(x'+\alpha\times(x-x'))}{\partial x_i} \mathrm{d}\alpha.
    \end{aligned}
\end{equation}

\section{fMRI time series data description in detail}
\label{fMRI_preprocessing}

AOMIC is an fMRI collection which composes AOMIC PIOP1, AOMIC PIOP2 and AOMIC ID1000\cite{Snoek2021}. 

The AOMIC PIOP2 collects subjects' data for mutiple tasks such as emotion-matching, working-memory and so on. Here, we just use 50 subjects' resting fMRI data since some time steps of other subject's have been thrown outby removing the effects of artificial motion, global signal, and white matter signal and cerebrospinal fluid signal using fMRIPrep results\cite{Esteban2019,Esteban2020}, which means a difficulty in time alignment.

The AOMIC ID 1000 collected when 881 subjects are watching movies. It contains both raw data and preprocsessed data(See more experimental and preprocsessing details in \cite{Snoek2021}). Here, we should notice that the movie is edited in a way for no clear semantic meanings but just some concatenated images from the movie. Therefore, it is expected that subjects' brain activation patterns shouldn't respond to some higher-order functions such as semantic understandings.

When working with fMRI data, the high dimensionality of the data (over 140,000 dimensions) poses computational challenges. To address this, dimension reduction techniques through brain atlas are employed to make the calculations feasible. We first preprocessed the data by removing the effects of artificial motion, global signal, and white matter signal and cerebrospinal fluid signal using fmriPrep results contained in AOMIC ID1000 \cite{Esteban2019,Esteban2020}. Simultaneously, we detrend and normalize the data during this process. By doing the whole preprocessing process, we also throw out 51 subjects' data since some time steps of these subject's have been removed. Furthermore, to facilitate the investigation of the correlation between the required brain function for watching movie clips and the emergent macro dynamics, we employed the Schaefer atlas to divide the brain into 100 functional areas. Also, we have committed the philosophy of anatomical Atlas, which seems to suggest the converging evidence(See support information section \ref{sec:AAL Results}). All of these preprocessing steps were performed using Nilearn\cite{scikit-learn,Abraham2014}. During the training process, the only difference is the use of PCA in dimension reduction to handle the KDE approximation required for reweighting if $q>10$.

\section{Preceding instructions}
Before we prove the theorem, let us first introduce the symbols that will be used in the following paragraphs. We will use capital letters to represent the corresponding random variables. For example, $X_t$ represents the random variable of the micro-state $\boldsymbol{x}_t$ at time $t$, and $Y_{t+1}$ represents the random variable corresponding to the macro-state $\boldsymbol{y}_{t+1}$. For any random variable $X$, $\Tilde{X}$ represents the same random variable $X$ after intervention. $\hat{X}$ means the prediction for $X$ given by neural networks.

Second, to understand how the intervention on $Y_t$ can affect other variables, a causal graph derived from the framework of NIS and NIS+ and depicts the causal relations among variables is very useful, this is shown in Figure \ref{fig:causal graph}. In which, $X_t$ and $\hat{X}_{t+1}$ denote the input random variable and predicted output random variable of the NIS+ on a micro scale, while $Y_t$ and $\hat{Y}_{t+1}$ represent the input and output of $f_{q}$ on a macro scale. After intervening on $Y_t$, $\Tilde{X}_t$, $\Tilde{X}_{t+1}, \Tilde{Y}_t, \Tilde{Y}_{t+1}$ represent different micro or macro variables under the intervention $do(Y_t\sim U(Y))$, corresponding to $X_t$,$\hat{X}_{t+1}, Y_t$, and $\hat{Y}_{t+1}$, respectively.

\begin{figure}
    \centering
    \includegraphics[width=0.8\textwidth]{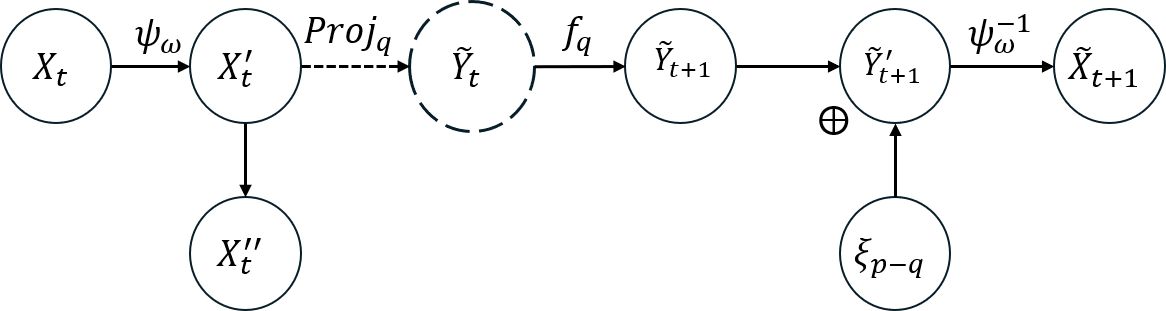}
    \caption{The causal graph among random variables after intervention on $Y_t$ according to the framework of NIS or NIS+. Because the intervention on $Y_t$ can only affect the variables in the upper part of NIS+ framework, we ignore the variables in the lower part. In the diagram, $X_t'$ represents the random variable obtained after reversible transformation of $X_t$, $X''_t$ represents the variable directly discarded during the projection process, $\Tilde{Y}_{t+1}'$ represents a new variable composed of $\Tilde{Y}_{t+1}$ concatenated with a standard normal distribution, and $\xi_{p-q}$ represents a $p-q$ dimensional standard normal distribution. The dashed circular shape in the diagram represents the variable that is directly intervened to a uniform distribution, and the dashed arrow represents the causal relationship that is severed due to the intervention.}
    \label{fig:causal graph}
\end{figure}

\section{Mathematical Proves}
In this sub-section, we will show the mathematical proves of the theorems mentioned in the main text. 

\subsection{Proof of Theorem \ref{thm:optimize}}
\label{sec:optimize}
To prove this theorem, we utilize a combination of the inverse macro-dynamic and probability re-weighting technique, along with information theory and properties of neural networks. In this proof, we need three lemmas, stated as follows:
\begin{lem}[Bijection mapping does not affect mutual information]
For any given continuous random variables $X$ and $Z$, if there is a bijection (one to one) mapping $f$ and another random variable $Y$ such that for any $x\in Dom (X)$ there is a $y=f (x)\in Dom (Y)$, and vice versa, where $Dom (X)$ denotes the domain of the variable $X$, then the mutual information between $X$ and $Z$ is equal to the information between $Y$ and $Z$, that is:
\begin{equation}
 I (X;Z)=I (Y;Z),
\end{equation}
\label{lem.invertible}
\end{lem}

\begin{lem}[Mutual information will not be affected by concatenating independent variables]
\label{lemma.concatenation}
If $X\in Dom(X)$ and $Y\in Dom(Y)$ form a Markov chain $X\rightarrow Y$, and $Z\in Dom(Z)$ is a random variable which is independent on both $X$ and $Y$, then:
\begin{equation}
    I(X;Y)=I(X;Y\bigoplus Z).
\end{equation}
\end{lem}

The proofs of Lemma\ref{lemma.concatenation} and Lemma\ref{lem.invertible} can be found in the reference \cite{zhang2022neural}.

\begin{lem}[variational upper bound of a conditional entropy]
\label{lemma.var_upper_bound}
Given a conditional entropy $H(\boldsymbol{y}|\boldsymbol{x})$, where $\boldsymbol{x}\in \mathcal{R}^s,\boldsymbol{y}\in \mathcal{R}^q$, there exists a variational upper bound on this conditional entropy:
\begin{equation}
    H(Y|X)\le -\iint p(\boldsymbol{y}, \boldsymbol{x})\ln g(\boldsymbol{y}|\boldsymbol{x}) \mathrm{d}\boldsymbol{y} \mathrm{d}\boldsymbol{x},
\end{equation}
where $g(\boldsymbol{y}|\boldsymbol{x}) \in \mathcal{R}^q \times \mathcal{R}^s$ is any distribution.
\end{lem}

\begin{proof}
First, we unfold the conditional entropy
\begin{gather}
    H(Y|X)=-\iint p(\boldsymbol{x})p(\boldsymbol{y}|\boldsymbol{x}) \ln p(\boldsymbol{y}|\boldsymbol{x}) \mathrm{d}\boldsymbol{y} \mathrm{d}\boldsymbol{x}
\end{gather}

Due to the property of KL divergence\cite{Fang_Zhu_2020}, for any distribution $g$, 
\begin{gather}
    D_{KL}(p||g)=\int p(\boldsymbol{x})\ln \frac{p(\boldsymbol{x})}{g(\boldsymbol{x})}\mathrm{d}\boldsymbol{x} \ge 0.
\end{gather}
On the other words, 
\begin{gather}
    \int p(\boldsymbol{x})\ln p(\boldsymbol{x})\mathrm{d}\boldsymbol{x} \ge \int p(\boldsymbol{x})\ln g(\boldsymbol{x})\mathrm{d}\boldsymbol{x}.
\end{gather}
So, 
\begin{equation}
\begin{aligned}
    H(\boldsymbol{y}|\boldsymbol{x})&=-\int p(\boldsymbol{x})\int p(\boldsymbol{y}|\boldsymbol{x}) \ln p(\boldsymbol{y}|\boldsymbol{x}) \mathrm{d}\boldsymbol{y}\mathrm{d}\boldsymbol{x}\\
    &\le -\int p(\boldsymbol{x})\int p(\boldsymbol{y}|\boldsymbol{x}) \ln g(\boldsymbol{y}|\boldsymbol{x}) \mathrm{d}\boldsymbol{y} \mathrm{d}\boldsymbol{x}\\
    &=-\iint p(\boldsymbol{y}, \boldsymbol{x})\ln g(\boldsymbol{y}|\boldsymbol{x}) \mathrm{d}\boldsymbol{y} \mathrm{d}\boldsymbol{x}.
\end{aligned}
\end{equation}
\end{proof}

To prove the theorem, we also use an assumption:

\textbf{Assumption:} The composition of the inverse dynamics $g_{\theta '}$ and the encoder $\phi$ can be regarded as a conditional probability $P(\hat{Y}_{t}|X_{t+1})$, and this probability can be approximated as a Gaussian distribution $N(g_{\theta'}(\phi(\boldsymbol{x}_{t+1})),\Sigma)$, where $\Sigma=diag(\sigma_1, \sigma_2,\cdot\cdot\cdot,\sigma_q)$, and $\sigma_i$ is the MSE loss of the $i$th dimension of output $\hat{Y}_{t+1}$. Further, suppose $\sigma_i$ is bounded, thus $\sigma_i\in[\sigma_m,\sigma_M]$ for any $i$, where $\sigma_m$ and $\sigma_M$ are the minimum and maximum values of MSEs.

This assumption is supported by the reference \cite{Blundell_Cornebise_Kavukcuoglu_Wierstra_2015}.

Next, we restate the theorem that needs to be proved:

\textbf{Theorem} \ref{thm:optimize}(Problem Transformation Theorem):
For a given value of $q$, assuming that $\omega^\ast, \theta^\ast$, and $\theta'^\ast$ are the optimal solutions to the unconstrained objective optimization problem defined by equation (\ref{new optimization}). Then $\phi^\ast\equiv Proj_q(\psi_{\omega^\ast}),f^\ast_q\equiv f_{\theta^\ast}$, and $\phi^{\dag,\ast}(\cdot)\equiv \psi^{-1}_{\omega^\ast}(\cdot\bigoplus \xi)$, where $\xi\sim \mathcal{N}(0,I_{p-q})$ are the optimal solution to the constrained objective optimization problem Eq.(\ref{old optimization}).

\begin{proof}
The original constrained goal optimization framework is as follows:
\begin{equation}
\begin{aligned}
&\max_{\phi,f_q,\phi^+} \mathcal{J}(f_q),\\
&s.t. \begin{cases}
|| \hat{\boldsymbol{x}}_{t+1}-\boldsymbol{x}_{t+1} || < \epsilon,\\
\hat{\boldsymbol{x}}_{t+1}=\phi^{\dag}(f_q(\phi(\boldsymbol{x}_t))).
\end{cases}
\end{aligned}
\end{equation}

We know that $\hat{X}_{t+1}=\psi_{\omega}^{-1}(\hat{Y}_{t+1}\bigoplus \xi)$, where $\psi_{\omega}^{-1}$ is a reversible mapping that doesn't affect the mutual information according to Lemma \ref{lem.invertible}. Therefore, based on Lemma \ref{lemma.concatenation}, we have the capability to apply a transformation to the mutual information $I(Y_t,\hat{Y}_{t+1})$:
\begin{gather}
\label{eq:mi tran}
    I(Y_t,\hat{Y}_{t+1})=I(Y_t,\hat{X}_{t+1}).
\end{gather}
Here, $I$ represents mutual information. By utilizing the property of mutual information, we can derive the following equation:
\begin{gather}
    I(Y_t,\hat{X}_{t+1})=H(Y_t)-H(Y_t|\hat{X}_{t+1}).
\end{gather}

Now, $H(\Tilde{Y}_t)=H(U_q)$, where $U_q$ is a uniform distribution on a macro space. Therefore, we have:
\begin{gather}
\label{eq:object to H}
    \mathcal{J}(f_{\theta,q})=H(U_q)-H(\Tilde{Y}_t|\Tilde{X}_{t+1}).
\end{gather}
The optimization of the objective function $\mathcal{J}(f_{q})$ can be reformulated as the optimization of the conditional entropy $H(\Tilde{Y}_t|\Tilde{X}_{t+1})$, since $H(U_q)$ is a constant. The variational upper bound on $H(\Tilde{Y}_t|\Tilde{X}_{t+1})$ is obtained by Lemma\ref{lemma.var_upper_bound}.
\begin{gather}
\label{eq:upper_bound}
    H(\Tilde{Y}_t|\Tilde{X}_{t+1}) \le -\iint \Tilde{p}(\boldsymbol{y}_t, \boldsymbol{x}_{t+1})\ln g(\boldsymbol{y}_t|\boldsymbol{x}_{t+1}) \mathrm{d}\boldsymbol{y}_t \mathrm{d}\boldsymbol{x}_{t+1},
\end{gather}
where $\Tilde{p}$ represents the probability distribution function of random variables in case of $\boldsymbol{y}_t$ being intervened.

We will use a neural network to fit the distribution $g(\boldsymbol{y}_t|\boldsymbol{x}_{t+1})$. Based on the assumption, we can consider it as a normal distribution. According to Lemma\ref{lemma.var_upper_bound}, since the conditional probability $g(\boldsymbol{y}_t|\boldsymbol{x}_{t+1})$ can be any distribution, we can assume it to be a normal distribution for simplicity. Predicting $\boldsymbol{y}_t$ with $\boldsymbol{x}_{t+1}$ as input can be divided into two steps. First, $\boldsymbol{x}_{t+1}$ is encoded by the encoder $\phi$ into the macro latent space. Then, the reverse macro dynamics are approximated using $g_{\theta'}$. As a result, the expectation of $g(\boldsymbol{y}_t|\boldsymbol{x}_{t+1})$ can be obtained using the following equation:
\begin{equation}
\begin{aligned}
\label{eq:expect_g}
    E_{g}(\Tilde{Y}_t|X_{t+1}=\boldsymbol{x}_{t+1})
    \equiv  g_{\theta'}(\phi(\boldsymbol{x}_{t+1}))
\end{aligned}
\end{equation}

Therefore, we have $g(\boldsymbol{y}_t|\boldsymbol{x}_{t+1})\sim N(\mu,\Sigma)$, where $\Sigma$ is a constant diagonal matrix and $\mu =g_{\theta'}(\phi(\boldsymbol{x}_{t+1}))$. In order to utilize the properties of intervention on $\boldsymbol{y}_t$, we need to separate $\Tilde{p}(\boldsymbol{y}_t)$. Following Equation \ref{eq:upper_bound}, we can further transform $H(\Tilde{Y}_t|\Tilde{X}_{t+1})$:

\begin{gather}
     H(\Tilde{Y}_t|\Tilde{X}_{t+1})\le -\iint \Tilde{p}(\boldsymbol{y}_{t})\Tilde{p}(\boldsymbol{x}_{t+1}|\boldsymbol{y}_t)\ln  g(\boldsymbol{y}_t|\boldsymbol{x}_{t+1}) \mathrm{d}\boldsymbol{y}_t \mathrm{d}\boldsymbol{x}_{t+1}.
\end{gather}

Assuming that the training is sufficient, we can have
\begin{gather}
    \Tilde{p}(\boldsymbol{x}_{t+1}|\boldsymbol{y}_t)\approx p(\boldsymbol{x}_{t+1}|\boldsymbol{y}_t).
\end{gather}

So, 
\begin{gather}
     H(\Tilde{Y}_t|\Tilde{X}_{t+1})\le -\iint \Tilde{p}(\boldsymbol{y}_{t})p(\boldsymbol{x}_{t+1}|\boldsymbol{y}_t)\ln  g(\boldsymbol{y}_{t}|\boldsymbol{x}_{t+1}) \mathrm{d}\boldsymbol{y}_t \mathrm{d}\boldsymbol{x}_{t+1}.
\end{gather}

According to Equation\ref{eq:expect_g}, we obtain the logarithm probability density function of \\
$g(\boldsymbol{y}_{t}|\boldsymbol{x}_{t+1})$:
\begin{equation}
\begin{aligned}
\label{eq:ln_g}
    \ln g(\boldsymbol{y}_t|\boldsymbol{x}_{t+1})
    &\approx \ln \frac{1}{(2\pi)^{\frac{m}{2}}|\Sigma|^\frac{1}{2}} e^{-\frac{(\boldsymbol{y}_t-g_{\theta'}(\phi(\boldsymbol{x}_{t+1})))^2}{2|\Sigma|}}\\
    &=-\frac{(\boldsymbol{y}_t-g_{\theta'}(\phi(\boldsymbol{x}_{t+1})))^2}{2|\Sigma|}+\ln \frac{1}{(2\pi)^{\frac{m}{2}}|\Sigma|^\frac{1}{2}}.
\end{aligned}
\end{equation}

Because $\ln \frac{1}{(2\pi)^{\frac{m}{2}}|\Sigma|^\frac{1}{2}}\geq \ln \frac{1}{(2\pi)^{\frac{m}{2}}|\Sigma|_{max}^\frac{1}{2}}$, so
\begin{equation}
\begin{aligned}
\label{eq:inequality_iint}
        H(\Tilde{Y}_t|\Tilde{X}_{t+1}) &\le \iint \Tilde{p}(\boldsymbol{y}_{t})p(\boldsymbol{x}_{t+1}|\boldsymbol{y}_t)\left[\frac{(\phi (\boldsymbol{x}_t)-g_{\theta'}(\phi(\boldsymbol{x}_{t+1})))^2}{2|\Sigma|}-\ln \frac{1}{(2\pi)^{\frac{m}{2}}|\Sigma|^\frac{1}{2}}\right] \mathrm{d}\boldsymbol{y}_t \mathrm{d}\boldsymbol{x}_{t+1}\\
        &\le \iint \Tilde{p}(\boldsymbol{y}_{t})p(\boldsymbol{x}_{t+1}|\boldsymbol{y}_t)\left[\frac{(\phi (\boldsymbol{x}_t)-g_{\theta'}(\phi(\boldsymbol{x}_{t+1})))^2}{2|\Sigma|_{min}}-\ln \frac{1}{(2\pi)^{\frac{m}{2}}|\Sigma|_{max}^\frac{1}{2}}\right] \mathrm{d}\boldsymbol{y}_t \mathrm{d}\boldsymbol{x}_{t+1},
\end{aligned}
\end{equation}

where $|\Sigma|_{min}=\sigma_{min}^q, |\Sigma|_{max}=\sigma_{max}^q$. Here, $\Tilde{p}(\boldsymbol{y}_{t})p(\boldsymbol{x}_{t+1}|\boldsymbol{y}_t)=\frac{\Tilde{p}(\boldsymbol{y}_{t})}{p(\boldsymbol{y}_{t})}p(\boldsymbol{x}_{t+1},\boldsymbol{y}_t)$. Thus, we obtain $\frac{\Tilde{p}(\boldsymbol{y}_{t})}{p(\boldsymbol{y}_{t})}$, where $\Tilde{p}(\boldsymbol{y}_{t})$ represents the target distribution and $p(\boldsymbol{y}_{t})$ represents the natural distribution obtained from the data. We can define the inverse probability weights $w(\boldsymbol{x}_t)$ as follows:
\begin{gather}
\label{eq:defi_W}
    w(\boldsymbol{x}_t) \equiv \frac{\Tilde{p}(\boldsymbol{y}_{t})}{p(\boldsymbol{y}_{t})}.
\end{gather}

Let $z=\frac{(\phi (\boldsymbol{x}_t)-g_{\theta'}(\phi(\boldsymbol{x}_{t+1})))^2}{2|\Sigma|_{min}}-\ln \frac{1}{(2\pi)^{\frac{m}{2}}|\Sigma|_{max}^\frac{1}{2}}$. Therefore, we can express Equation \ref{eq:inequality_iint} as follows:
\begin{equation}
\begin{aligned}
\label{eq:inequality_iint2}
        H(\Tilde{Y}_t|\Tilde{X}_{t+1}) \le \iint w(\boldsymbol{x}_t)p(\boldsymbol{x}_{t+1},\boldsymbol{y}_t)z \mathrm{d}\boldsymbol{y}_t \mathrm{d}\boldsymbol{x}_{t+1}.
\end{aligned}
\end{equation}

Because we train neural networks using discrete samples $\{x_t\}$, we can use the sample mean as an approximate estimate of the expectation in Equation\ref{eq:inequality_iint2}. Therefore, the variational upper bound of $H(\Tilde{Y}_t|\Tilde{X}_{t+1})$ can be written as:
\begin{equation}
\begin{aligned}
\label{eq:inequality_sum}
        H(\Tilde{Y}_t|\Tilde{X}_{t+1}) \le \frac{1}{T}\sum_{i=0}^{T-1}w(\boldsymbol{x}_t)z.
\end{aligned}
\end{equation}

Substituting the equation into Equation \ref{eq:object to H}, we obtain the variational lower bound of the original objective function:
\begin{equation}
\begin{aligned}
\label{eq:var_lower_bound}
    \mathcal{J}(f_{\theta,q}) \ge H(U_q)-\frac{1}{T}\sum_{i=0}^{T-1}w(\boldsymbol{x}_t)z.
\end{aligned}
\end{equation}

Therefore, the optimization problem(Equation \ref{old optimization}) is transformed into
\begin{gather}
    \min_{\omega,\theta,\theta'} \sum_{i=0}^{T-1}w(\boldsymbol{x}_t)|\phi(\boldsymbol{x}_t)-g_{\theta'}(\phi(\boldsymbol{x}_{t+1}))|^2 \\
    s.t. || \hat{\boldsymbol{x}}_{t+1}-\boldsymbol{x}_{t+1} || < \epsilon,
\end{gather}
where $\omega$, $\theta$, $\theta'$ respectively represent the parameters of the three neural networks $\psi$, $f_{\theta}$, $g_{\theta'}$ in the NIS+ framework.

Then we construct Lagrange function,
\begin{gather}
    L(\omega,\theta,\theta',\lambda)=\sum_{i=0}^{T-1}w(\boldsymbol{x}_t)|\phi(\boldsymbol{x}_t)-g_{\theta'}(\phi(\boldsymbol{x}_{t+1}))|^2+\lambda|| \phi^\dag(\boldsymbol{y}_{t+1})-\boldsymbol{x}_{t+1} ||
\end{gather}

Therefore, the optimization goal is transformed into
\begin{gather}
    \min_{\omega,\theta,\theta'} \sum_{i=0}^{T-1}w(\boldsymbol{x}_t)||\boldsymbol{y}_t-g_{\theta'}(\boldsymbol{y}_{t+1})||+\lambda|| \hat{\boldsymbol{x}}_{t+1}-\boldsymbol{x}_{t+1} ||
\end{gather}
\end{proof}

\subsection{Proof of Theorem \ref{thm:optimize_extension}}
\label{sec:op extension}
We will first use two separate lemmas to prove the scalability of stacked encoder and parallel encoder, thereby demonstrating that their arbitrary combination can keep the conclusion of Theorem \ref{thm:optimize} unchanged.

\begin{lem}[Mutual information will not be affected by stacked encoder]
\label{lemma.concatenation_stacked}
If $X\in Dom(X)$ and $Y\in Dom(Y)$ form a Markov chain $X\rightarrow Y$, and $\Phi_L$ and $\Phi^\dag_L$ represent the L-layer stacked encoder and decoder, respectively, then:
\begin{equation}
    I(X;Y)=I(X;\Phi^\dag_L(Y)).
\end{equation}
\end{lem}
\begin{proof}
According to Figure\ref{fig:structures}(a), we can obtain
\begin{equation}
    \Phi^\dag_{L}(Y) =  \phi^\dag_{q_1}\circ \phi^\dag_{q_2}\circ\cdot\cdot\cdot\circ \phi^\dag_{q_L}(Y).
\end{equation}
$\phi^\dag_{q_i}(i=1, 2,\cdot\cdot\cdot, L):\mathcal{R}^{q_i}\rightarrow \mathcal{R}^{q_{i-1}}$ is a basic decoder as shown in Equation \ref{eq:phidag}. $\circ$ represents the composition of functions. Therefore, according to Lemma \ref{lemma.concatenation} and the fact that reversible mappings do not change mutual information, we obtain
\begin{equation}
    I(X;Y)=I(X;\phi^\dag_L(Y)).
\end{equation}
Let $Y_{L-1}=\phi^\dag_L(Y)$, we can further obtain $I(X;Y_{L-1})=I(X;\phi^\dag_{L-1}(Y_{L-1}))$, and so on, leading to the final result
\begin{equation}
    I(X;Y)=I(X;\Phi^\dag_L(Y)).
\end{equation}
\end{proof}

\begin{lem}[Mutual information will not be affected by parallel encoder]
\label{lemma.concatenation_parallel}
If $X\in Dom(X)$ and $Y\in Dom(Y)$ form a Markov chain $X\rightarrow Y$, and $\Phi_T$ and $\Phi^\dag_T$ represent parallel encoder and decoder composed of T ordinary encoders or decoders, respectively, then:
\begin{equation}
    I(X;Y)=I(X;\Phi^\dag_T(Y)).
\end{equation}
\end{lem}
\begin{proof}
the decoding process can be divided into two steps:
\begin{equation}
    \Phi^\dag_T(Y) = \Psi_T(Y\bigoplus\xi).
\end{equation}
$Y$ can be decomposed as $Y = Y_1 \bigoplus Y_2 \bigoplus \dots \bigoplus Y_T$, and all the introduced noise $\xi = \xi_1 \bigoplus \xi_2 \\
\bigoplus \dots \bigoplus \xi_T$, even if the order of their concatenation changes, it will not affect the mutual information between the overall variable and other variables. Let $\Psi_{T}(Z) = \psi_{1}(Z_1) \bigoplus \psi_{2}(Z_2) \bigoplus \dots \\
\bigoplus \psi_{T}(Z_T)$, where $Z_i (i=1, 2, \dots, T)$ is a partition of $Z$, and $\psi_{i} (i=1, 2, \dots, T): \mathcal{R}^{p} \rightarrow \mathcal{R}^{p}$ are all invertible functions. Therefore, $\Psi_{T}$ is also an invertible function, meaning that this function transformation does not change the mutual information. According to Lemma \ref{lemma.concatenation}, we have
\begin{equation}
    I(X;Y)=I(X;Y\bigoplus\xi)=I(X;\Psi_T(Y\bigoplus\xi)).
\end{equation}
Finally proven,
\begin{equation}
    I(X;Y)=I(X;\Phi^\dag_T(Y)).
\end{equation}
\end{proof}

Next, we provide the proof of Theorem \ref{thm:optimize_extension}.

\textbf{Theorem \ref{thm:optimize_extension}}(Problem Transformation in Extensions of NIS+ Theorem)
When the encoder of NIS+ is replaced with an arbitrary  combination of stacked encodes and parallel encoders, the conclusion of Theorem 2.1 still holds true.

\begin{proof}
To demonstrate that Theorem \ref{thm:optimize} remains applicable to the extended NIS+ framework, we only need to prove that Equation \ref{eq:mi tran} still holds true, as the only difference between the extended NIS+ and the previous framework lies in the encoder. First, we restate Equation \ref{eq:mi tran} as follows:
\begin{gather}
     \mathcal{J}(f_{q}) =I(Y_t,\hat{Y}_{t+1})=I(Y_t,\hat{X}_{t+1}).
\end{gather}
According to Lemma \ref{lemma.concatenation_stacked} and Lemma \ref{lemma.concatenation_parallel}, Equation \ref{eq:mi tran} holds true when the encoder of NIS+ is replaced with either a stacked encoder or a parallel encoder. And any combination of these two encoders is nothing more than an alternating nesting of stacking and parallelization, so the conclusion still holds.
\end{proof}

\subsection{Proof of Theorem \ref{thm.universal}}
\label{sec:universal}
To prove the universality theorem, we need to extend the definition of the basic encoder by introducing a new operation $\eta_{p,s}: \mathcal{R}^p\rightarrow \mathcal{R}^s$, which represents the self-replication of the original variables.
\begin{equation}
    \eta_{p,s}(\boldsymbol{x})=\boldsymbol{x}\bigoplus \boldsymbol{x}_{s-p}.
\end{equation}
The vector $\boldsymbol{x}_{s-p}$ is composed of $s-p$ dimensions, where each dimension is a duplicate of a specific dimension in $\boldsymbol{x}$. For example, if $\boldsymbol{x}=(0.1,0.2,0.3)$, then $\eta_{2,5}(\boldsymbol{x})=(0.1,0.2,0.3,0.1,0.2)$.

The basic idea to prove theorem \ref{thm.universal} is to use the famous universal approximation theorems of common feed-forward neural networks mentioned in \cite{hornik1991universal,Haykin1998neural} and of invertible neural networks mentioned in \cite{Teshima2020inn,Teshima2020ode} as the bridges, and then try to prove that any feed-forward neural network can be simulated by a serious of bijective mapping($\psi$), projection($\chi$) and vector expansion ($\eta$) procedures. The basic encoder after the extension for vector expansion can be expressed by the following equation:
\begin{equation}
    \phi= Proj_q \circ\psi_{s} \circ \eta_{p,s}\circ \psi_{p}.
\end{equation}
The functions $\psi_s: \mathcal{R}^s\rightarrow \mathcal{R}^s$ and $\psi_p: \mathcal{R}^p\rightarrow \mathcal{R}^p$ represent two reversible mappings. The final dimensionality $q$ that is retained may be larger than the initial dimensionality $p$. In the context of coarse-graining microscopic states to obtain macroscopic states, for the sake of convenience, we usually consider $\phi$ as a dimension reduction operator.
  
First, we need to prove a lemma.

\begin{lem}
\label{lemma.wx}
For any vector $X\in \mathcal{R}^p$ and matrix $W\in \mathcal{R}^{s\times p}$, where $s,p\in \mathcal{N}$, there exists an integer $s_1\leq \min(s,p)$ and two basic units of encoder: $\psi_{s}\circ\eta_{s_1,s}$ and $\chi_{p,s_1}\circ \psi_{p}$ such that:
\begin{equation}
    W\cdot X\simeq(\psi_{s}\circ\eta_{s_1,s})\circ(\chi_{p,s_1}\circ \psi_{p})(X),
\end{equation}
\end{lem}

where, $\simeq$ represents approximate or simulate.
\begin{proof}
For any $W\in \mathcal{R}^{s\times p}$, we can SVD decomposes it as:
\begin{equation}
    W = U\cdot\Lambda\cdot V,
\end{equation}
where $U\in \mathcal{R}^{s\times s}, V\in \mathcal{R}^{p\times p}$ are all orthogonal matrices, $\Lambda=\begin{pmatrix}
    diag(\lambda_1,\lambda_2,\cdot\cdot\cdot,\lambda_{s_1}) & \mathbf{0}\\
    \mathbf{0}&\mathbf{0}\\
\end{pmatrix}_{s\times p}$is a diagonal matrix composed by nonzero eigenvalues of $W$: $\lambda_i, i\in[1,s_1]$, and zeros, where $s_1$ is the rank of the matrix $W$.
Because all matrices can be regarded as functions, therefore:
\begin{equation}
    W(X) = U(\eta_{s_1,s}(\chi_{p,s_1}(\Lambda'\cdot V)(X))),
\end{equation}
where, $\Lambda'=diag(\lambda_1,\lambda_2,\cdot\cdot\cdot,\lambda_{s_1},1,\cdot\cdot\cdot,1)$. That is, the functional mapping process of $W(X)$ can be decomposed into four steps: matrix multiplication by $\Lambda'\cdot V$, projection $\chi_{p,s_1}$(projecting a $p$ dimensional vector to an $s_1$ dimensional vector), vector expansion $\eta_{s_1,s}$ (expanding the $s_1$ dimensional vector to an $s$ dimensional one), and matrix multiplication by $U$. Notice that the first and the last steps are invertible because the corresponding matrices are invertible.
Therefore, according to \cite{Teshima2020inn,Teshima2020ode}, there are invertible neural networks $\psi_s$ and $\psi_p$ that can simulate orthogonal matrix $U$ and the invertible matrix $\Lambda'\cdot V$, respectively. Thus, $\psi_s\simeq U$, and $\psi_{p}\simeq \Lambda'\cdot V$. Then we have:
\begin{equation}
    (\psi_{s}\circ\eta_{s_1,s})\circ(\chi_{p,s_1}\circ \psi_{p})\simeq W,
\end{equation} 
That is $W(X)$ can be approximated by an extended stacked encoder.
\end{proof}

With lemma \ref{lemma.wx}, we can prove theorem \ref{thm.universal}. We at first restate theorem \ref{thm.universal} here:

\textbf{Theorem} \ref{thm.universal}(Universal Approximating Theorem of Stacked Encoder): For any continuous function $f$ which is defined on $K\times \mathcal{R}^p$, where $K\in \mathcal{R}^p$ is a compact set, and $p>q\in \mathcal{Z^+}$, there exists an integer $s$ and an extended stacked encoder $\phi_{p,s,q}: \mathcal{R}^p\rightarrow \mathcal{R}^q$ with $s$ hidden size and an expansion operator $\eta_{p,s}$ such that:

\begin{equation}
    \phi_{p,s,q}\simeq f,
\end{equation}
Thereafter, an extended stacked encoder is of universal approximation property which means that it can approximate(simulate) any coarse-graining function defined on $\mathcal{R}^p\times \mathcal{R}^q$.

\begin{proof}
According to universal approximation theorem in \cite{hornik1991universal,Haykin1998neural}, for any function $f$ defined on $K\times \mathcal{R}^p$, where $K\in \mathcal{R}^p$ is a compact set, and $p>q\in \mathcal{Z^+}$, and small number $\epsilon$, there exists an integer $s$ and $W\in\mathcal{R}^{s\times p}, W'\in\mathcal{R}^{q\times s}, b\in\mathcal{R}^{s}$ such that:

\begin{equation}
    W'\cdot \sigma(W+b)\simeq f,
\end{equation}
where, $\sigma(\boldsymbol{x})=1/(1+\exp(-\boldsymbol{x}))$ is the sigmoid function on vectors.

According to Lemma \ref{lemma.wx} and $+b$ and $\sigma(\cdot)$ are all invertible operators, therefore, there exists invertible neural networks $\psi_{q},\psi_{s}',\psi_{s},\psi_{p}$, and two integers $s_1,s_2$ which are ranks of matrices $W'$ and $W$, respectively, such that:
\begin{equation}
    (\psi_{q}\circ\eta_{s_2,q}\circ\chi_{s,s_2}\circ\psi_{s}')\circ(\psi_{s}\circ\eta_{s_1,s}\circ\chi_{p,s_1}\circ\psi_{p})\simeq W'\cdot\sigma(W\cdot+b),
\end{equation}
where, $\psi_{s}\circ\eta_{s_1,s}\circ\chi_{p,s_1}\circ\psi_{p}$ approximates(simulates) the function $\sigma(W\cdot+b)$ and $\psi_{q}\circ\eta_{s_2,q}\circ\chi_{s,s_2}\circ\psi_{s}'$ approximates(simulates) the function $W'\cdot$.

Therefore, if we let $\phi_{p,s,q}=(\psi_{q}\circ\eta_{s_2,q}\circ\chi_{s,s_2}\circ\psi_{s}')\circ(\psi_{s}\circ\eta_{s_1,s}\circ\chi_{p,s_1}\circ\psi_{p})$, then $\phi_{p,s,q}\simeq f$
\end{proof}
In real applications, although the basic encoder and the extended versions do not include expansion operator, we always expand input vector before it is input for the encoder. Therefore, it is reasonable to believe that Theorem \ref{thm.universal} still holds for stacked encoders.
\section{SIR Model}
\label{sec:sir}

In this section, we will introduce various parameters and specific methods in the SIR experiment. 

\subsection{Data generation}
\label{sec:sir data detail}
In Equation\ref{eq:sir noise}, $\boldsymbol{\xi}_1,\boldsymbol{\xi}_2 \sim \scriptsize{N}(0,\Sigma)$ are two-dimensional Gaussian noises and independent each other, and $\Sigma=\begin{pmatrix}\sigma^2 & -\frac{1}{2}\sigma^2 \\-\frac{1}{2}\sigma^2 & \sigma^2\end{pmatrix}$. For this covariance matrix, we can easily deduce that $\boldsymbol{\xi}_1$ and $\boldsymbol{\xi}_2$ are both two-dimensional normal distributions with a correlation coefficient of $-\frac{1}{2}$. The parameter $\sigma$ controls the overall magnitude of the noise, which is set to $\sigma=0.03$ in the experiments conducted in this paper. For macroscopic state data, we discretize the original continuous model with $dt=0.01$ and run 7 time steps for each starting point. After sampling with different starting points, we concatenate equal amounts of $\boldsymbol{\xi}_1$ and $\boldsymbol{\xi}_2$ to obtain four-dimensional microscopic data. 

In the SIR model experiment conducted in this article, we generated two datasets: the complete dataset and the partial dataset. For the complete dataset, we uniformly sampled 9000 initial points from the entire range of possible values for $S$ and $I$ ($S\geq0,I\geq0,S+I\leq1$), resulting in a total of 63000 data points. For the partial dataset, we uniformly sampled 6000 initial points from the region where $S\geq\frac{1}{3}$. Additionally, to satisfy the requirements of KDE estimation for the support set, we randomly sampled 10 initial points from the region where $S\le \frac{1}{3}$. In total, the partial dataset consists of 42070 sample points. The training and testing results of the NIS+ neural network on these two different sample sets are presented in Figure \ref{fig:sir}.

\subsection{Training and testing details}
\label{sec:sir train}
For NIS+, two stages are separated at the 3000th epoch(see method section \ref{sec:train nis+}). In the second stage (after 3000 epochs), NIS+ incorporates inverse probability weighting and bidirectional dynamic learning, distinguishing it from the NIS model. 

In Figure \ref{fig:sir}(d), the test involves selecting 20 starting points within an area where the training samples are missing, evolving them for 10 steps using the learned macro-dynamics, and then decoding the results back into micro-states. The mean of the multi-step prediction errors for these 20 starting points is displayed.
In Figure \ref{fig:sir}(e), considering the significant differences in the range of experimental data, the threshold $\epsilon$ measures the magnitude of relative error, which is obtained by dividing the absolute error by the standard deviation of the data itself. We set $\epsilon=0.3$, which means that the machine training is considered sufficient when the relative error is less than or equal to 0.3.

\subsection{Methods for the vector field results}
\label{sec:sir vector field}
 The following step-by-step explanation outlines how to average the vector field results of emergent dynamics, as demonstrated in Figure \ref{fig:sir_si}. The aim of this process is to compute the vector field through extensive sampling and averaging. Initially, we generate grid points on the phase space with a step size of 0.007, as shown in Figure \ref{fig:sir_si}(a). This step serves to obtain a significant number of samples for further analysis. To achieve this, we introduce noise to these grid points and their subsequent time steps, extending them to four-dimensional data. Next, we employ the well-trained NIS+ model to encode these extended data points, resulting in samples in the latent space. These samples are then divided into different grids on the latent space with a step size of 0.015, enabling a comprehensive analysis of the emergent dynamics. The subsequent calculation involves determining the average motion vector for each grid using the methodology outlined in Figure \ref{fig:sir_si}(d). It is important to note that the primary objective of this step is to compute the average vector field by aggregating the motion vectors across the grids. For comparison, we also generate sparser grid points on the phase space with a step size of 0.07. By applying the same noise and encoding techniques, we establish a one-to-one correspondence between the motion vectors of the grids and the new grid points. This correspondence allows us to depict the vector field of the macro latent space of NIS+ as shown in Figure \ref{fig:sir}(c). Similarly, by replacing the encoder of NIS+ with the encoder of NIS, we can obtain the latent space vector field of NIS, as illustrated in Figure \ref{fig:sir}(f). This approach, based on extensive sampling and averaging, provides valuable insights into the macro features of emergent dynamics.


\begin{figure}
    \centering
    \includegraphics[width=0.8\textwidth]{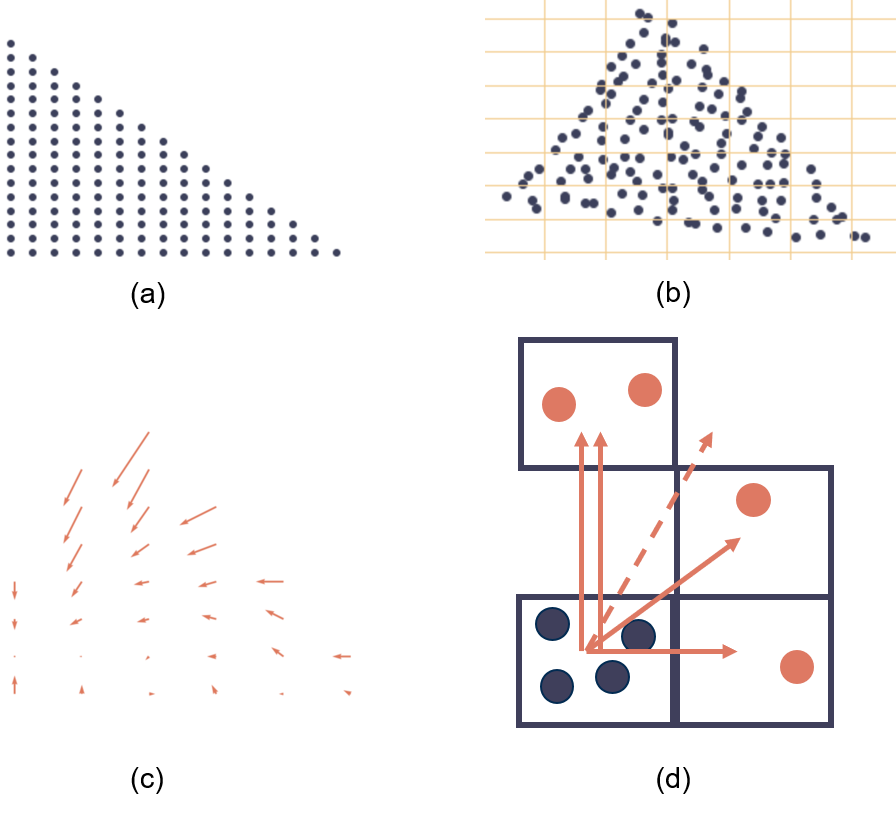}
    \caption{The method of generating the vector fields in Figure \ref{fig:sir}(c)(f). (a) Generate a sufficient number of grid points in phase space. (b) Divide the encoded latent space into grids with certain intervals, and each sample point in the latent space will fall into a certain grid. (c) Take the expectation of different direction vectors for each grid, obtaining the corresponding motion vector for that grid. (d) Zoom in on the local amplification schematic of the averaging process. The black dots in the figure represent the coordinates of sample points at a certain time, and the red dots represent the coordinates of these sample points at the next time step. If a sample point jumps from the i-th grid to the j-th grid in one time step, then a count is made in the direction of i→j, represented by a solid red arrow in the figure. Taking the expectation of all solid arrows yields the motion vector corresponding to that grid, represented by a dashed red arrow.}
    \label{fig:sir_si}
\end{figure}

In Figure \ref{fig:sir}(c) and (f), we provide a theoretic vector field. It is obtained by the folowing steps. First, 
\begin{gather}
     \frac{d\boldsymbol{y}_t}{dt}=\frac{\partial\phi(\boldsymbol{x}_t)}{\partial t}=\frac{\partial\phi(\boldsymbol{x}_t)}{\partial\boldsymbol{x}_t}\frac{d\boldsymbol{x}_t}{dt},
\end{gather}
where $J_{\phi}\equiv\frac{d\phi(\boldsymbol{x}_t)}{d\boldsymbol{x}_t}$ is the Jacobian matrix of the encoder. Therefore, we could obtain the vector by representing the relationship between the true dynamics and the latent space dynamics. We notice that according to Equation \ref{eq:sir noise}: 
\begin{gather}
     \frac{d\boldsymbol{x}_t}{dt}=\frac{d(S_t,S_t,I_t,I_t)}{dt},
\end{gather}
Because $\boldsymbol{x}_t = (\boldsymbol{S}'_t,\boldsymbol{I}'_t)$ and the noises are independent of time and have zero mean, we can replace $\frac{d\boldsymbol{x}_t}{dt}$ in the experiment. Finally, we obtain:
\begin{gather}
     \frac{d\boldsymbol{y}_t}{dt}=J_{\phi}\cdot\frac{d(\boldsymbol{S}_t,\boldsymbol{I}_t)}{dt},
\end{gather}
where, $\boldsymbol{S}_t=(S_t,S_t)$ and $\boldsymbol{I}_t=(I_t,I_t)$. Therefore, after multiplying the ground truth vector field with the Jacobian matrix $J_{\phi}$, we obtain the theoretic prediction of the vector field in the graph. The difference between this theoretic vector field and the true dynamics is due to the errors of the encoder $\phi$.


In order to achieve a better dynamic vector field, we specifically conducted multi-step prediction training when drawing Figure \ref{fig:sir}(c) and (f). Specifically, given the input of microstate at a certain moment, we make multi-step predictions for the next 10 microstates respectively and combine them with certain weights to form a loss for gradient backpropagation. The weights decay at a rate of $e^{-0.2t}$ as the prediction steps increase. For reverse prediction, we take the target corresponding to the last step as the input and perform multi-step prediction training in the same way. In addition, we adopt the approach of interval sampling, taking every 100 sample points as a step, so that NIS+ can make long-term predictions on the SIR model.

\section{Boids Model}
\label{sec:Boidsmodel}
In 1986, Craig Reynolds created a simulation of the collective behavior of birds, known as the Boids model \cite{Reynolds1987}. This model only used three simple rules to control the interactions between individuals, resulting in flock-like behavior.
\subsection{Method of Simulating the Trajectory of Boids}
\begin{figure}[htbp]   
    \centering       
    \includegraphics[width=0.6\textwidth]{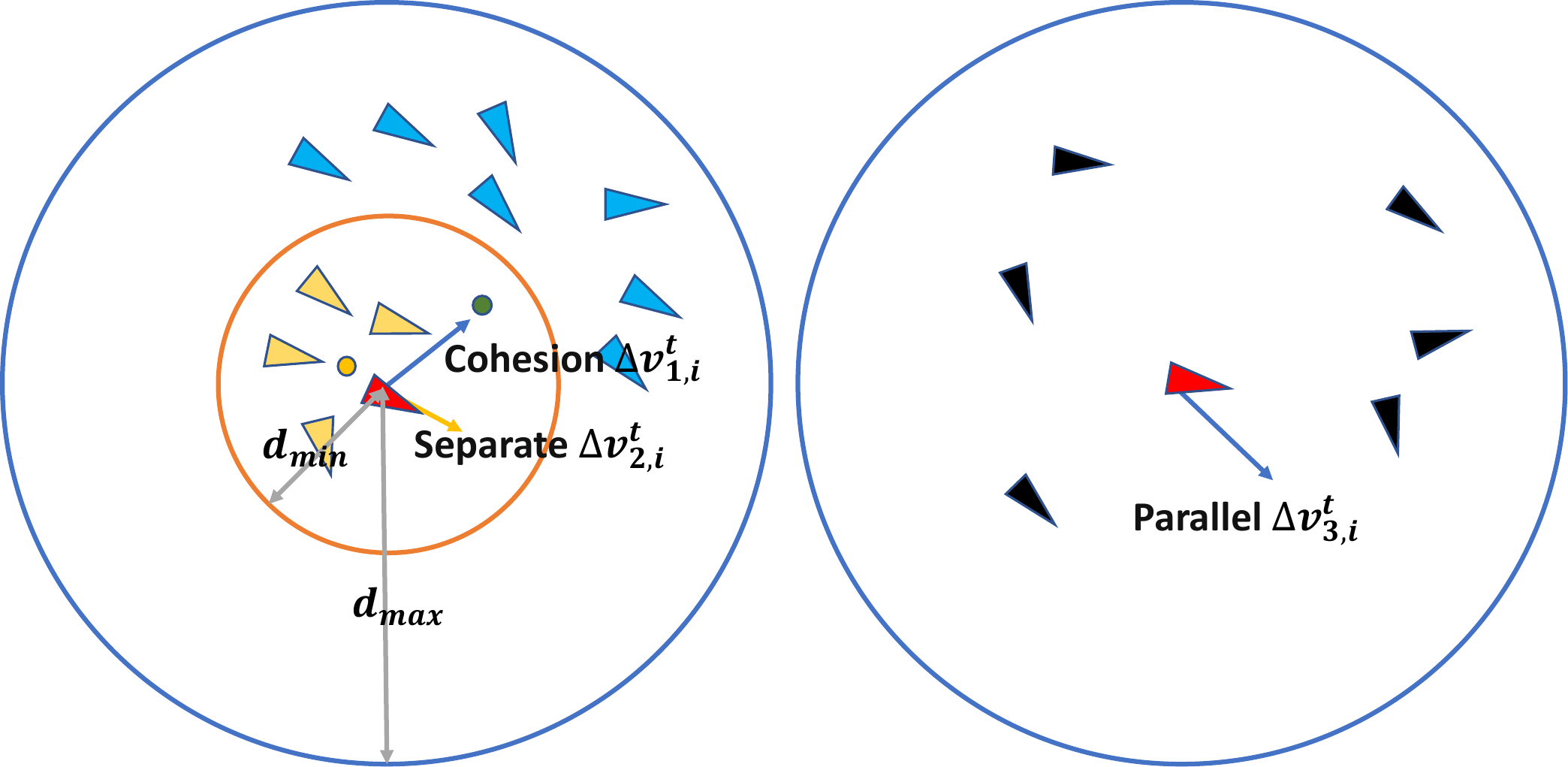}
  \caption{Boids model. By modifying $d_{\max}$ and $d_{\min}$ at time $t$, the aggregation of bird flocks will change. Specify $\Delta v_ {1, i}^t$, $\Delta v_{2, i}^t$, $\Delta v_{3, i}^t$ as the variation of separation, parallelism, and cohesion velocities, respectively. The $i$-th bird is marked in red. At time $t$, birds with a distance less than $d_{min}$ from bird $i$ will affect $\Delta v_{2, i}^t$. Birds with a distance less than $d_{max}$ from bird $i$ will affect $\Delta v_ {1, i}^t$, $\Delta v_{3, i}^t$. In this way, birds can maintain a distance from each other and allow the Boids to fly in an orderly manner in one direction.}  
  \label{fig:Boidsmodel}
\end{figure}

In Boids model, the micro-state of each individual boid $i$ is a four dimensional vector $(x_i,y_i,v_{x,i},v_{y,i})$, where $\boldsymbol{x}_i=(x_i,y_i)$ is the positional vector, and $\boldsymbol{v}_i=(v_{x,i},v_{y,i})$ is the velocity vector. Thus, the complete micro-state for all $N$ boids is a $4N$ dimensional vector. The micro-states follow the dynamical update equations:

\begin{eqnarray}
\label{Boids}
\begin{cases}
\boldsymbol{x}_i^{t+1}=\boldsymbol{x}_i^{t}+\boldsymbol{v}_i^{t},\\
\boldsymbol{v}_i^{t+1}=\displaystyle\frac{\boldsymbol{v}_i^{t}+\Delta \boldsymbol{v}_ {i}^t+\varepsilon_i^t}{||\boldsymbol{v}_i^{t}+\Delta \boldsymbol{v}_ {i}^t+\varepsilon_i^t||}||\boldsymbol{v}_i^{t}||.
\end{cases}
\end{eqnarray}
Where, $\Delta\boldsymbol{v_ {i}}^t$ is the external force exerted on $i$ by its neighbors, and it is composed with three components which called cohesion, alignment, and separation, respectively, according to the dynamical rules of Boids model.  $\varepsilon_i^t=(\cos\Delta\alpha_i^t,\sin\Delta\alpha_i^t)^T$ is a force on turning direction exerted on each boid at each time step, where $\Delta\alpha\in[-\pi,\pi]$ is a fixed value or random number.

The motion vector of bird $i$ at time $t$ can be expressed as $(x_i^t,y_i^t,v_{x,i}^t,v_{y,i}^t)$ where $x_i^t$ and $y_i^t$ represent the position coordinates, $v_{x,i}^t$ and $v_{y,i}^t$ represent the projection of velocity on two coordinate axes. With $N$ boids, $i=1, 2 , \dots, N$. The motion vector can be decomposed into position vector and velocity vector $\textbf{x}_i^t=(x_i^t,y_i^t)$ and $\textbf{v}_i^t=(v_{x,i}^t,v_{y,i}^t)$. The distance between two boids $i$ and $j$ can be expressed as $d_{ij}=||\textbf{x}_i^t-\textbf{x}_j^t||=\sqrt{(x_i-x_j)^2+(y_i-y_j)^2}$. Although the Boids model belongs to a continuous system, in order to facilitate the simulation process and subsequent experiments, we uniformly set $\Delta t=1$. For a single boid in the model, the acceleration $\textbf{a}_i^t=\textbf{v}_i^{t+1}-\textbf{v}_i^{t}$ can be divided into four parts: separation, parallelism, cohesion, and random deflection angle noise. We stipulate that the velocity magnitude of a boid during the flight remains stable, so the impact of acceleration mainly changes the direction of Boids' flight. After knowing the motion state vector of each boid in the flock at time $t$, the motion state at time $t+1$ can be generated. Specify $\Delta v_ {1, i}^t$, $\Delta v_{2, i}^t$, $\Delta v_{3, i}^t$ as the variation of separation, parallelism, and cohesion velocities, respectively. In which 
\begin{eqnarray}
\label{delteVboids}
\begin{cases}
\Delta v_ {1, i}^t=\displaystyle\frac{\textbf{x}_i^t-\textbf{x}_{\Phi_i}^t}{||\textbf{x}_i^t-\textbf{x}_{\Phi_i}^t||},\\
\Delta v_ {2, i}^t=\displaystyle\frac{\sum_{k\in\Psi_i^t}\textbf{v}_k^t}{||\sum_{k\in\Psi_i^t}\textbf{v}_k^t||},\\
\Delta v_ {3, i}^t=-\displaystyle\frac{\textbf{x}_i^t-\textbf{x}_{\Psi_i}^t}{||\textbf{x}_i^t-\textbf{x}_{\Psi_i}^t||}.\\
\end{cases}
\end{eqnarray}
We stipulate that the maximum and minimum detection distances for Boids are $d_{\max}$ and $d_{\min}$. The sets of Boids within the minimum and maximum detection range are specified as $\Phi_i^t=\{j|d_{ij}^t<d_{\min}\}$ and $\Psi_i^t=\{j|d_{ij}^t<d_{\max}\}$. So the centers of gravity of $\Phi_i$ and $\Psi_i$ are represented as $\textbf{x}_{\Phi_i}^t=\sum_{k\in\Phi_i^t}\textbf{x}_k^t/|{\Phi_i}|$ and $\textbf{x}_{\Psi_i}^t=\sum_{k\in\Psi_i^t}\textbf{x}_k^t/|{\Psi_i}|$, which are the mean values of the boids' position coordinates within two sets. With equation 
\begin{equation}
\Delta \boldsymbol{v}_ {i}^t=\Delta v_ {1, i}^t+\Delta v_{2, i}^t+\Delta v_{3, i}^t
\end{equation}
we can obtain Equation \ref{Boids} in Section \ref{sce:Boids}. The relationship between the above variables is shown in Fig.\ref{fig:Boidsmodel}. By modifying the starting vector $(x_i^0,y_i^0,v_{x,i}^0,v_{y,i}^0)$ or $d_{\max}$ and $d_{\min}$, the aggregation of bird flocks will change. And we can adjust the trajectory and randomness of Boids by modifying the mean and variance of the deflection angle $\varepsilon_i^t$ in Equation \ref{Boids}.

\subsection{Methods for Collecting Experimental Data and Training}
We record the velocity vector and position vector of each bird at each time $t$ during its flight as data. For each time step $t$, we need to record the data at that time and the data at the next time step $t+1$. We merge the data into $4N$ dimensional vectors at each time step $t$ as
\begin{equation}
X_t=(x_1^t,y_1^t,v_{x,1}^t,v_{y,1}^t, \dots, x_N^t,y_N^t,v_{x,N}^t,v_{y,N}^t).
\end{equation}
This vector corresponds to our micro-state data $X_t$, $p=64$. We can repeat the preparation and recording stages multiple times until the data volume meets our expectations. The ID of birds is $1,2,...,16$, evenly divided into two sets, $\{1, 2, \dots, 8\}$ and $\{9, 10, \dots, 16\}$. The model we generated contains two sets of birds randomly generated near the center of the canvas. 

We generate observation data in two phases: preparation and recording. In the preparation phase, one group of birds is randomly generated within a circle with $(148,150)$ as the center and 5 as the radius and the other group of birds is randomly generated within a circle with $(152,150)$ as the center and 5 as the radius. Then by modifying the $(v_{x,i}^0,v_{y,i}^0)$ or the mean and variance of the deflection angle $\varepsilon_i^t$ two groups of birds can fly different trajectories in Fig.\ref{fig:boids1} and \ref{fig:boids4}. The norm of the velocity vector is always 1, that is $||\textbf{v}_i^{t}||\equiv1, t=1,2,\dots.$ We set $d_{max}=5$ and $d_{min}=1$ so that the birds within the same group will interact with each other, while the effects between birds in different groups can be ignored. Subsequently, we allowed two groups of birds to fly for 20 steps without affecting each other, resulting in two groups of birds with distance and relatively stable internal stability. Then the positions and velocities of all birds serve as the starting point for us to record data. 

In the recording phase, we have the birds continue to run 50 steps on the basis of the preparation phase, recording the velocity and position vectors of each bird separately. We randomly sort the generated data and input them into our NIS+ model for training, using the multistory structure of the NIS+ in Fig.\ref{fig:structures} during the training process. We sample the data for 4 batches. After inputting $X_t$, the microstate $Y_t$, $Y_{t+1}$ and predicted value $\hat{X}_{t+1}$ can be output, and training can begin. The encoders and decoders are stacked, and the dynamics learners in different layers, $q=64,32,16,8,4,2,1$, operate in parallel. This design facilitates simultaneous training for different dynamics learners with distinct dimensions and enables parallel searching for the optimal dimension, which is $q=8$. Due to the previous experiments showing that the model performs better when $q=8$ than other scales, and having too many layers can also slow down learning efficiency, we can discard layers with $q<8$ in the later experiments. All the boids are separated into two groups by forcing their $\Delta\alpha_i^t$ as two distinct values $\Delta\alpha_i^t\sim U(0.0058\pi, 0.0098\pi)$ for boids with $i\leq 8$, and $\Delta\alpha_i^t\sim U(-0.0124\pi, -0.0084\pi)$ for boids $i>8$. Therefore, the two groups will have separating trajectories with different turning angles as shown in Figure \ref{fig:boids1}. 

After training with 800,000 epochs by NIS, $\mathcal{J}_q$ reaches its maximum value when $q=8$. We continue to optimize the existing model using NIS+ training 400,000 epochs, the results are shown in Fig.\ref{fig:boids3}. Then we can obtain multi-step prediction data as shown in Figure \ref{fig:boids1}. It can be seen that the multi-step predicted trajectory matches the flight trajectory of the real bird flock. However, if the variance of $\varepsilon_i^t$ is large, the training effect of the model will deteriorate, and the multi-step prediction distance will become shorter as shown in \ref{fig:boids4}. Causal emergence occurs for all tested dimensions($q$) (see Figure \ref{fig:boids3}). With up to $q=8$ dimensional macro-state vectors, NIS+ can best capture the emergent collective flying behaviors of the two groups by tracing their centers of the trajectories. This can be visualized by decoding the predicted macro-states into the predicted micro-states as shown in Figure \ref{fig:boids1} by the two solid lines. With the method of generating data, we can observe the impact of noise. For different observation noises, in cases where the noise is not too significant, we can use the same method to train the model with NIS first, then optimize it with NIS+ to obtain Fig.\ref{fig:boids5}. For different deflection angle noises, after training from random initial parameter values, in order to unify the control variables, we optimized them in the dimension of $q=8$ and obtained Fig.\ref{fig:boids6}.

We test the degree of Rosas' $\Psi$ using the same learned macro-state variable of NIS+. The results, shown in Figure \ref{BoidPSI}, display results of Rosas' $\Psi$ for different extrinsic noises. For Rosas' $\Psi$, all cases yield values are far less than 0. One possible reason to explain this is that much redundant information is ignored by the approximation by $\Psi$. Another possible reason is that the Boids' coordinate data has a large order of magnitude, which can result in a significant increase in the value of lost information. Thus, unreasonable results are produced which makes Rosas' $\Psi$ impossible to determine whether CE occurs. Therefore, the proposed $\Delta\mathcal{J}$ in this paper is a superior method for identifying CE.

\begin{figure}[htbp]   
    \centering       
    \includegraphics[trim={10 0 2 0},clip,width=0.5\textwidth]{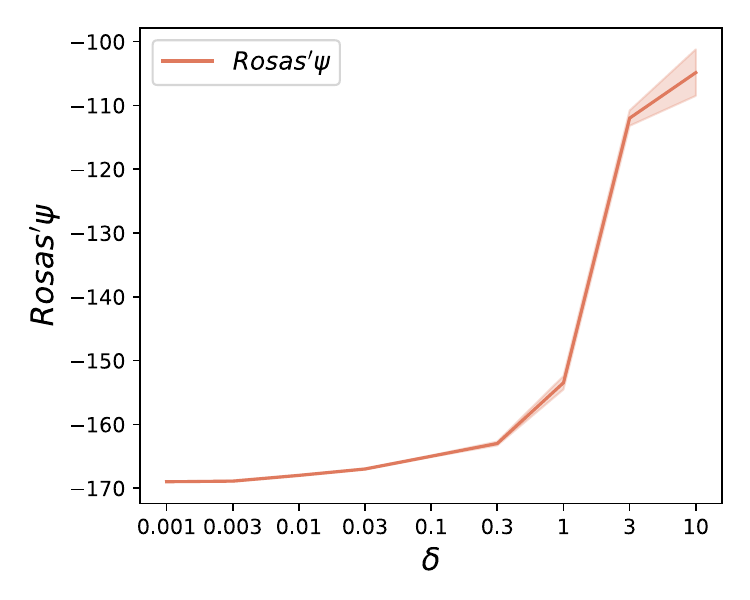}   
  \caption{The results of Rosas' $\Psi$ for different extrinsic noises. For Rosas' $\Psi$, all cases yield values are far less than 0. Thus, unreasonable results are produced which makes Rosas' $\Psi$ impossible to determine whether CE occurs. Therefore, the proposed $\Delta\mathcal{J}$ in this paper is a superior method for identifying CE for Boids model.}  
  \label{BoidPSI}
\end{figure}

\section{Game of Life}
\label{sec:gamelife}
Conway's Game of Life is a famous two dimensional cellular automota model on which various interesting dynamical patterns like glider, square, flower, signal light, honeycomb, traffic light emerge. Different from SIR model and Boids model, the micro-states of Game of Life at each time step are discrete(0 or 1) on a large regular grid as shown in Figure \ref{fig:gamelife}. Further, the micro-dynamics can not be represented by differential or difference equations but by rule tables (details can be referred to section~\ref{sec:data_generation}). 

We train NIS+ using data generated by the Game of Life simulation with random initial conditions and extract the time series of states from the 100th step to the 120th step. Figures \ref{fig:gamelife}(a),(b) and (c) show the dynamical patterns that are generated by the ground truth simulations(the first row) and the predictions by NIS+(the third row), as well as the emergent macro-states that can make those predictions(the second row). We input two images with successive time steps into NIS+, and obtain another image pair with next two successive time steps. Compare the upper picture and lower ones, the patterns are similar. However, the learned and predicted patterns in the third column, specifically the ``glider'' pattern, appear vague due to the limited occurrence of training samples with this pattern in random initial conditions. To enhance the quality of predictions, we can generate a new set of training samples that include initial conditions with two ``gliders''. As a result, the predictions become clearer, as depicted in Figure \ref{fig:gamelife}(d), even though the number of gliders in this test environment is three. That means, NIS+ can capture the patterns including moving, static, and oscillating structures. 

Furthermore, we evaluate the generalization ability of the models in test environments that differ from the training environments. We compare the multi-step prediction performance of NIS and NIS+ across eight different pattern types, which are distinct from the initial random patterns. The results are depicted in Figure \ref{fig:gamelife}(f) showing that NIS+ consistently achieves a higher AUC (area under the curve) than NIS for all the pattern types. Where, in the tick labels of the x-coordinate, we adopt the format of ``pattern name (quantity)'' to represent various initial conditions. For instance, ``glider(2)'' signifies an initial configuration comprising two gliders. Thus, we have provided the evidence that the enhanced NIS+ method possesses superior generalization ability in capturing these patterns.

We further test the degree of CE ($\Delta\mathcal{J}$) and compare it with Rosas' $\Psi$ using the same learned macro-state variable of NIS+. We use the same patterns as initial conditions to test for the comparison. The results, shown in Figure \ref{fig:gamelife}(g), display results of CE, in which each bar representing a combination of $\Delta\mathcal{J}$ and Rosas' $\Psi$ for one initial pattern. Regarding $\Delta\mathcal{J}$, except for the ``random'' case, all eight cases demonstrate the occurrence of CE. The case with ``glider'' patterns exhibit the lowest degree of CE due to poor prediction (see Figure \ref{fig:gamelife}(c)). The remaining seven patterns show similar $\Delta\mathcal{J}$ values. These results indicate that $\Delta\mathcal{J}$ provides a more reasonable indication of the occurrence of CE, aligning with our intuition. However, for Rosas' $\Psi$, all cases yield values less than or equal to 0. One possible reason to explain this is that much redundant information is ignored by the approximation by $\Psi$. Thus, unreasonable results are produced which makes Rosas' $\Psi$ impossible to determine whether CE occurs. Therefore, the proposed $\Delta\mathcal{J}$ in this paper is a superior method for identifying CE.

Additionally, this example highlights the versatility of NIS+. In order to conduct the aforementioned experiments, we needed to coarse-grain the micro-states of the cellular automata in both spatial and temporal dimensions. To address this challenge, we incorporated the concept of spatiotemporal convolution. The architecture used in this experiment is illustrated in Figure \ref{fig:gamelife}(e). The entire coarse-graining process can be divided into two steps: first, aggregating information within a fixed-size window (a 3x3 window in this paper) to obtain spatial coarse-grained results; and second, aggregating these results over multiple successive time steps to form a spatiotemporal coarse macro-state. All of these processes are implemented through parallel encoders in NIS+.

\begin{figure}[htbp]    
  \centering
     \includegraphics[width=1\textwidth]{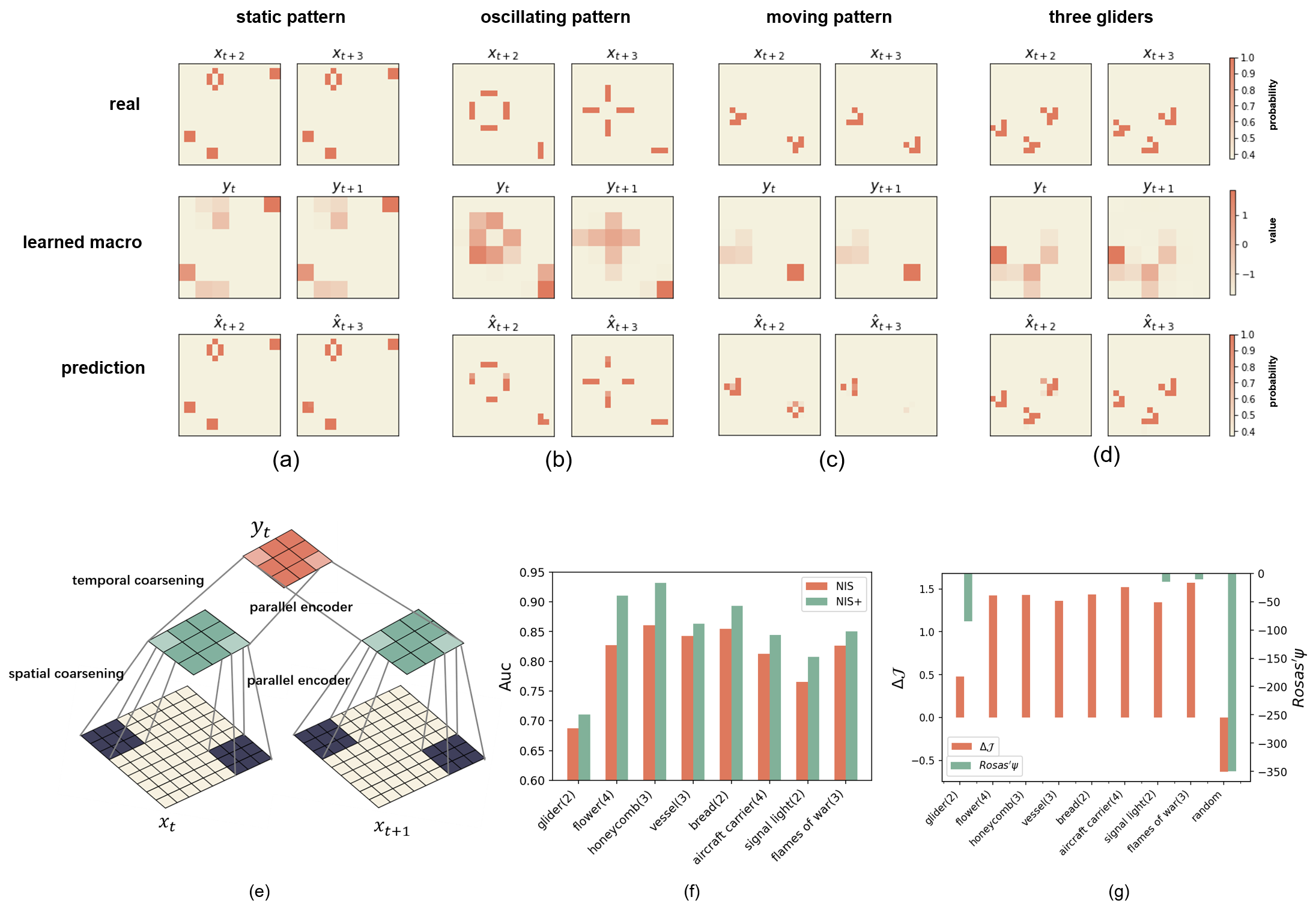}
      

  \caption{Experiments on the Game of Life. (a-d) depict the comparisons between the ground truth simulations, learned macro-states, and the predictions made by NIS+ for different patterns. In (a-c), the training data is generated by simulating the Game of Life for 120 steps with random initial conditions, and only the images from steps 100-120 are selected to form the training set. (d) depicts a generalization test involving the number of gliders was conducted, transitioning from the original input of two gliders to three gliders. In (a-d), the predicted images are the averaged results obtained by sampling from the predicted macro-states using the decoder 20 times. The colors represent the probability of repetitions in the 20 samples. (e) The architecture of the encoder is illustrated, which can be divided into spatial encoding (mapping 3x3 grids to 1 grid) and temporal encoding (mapping 2 time steps to 1 step) stages. These processes are implemented using parallel encoders. (f) shows the comparisons of AUC (area under the curve) between NIS and NIS+ on different patterns, where each result is the average of multi-step predictions (2 steps). Each case represents multiple quantity combinations of a selected pattern as the initial states. There are 9 patterns, which can be divided into four categories: static (flower, vessel, bread, aircraft carrier and honeycomb), oscillating (signal light, flames of war), moving (glider), and random (random), The horizontal coordinate corresponding to each bar in the figure is represented by the format of ``pattern name (quantity)''. (g) Comparisons of the degrees of CE ($\Delta\mathcal{J}$, dimension averaged EI and Rosas' $\Psi$) are tested with various patterns. }
  \label{fig:gamelife}            
\end{figure}


\subsection{Data generation}
\label{sec:data_generation}
In this paper, we use Conway's Game of Life as the experimental object, in which each cell has two states for two-dimensional state input: alive~(1) or dead~(0), and each cell is affected by its eight neighbors. The evolution of the Game of Life is only affected by the input state and its update rules, in which the Game of Life has four evolutionary rules, respectively corresponding to cell reproduction and death, and so on. The update rules for the Game of Life are shown in the following table:

\begin{table}[htbp]
  \centering
  \caption{The update rules for the Game of Life.}
  \begin{tabular}{|c|c|}
    \hline
     $x_t$ & $x_{t+1}$ \\ 
    \hline
    0 & 1~(there are three living cells around) \\ 
    \hline
    \multirow{3}{*}{1} & 0~(less than two surviving cells around the cell (excluding two)) \\
    \cline{2-2}
     & 1~(there are two or three living cells around) \\ 
    \cline{2-2}
     & 0~(there are more than three living cells around) \\ 
    \hline
  \end{tabular}
  \label{table:game_of_life}
\end{table}

The training sample generation process of the Game of Life is as follows: firstly, a state $x_t$ is initialized. When considering a temporal coarse-grained of two steps, the subsequent three steps of states $x_{t+1}$, $x_{t+2}$, and $x_{t+3}$ are then generated based on the update rule and are input to the machine learning model. The two input states are $x_{t}$ and $x_{t+1}$, with the micro-dynamics output being $x_{t+1}$ and $x_{t+2}$. Due to the utilization of spatiotemporal coarse-graining, the macro-dynamics will output a macro-state which will be decoded into the micro-states $x_{t+2}$ and $x_{t+3}$. This process is repeated multiple times (50,000 samples) and generate the data for training in Figure~\ref{fig:gamelife}d. While in the other experiments, we generate 500,000 samples. 

\subsection{The model predictive ability on glider pattern}
\label{sec:gofl_generalization}

We then test the capability of capturing dynamical patterns on glider pattern, 
where the model was trained based on two glider patterns. The model depicts a good prediction effect and the results are shown in Figure~\ref{fig:game_of_life_si}. 

In addition, please refer to Table~\ref{table:parameters} for more detailed information on the other model parameters.

\begin{figure}[htbp]    
  \centering  
  \includegraphics[width=0.5\textwidth]{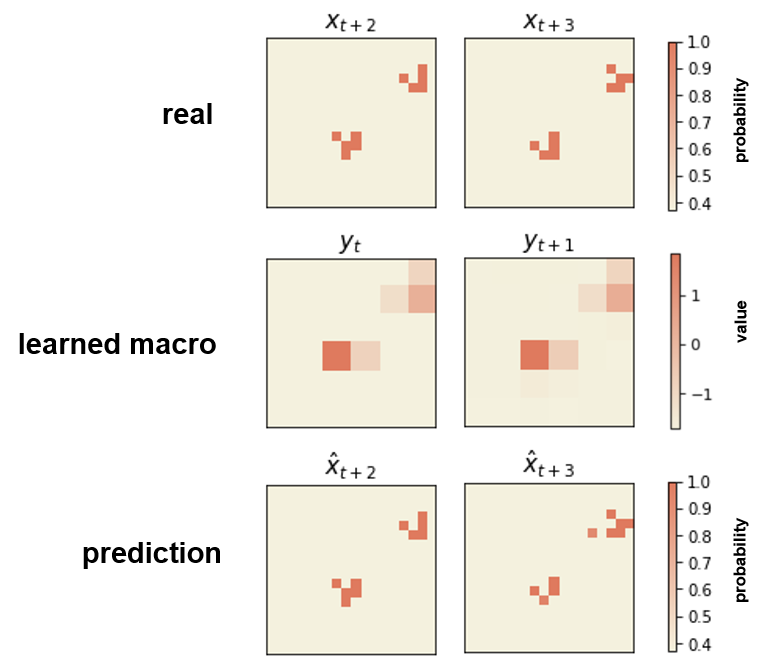}
  \caption{\textbf{Experiments of Game of Life.} 
  The results show the dynamical patterns that are generated by the ground truth simulations(the first row) and the predictions by NIS+(the third row), as well as the emergent macro-states that can make those predictions(the second row). Here the input state size of the model test is 18$\times$18$\times $2 and also consists of two gliders.}  
  \label{fig:game_of_life_si}   
  \vspace{-5pt}
\end{figure}


\section{Brain Results}
\label{sec:brain}

\subsection{resting state}\label{sec:resting state}

We also apply integrated gradient to see what seven macro dimensions mean for micro brain areas. Different from results from visual fMRI,we can see some distinct patterns. Unfortunately, we can't see that some one-to-one relationship between these seven dimensions and seven systems defined by Schaefer atlas(See Figure \ref{fig:brain_si_2})

\begin{figure}[!ht]
	\centering
	\includegraphics[width=\textwidth]{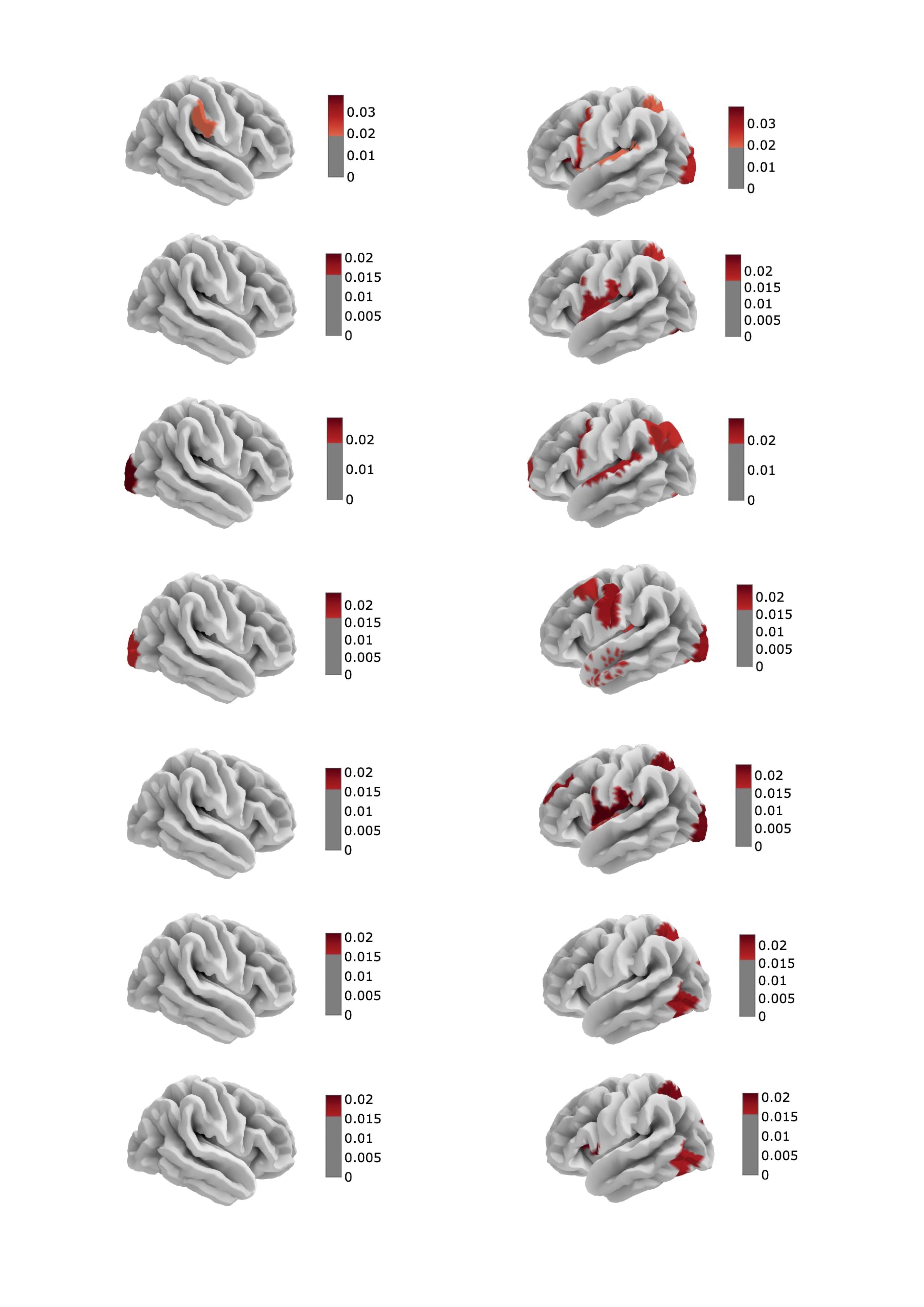}
	\caption{The integrated gradient results in resting NIS+ analysis when $q=7$.Each pair of left hemisphere and right hemisphere correspond to one macro dimension.Only five percent of most important areas in attribution value are drawn to better show the distinct patterns between seven different marco dimenions}
	\label{fig:brain_si_2}
\end{figure}

\subsection{AAL Results}\label{sec:AAL Results}

We repeat what we have done in sections \ref{Schaefer Results} and \ref{fMRI_preprocessing} expect that we have utilized AAL3 atlas and made $\lambda$ to be 0.5 for layer 1 training of NIS+ since $\lambda=1$ will lead to a trivial solution where macro predictions for all areas will remain the same, which make the prediction can't be meaningful. We can see a converging evidence that NIS+ framework will lead to an improvement of $\Delta\mathcal{J}$(see Figure \ref{fig:brain_si_preprocessed}(b) and a stronger correlation with those areas which responds to the visual areas(See Figure \ref{fig:brain_si_preprocessed}(a) and (c)).

\begin{figure}[!ht]
	\centering
	\includegraphics[width=\textwidth]{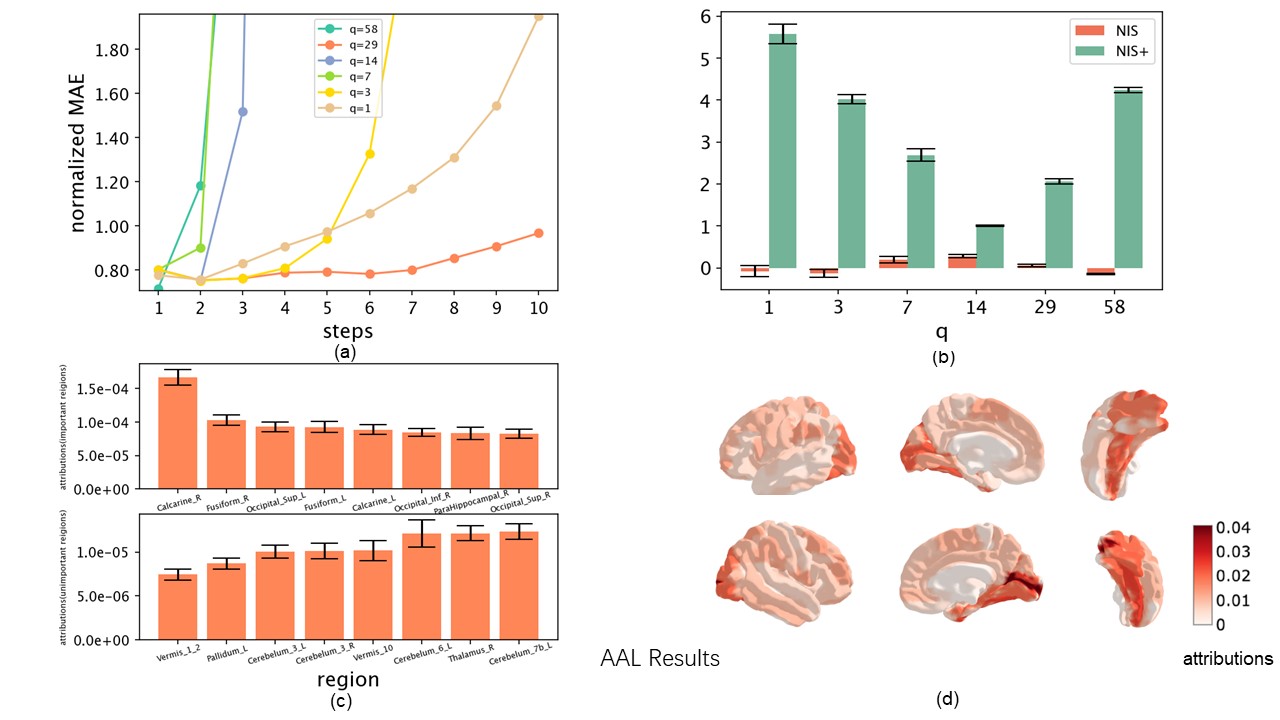}
 	\caption{The learning results, the degree of causal emergence, and the attribution analysis of NIS+ and NIS on the fMRI data for the brains. (a) The mean errors of the multi-step predictions increase with the prediction steps under different scales ($q$) on the test dataset. (b) Measures of CE (dimension averaged, $\Delta\mathcal{J}$) are compared among different models and different datasets including movie-watching fMRI (visual fMRI) and resting fMRI. The bars show the averaged results for 10 repeating experiments, and the error bar represents the standard deviation. (c) The attributions of top 8 significant and insignificant areas under the AAL Atlas are presented. The error bar represents the standard errors. (d) Attribution maps for movie-watching(visual) fMRI data are displayed. The maps show the left hemisphere from the view of the left, left hemisphere from the view of the right, right hemisphere from the view of the right, and right hemisphere from the view of the left. Also, the right column reflects a detailed map for visual areas. The upper is left visual areas and the bottom is right visual areas. The colors represent the normalized absolute values of the integrated gradient.  }
	\label{fig:brain_si_preprocessed}
\end{figure}

\section{The Structure and parameters of neural networks.}
\label{sec:parameter table}

Next, we list the neural network parameters used in our experimental study, as shown in Table \ref{table:parameters}. The parameter $\lambda$ in the table is derived from Equation \ref{new optimization} and serves as the weight coefficient for the forward prediction error. $L$ is an assumption about the range of data distribution, which represents the length of a hypercube with side length $2L$ that is used to calculate the probability value for a uniform distribution. It is not difficult to see that different types of neural networks have been used to handle different types of complex systems. Both the Boids and fMRI experiments employ the Multistory NIS+ framework, which traverses different macroscopic scales and eventually selects an optimal scale (the dimension of input or output in the Inverse Dynamics Learner table).

\begin{table}[]
  \centering
  \caption{Parameter table for all experiments conducted using neural networks.}
  \scalebox{0.7}{
  \begin{tabular}{|c|c|c|}
    \hline
    Experiment & Module & Parameter\\ \hline
    \centering\multirow{3}{*}{SIR} & Encoder(Decoder) &\makecell[l]{Three RealNVP modules \\
                                            Input(Onput) dimensions: 4 \\ 
                                            Macro-state dimensions: 2} \\ \cline{2-3}
                         &Dynamics Learner &\makecell[l]{Input(Onput) dimensions: 2 \\
                                            Hidden Layers: 2 \\
                                            Hidden Units: 64 \\
                                            Activation Function: LeakyReLu \\ 
                                            Total Epochs:30000 \\
                                            Batch Size:700 \\
                                            L: 1}\\ \cline{2-3}
                        & Inverse Dynamics Learner &\makecell[l]{$\lambda$: 0.33 \\
                                                    Input(Onput) dimensions: 2 \\
                                                    Hidden Layers: 2 \\
                                                    Hidden Units: 64 \\
                                                    Activation Function: LeakyReLu\\
                                                    Stage II Epochs:27000}\\ \hline
    \centering\multirow{3}{*}{Boids} & Encoder(Decoder) &\makecell[l]{Three RealNVP modules\\ 
                                            Input(Onput) dimensions: 64\\ 
                                            Macro-state dimensions: 32,16,8  }\\ \cline{2-3}
                         &Dynamics Learner &\makecell[l]{Input(Onput) dimensions: 64,32,16,8  \\
                                            Hidden Layers: 2 \\
                                            Hidden Units: 32 \\
                                            Activation Function: LeakyReLu \\
                                            Total Epochs:800,000\\
                                            Batch Size:4\\
                                            L: 100} \\ \cline{2-3}
                        & Inverse Dynamics Learner &\makecell[l]{$\lambda$: 1 \\
                                                    Input(Onput) dimensions: 8 \\
                                                    Hidden Layers: 2 \\
                                                    Hidden Units: 32 \\
                                                    Activation Function: LeakyReLu\\
                                                    Stage II Epochs:400,000}\\ \hline
    \centering\multirow{3}{*}{Game of life} & Encoder(Decoder) &\makecell[l]{Three RealNVP modules\\ 
                                            Input(Onput)  dimensions: $18\times18\times2$\\ 
                                            Macro-state dimensions: $6\times6\times1$ }\\ \cline{2-3}
                         &Dynamics Learner &\makecell[l]{Input(Onput) dimensions: 36 \\
                                            Hidden Layers: 2 \\
                                            Hidden Units: 64 \\
                                            Activation Function: LeakyReLu \\ 
                                            Total Epochs: 300,000\\
                                            Batch Size: 50 \\
                                            L: 100}\\ \cline{2-3}
                        & Inverse Dynamics Learner &\makecell[l]{$\lambda$: 1 \\
                                                    Input(Onput) dimensions: $18\times18\times2$ \\
                                                    Hidden Layers: 2 \\
                                                    Hidden Units: 64 \\
                                                    Activation Function: LeakyReLu\\
                                                    Stage II Epochs: 200,000}\\ \hline
    \centering\multirow{3}{*}{fMRI} & Encoder(Decoder) &\makecell[l]{Six RealNVP modules\\ 
                                            Input(Onput) dimensions: 100\\ 
                                            Macro-state dimensions: 52,27,14,7,3,1 }\\ \cline{2-3}
                         &Dynamics Learner &\makecell[l]{Input(Onput) dimensions: 100,52,27,14,7,3,1 \\
                                            Hidden Layers: 5 \\
                                            Hidden Units: 256 \\
                                            Activation Function: LeakyReLu and a final tanh \\
                                            Total Epochs: 60000\\
                                            Batch Size: 100\\
                                            L: 1} \\ \cline{2-3}
                        & Inverse Dynamics Learner &\makecell[l]{$\lambda$: 1 \\
                                                    Input(Onput) dimensions: 1 \\
                                                    Hidden Layers: 5 \\
                                                    Hidden Units: 256 \\
                                                    Activation Function: LeakyReLu
                                                    Stage II Epochs: 10000}\\ \hline

  \end{tabular}}
  \label{table:parameters}
\end{table}
\clearpage
\bibliographystyle{IEEEtran}
\bibliography{main} 
\end{document}